\documentclass[11pt,fleqn]{article}
\usepackage[utf8]{inputenc}
\usepackage{fullpage}

\usepackage{amsthm, amssymb}
\usepackage[dvipsnames]{xcolor}
\usepackage[hidelinks]{hyperref}
\usepackage{todonotes}
\usepackage{qcircuit}
\usepackage{tikz}
\usepackage{tikzsymbols}
\usetikzlibrary{calc}

\usepackage{pgfplots}
\pgfplotsset{compat=1.16}
\usepackage{svg}
\usepackage{mathtools}
\usepackage{titling}
\usepackage{textcomp}
\usepackage{mdframed}
\usepackage{titletoc}
\titlecontents{section}
  [1.5em]                                 %
  {}                             %
  {\contentslabel{1.5em}}                 %
  {}                                      %
  {\titlerule*[0.5pc]{.}\contentspage}    %

\usepackage{libertine}
\usepackage{libertinust1math}

\usepackage{tikz}
\usetikzlibrary{shapes,backgrounds}

\usepackage{caption}
\usepackage{complexity}

\newif\ifstoc
\stocfalse

\newif\ifarxiv
\arxivtrue

\newif\ifsubmission
\submissionfalse
\ifstoc\submissiontrue\fi
\ifarxiv\submissiontrue\fi

\ifstoc
\else
    \ifarxiv
    \else
    \usepackage{setspace}
    \fi
\fi

\usepackage{graphicx, subfigure}
\usepackage{enumitem}
\usepackage{xparse}
\usepackage{physics}
\usepackage{color}
\usepackage{comment}

\newcommand{\x}{\mathbf{x}}
\renewcommand{\i}{\mathbf{i}}
\newcommand{\y}{\mathbf{y}}
\newcommand{\e}{\mathrm{e}}

\newcommand{\bb}{\mathbf{b}}

\newcommand{\CC}{\mathbb{C}}
\newcommand{\FF}{\mathbb{F}}
\newcommand{\NN}{\mathbb{N}}

\newcommand{\RR}{\mathbb{R}}
\newcommand{\ZZ}{\mathbb{Z}}

\newcommand{\Aa}{\mathcal{A}}
\newcommand{\Bb}{\mathcal{B}}

\newcommand{\Hh}{\mathcal{H}}
\newcommand{\Ll}{\mathcal{L}}

\newcommand{\Oo}{\mathcal{O}}

\renewcommand{\H}{\mathbf{H}}
\renewcommand{\G}{\mathbf{G}}

\newcommand{\Exp}{\mathop{\mathbb{E}\hspace{0.13em}}}

\newcommand{\inv}[1]{#1^{-1}}

\newcommand{\id}{\mathrm{id}}
\newcommand{\negl}{\mathrm{negl}}
\renewcommand{\poly}{\mathrm{poly}}

\newcommand{\diag}{\mathrm{diag}}

\renewcommand{\Pr}{\mathop{\mathbb{P}\hspace{0.05em}}}

\newcommand{\eps}{\varepsilon}
\renewcommand{\epsilon}{\varepsilon}

\newcommand{\st}{\text{ s.t. }}
\newcommand{\half}{\frac{1}{2}}

\newcommand{\defeq}{\mathrel{\overset{\makebox[0pt]{\mbox{\normalfont\tiny\sffamily def}}}{=}}}

\usepackage{centernot}

\usepackage{complexity}

\renewcommand{\exp}{\mathsf{exp}}

\newcommand{\bits}{\{0,1\}}

\newcommand{\fn}[3]{#1: #2 \rightarrow #3}

\newcommand{\qubit}{\CC^{2}}
\newcommand{\qubits}[1]{(\qubit)^{\otimes #1}}

\renewcommand{\hat}{\widehat}
\renewcommand{\tilde}{\widetilde}

\newcommand{\yes}{\mathsf{yes}}
\newcommand{\no}{\mathsf{no}}

\newcommand{\Eqref}[1]{\text{Eq.~}(\ref{#1})}
\renewcommand{\eqref}[1]{\text{eq.~}(\ref{#1})}

\renewcommand{\exp}{\mathrm{exp}}

\usepackage{todonotes}

\newcounter{commentcount}

\newenvironment{mathinlay}[1][]%
  {
  \begin{mdframed}
  [leftmargin=1em,rightmargin=1em,#1]\begingroup\small}%
  {\endgroup
  \end{mdframed}
  }

\newcommand{\reg}[1]{\mathsf{#1}} 
\allowdisplaybreaks 

\usepackage{tikz}
\usetikzlibrary{shapes,backgrounds}

\usepackage{cleveref}
\crefformat{section}{Section~\textsection#2#1#3} %
\crefformat{subsection}{Subsection~\textsection#2#1#3}
\crefformat{subsubsection}{Subsubsection~\textsection#2#1#3}

\newtheorem{theorem}{Theorem}[section]
\newtheorem*{theorem*}{Theorem}
\newtheorem{definition}[theorem]{Definition}
    \crefname{definition}{definition}{definitions}
\newtheorem{lemma}[theorem]{Lemma}
    \crefname{lemma}{lemma}{lemmas}

    \crefname{conjecture}{conjecture}{conjectures}
\newtheorem{corollary}[theorem]{Corollary}
    \crefname{corollary}{corollary}{corollaries}
\newtheorem*{corollary*}{Corollary}
\newtheorem{fact}[theorem]{Fact}
    \crefname{fact}{fact}{facts}
\newtheorem{remark}[theorem]{Remark}
    \crefname{remark}{remark}{remarks}
\newtheorem{claim}[theorem]{Claim}
    \crefname{claim}{claim}{claims}
\newtheorem*{claim*}{Claim}  %

\newenvironment{xalign}[1][]{
    \subequations
    \label{#1}
    \align
}
{
    \endalign
    \endsubequations
}

\makeatletter
\let\oldendproof\endproof
\renewcommand{\endproof}{\oldendproof\par\noindent\ignorespaces}
\makeatother

\makeatletter
\renewcommand{\maketitle}{\bgroup\setlength{\parindent}{0pt}
\begin{flushleft}
  \LARGE {\@title} %
  \vspace{4.5mm}
  
  \normalsize \@author \linebreak
\end{flushleft}\egroup
}
\makeatother

\title{Separating $\mathsf{QMA}$ from $\mathsf{QCMA}$ with a classical oracle}

\ifstoc
\author{Authorship omitted for STOC submission}
\else
\usepackage{authblk}
\author[1]{John Bostanci}
\author[2,3]{Jonas Haferkamp}
\author[4]{Chinmay Nirkhe}
\author[5,6]{Mark Zhandry}

\affil[1]{\small Columbia University, \textit{New York, NY, USA}}
\affil[2]{\small Saarland University, \textit{Saarbr\"ucken, Germany}}
\affil[3]{\small Harvard University, \textit{Cambridge, Mass., USA}}
\affil[4]{\small University of Washington, \textit{Seattle, Wash., USA}}
\affil[5]{\small Stanford University, \textit{Stanford, Calif., USA}}
\affil[6]{\small NTT Research, Inc, \textit{Sunnyvale, Calif., USA}}

\fi

\date{}

\newcommand{\vac}{\mathrm{vac}}

\begin{document}
\maketitle

\begin{abstract}
    We construct a classical oracle proving that, in a relativized setting, the set of languages decidable by an efficient quantum verifier with a quantum witness ($\QMA$) is strictly bigger than those decidable with access only to a classical witness ($\QCMA$). The separating classical oracle we construct is for a decision problem we coin \emph{spectral Forrelation} -- the oracle describes two subsets of the boolean hypercube, and the computational task is to decide if there exists a quantum state whose standard basis measurement distribution is well supported on one subset while its Fourier basis measurement distribution is well supported on the other subset. This is equivalent to estimating the spectral norm of a ``Forrelation'' matrix between two sets that are accessible through membership queries. \\

    Our lower bound derives from a simple observation that a query algorithm with a classical witness can be run multiple times to generate many samples from a distribution, while a quantum witness is a ``use once'' object. This observation allows us to reduce proving a $\QCMA$ lower bound to proving a sampling hardness result which does not simultaneously prove a $\QMA$ lower bound. To prove said sampling hardness result for $\QCMA$, we observe that quantum access to the oracle can be compressed by expressing the problem in terms of bosons -- a novel ``second quantization'' perspective on compressed oracle techniques, which may be of independent interest. Using this compressed perspective on the sampling problem, we prove the sampling hardness result, completing the proof.
\end{abstract}

\pagebreak

\setcounter{tocdepth}{1}
\tableofcontents
\newpage

\part{Introduction}
\label{part:introduction}
\label{sec:introduction}

The problem of finding a standard oracle separation between $\QMA$ (the class of problems that can be verified with a quantum computer and quantum witness) and $\QCMA$ (the class of problems that can be verified with a quantum computer and \emph{classical} witness) is a central open problem in the field of quantum query complexity and is the first question mentioned in Aaronson's list of open query complexity problems~\cite{aaronson2021open}.  The goal of separating $\QMA$ from $\QCMA$ is, in some sense, to understand the following question:
\begin{quote}
    \begin{center}
    \textit{
        Do quantum witnesses offer more power than classical witnesses?
        }
    \end{center}
\end{quote}

This question was first posed by Aharonov and Naveh~\cite{aharonov2002quantum}, and partially answered by the work of Aaronson and Kuperberg~\cite{aaronson2007quantum}, which provided a \emph{quantum} oracle separation between the two classes. Because $\P \subseteq \QCMA \subseteq \QMA \subseteq \PSPACE$, any unconditional separation of the two complexity classes would imply $\P \neq \PSPACE$ and seems unlikely without significantly stronger new tools.

However, the quantum oracle separation of~\cite{aaronson2007quantum} could be considered an unsatisfying oracle separation between $\QMA$ and $\QCMA$, as it avoids answering deeper questions about the power of quantum witnesses over classical ones. The separation constructs a unitary property testing problem for which a verifier must exactly know a Haar random state to solve the problem, essentially forcing that any classical witness for the problem must provide a full classical description of the Haar random state\footnote{The notion of ``knowing'' the state is meant to be more intuitive rather than formal due to the result of \cite{irani2021quantum}, where it is shown that having an oracle that solves the separating problem of \cite{aaronson2007quantum} does not allow one to synthesize the state specified by the oracle.}. The question of whether $\QMA$ equals $\QCMA$ necessitates more than a lower bound on the classical description complexity of a random quantum state; it requires proving that even the single bit identified by the $\QMA$ decision problem cannot be verified by a classical witness. It could be that if $\mathsf{QMA}$ does equal $\mathsf{QCMA}$, the $\mathsf{QCMA}$ verifier for a $\mathsf{QMA}$-complete problem would not receive a ``verbatim'' description of the quantum witness (such as a circuit preparing the quantum witness), but rather some other classical information derived (potentially inefficiently) from the instance or the original quantum witness, which can be used to answer the decision problem but not necessarily to reproduce the quantum witness.

By forcing the oracle to be classical, one hopes to identify more meaningful reasons why useful properties -- beyond the full description of a quantum witness -- should be hard to write down in a classical witness.  A final reason for studying the problem of a classical oracle separation is that the question lies in the rich field of quantum query complexity and has been linked to open questions in quantum cryptography, such as the existence of quantum money~\cite{lutomirski2011component,nehoran2024computational} and pseudorandomness against quantum adversaries~\cite{liu2024qma}.  One hopes that finding a classical oracle separation between $\mathsf{QMA}$ and $\mathsf{QCMA}$ will provide an improved understanding of the complexity theory underlying these cryptographic questions.

\section{Proof overview}
\begin{theorem}[Classical oracle separation between $\QMA$ and $\QCMA$]
\label{thm:main}
There exists a pair of oracles $\fn{S,U}{\bits^*}{\bits}$ and a language $\Ll^{S,U}$ such that $\Ll^{S,U} \in \QMA^{S,U} \setminus \QCMA^{S,U}$ -- i.e., there exists a polynomial-time quantum verifier taking as input quantum witnesses as input for deciding membership in $\Ll^{S,U}$, while no polynomial-time quantum verifier taking as input classical witnesses can decide membership in $\Ll^{S,U}$.
\end{theorem}
\noindent This is the main result of this work. While we express the statement in terms of two oracles, as our construction is most natural to describe in terms of a quantum algorithm comparing two oracles, it is easy enough to convert the statement into that of a single oracle.

\subsection{The challenges of an oracle separation}

To prove this classical oracle separation between $\QMA$ and $\QCMA$, it suffices to construct a property testing problem about classical oracles that can be decided with quantum witnesses but cannot be decided with classical witnesses. The transformation from the property testing problem to an oracle separation between complexity classes is a diagonalization argument that takes some work but is fairly standard.

The central challenge in separating $\QMA$ from $\QCMA$ is illustrated in the following thought experiment: Consider measuring a quantum witness $\ket{\psi}$ for a $\QMA$ problem to generate a classical witness. If the resulting classical object is accepted by the $\QMA$ verifier, then the problem must have been in $\QCMA$.
Therefore, for a problem to be in $\QMA \setminus \QCMA$, it follows that for any quantum witness $\ket{\psi}$ and any computational basis $\ket{x}$, $\abs{\braket{x}{\psi}} \leq \negl(n)$, as otherwise $x$ would serve as a good classical witness for the problem\footnote{This can be extended to show that the measurement distribution for any quantum witness for the problem in \emph{every computationally efficient basis} cannot have more than $\negl(n)$ support on any basis vector.}. Here, the phrase $\negl(n)$ is used to refer to any function that is smaller than every function $1/\poly(n)$. Therefore, assuming $\QMA \neq \QCMA$, the $\QMA$ verification procedure must certify (at least implicitly) that the quantum witness is a superposition of a super-polynomial number of basis vectors. 

Most techniques for verifying such a complex superposition require highly structured oracles, which presents an inherent challenge for proving a $\QCMA$ lower bound. This is because most techniques we currently know of for proving lower bounds against quantum query algorithms are not amenable to highly structured oracles. Additionally, a classical witness would naturally be treated as a form of advice about the oracle, and most quantum query complexity techniques are not very good at distinguishing between quantum advice and classical advice.  If a typical technique succeeded in proving that a language is outside of $\mathsf{QCMA}$, but it can not distinguish between quantum and classical advice, it would likely show that the language is outside of $\mathsf{QMA}$ as well, failing to give a separation.  

Our starting point is the following simple observation first made by \cite{nehoran2024computational}  and later by Zhandry~\cite{zhandry2024toward}, which separates the functionality of quantum and classical witnesses: if a highly-entangled superposition of computational basis states $\ket{\phi}$ is efficiently generated from a classical witness $w$, then we can generate a polynomial number of measurement samples from $\ket{\phi}$, whereas if $\ket{\phi}$ is generated from a quantum witness $\ket{\psi}$, then we expect that only one measurement sample can be extracted. This is because the classical witness can be efficiently copied while the quantum witness cannot be. Somewhat paradoxically, this demonstrates a particular way in which a classical witness is \emph{more} powerful than a quantum witness. It is precisely this boost in power that we show is too good to be true, thereby proving the impossibility of an efficient verification of a classical witness. This observation, however, does not yet indicate how to actually design the oracles used in the separation.

\subsection{Spectral Forrelation}
\label{sec:spectral-forrelation-intro}

Zhandry~\cite{zhandry2024toward} gives a candidate oracle for a separation and an initial analysis, but ultimately was unable to prove the separation. Our oracles are inspired by Zhandry's, though our approach to analyzing them is very different. 
Zhandry~\cite{zhandry2024toward} defines a variant of the problem that we will call \emph{spectral Forrelation}\footnote{The name comes from the notion of Forrelation, defined by Aaronson~\cite{aaronson2010bqp}. He defined the ``Forrelation'' between two functions $f,g:\bits^n \to \{\pm1\}$ as the quantity $\bra{+}^{\otimes n} \diag(f) \cdot H^{\otimes n} \cdot \diag(g) \ket{+}^{\otimes n}.$
When $f,g$ are described by oracles, Aaronson~\cite{aaronson2010bqp,aaronson2015forrelation} gives a simple $\BQP$ algorithm for deciding if the absolute value of the Forrelations is $\geq 3/5$ or $\leq 1/100$: simply prepare $|+\rangle^{\otimes n}$, apply the reflection $\diag(g)$, perform the Hadamard transform $H^{\otimes n}$, and apply the reflection $\diag(f)$, and measure to see if the output is $|+\rangle^{\otimes n}$. On the other hand, Aaronson shows that any algorithm making randomized classical queries to $S$ and $U$ needs an exponential number of queries to decide Forrelation.}, which is plausibly in $\QMA\setminus\QCMA$: we say a pair of subsets $S,U\subseteq \{0,1\}^n$ are $\alpha$-spectrally Forrelated\footnotetext{Note that Zhandry~\cite{zhandry2024toward} actually used different notation and used the QFT mod $N=2^n$ in place of the Hadamard transform, but the underlying idea of two projectors in anti-commuting bases remains the same.} if
\begin{xalign}
    \alpha &= \norm{\Pi_U\cdot H^{\otimes n}\cdot \Pi_S}_\mathrm{op}^2 = \max_{\norm{\ket{\psi}} = 1} \norm{\Pi_U\cdot H^{\otimes n}\cdot \Pi_S \ket{\psi}}^2 \\
    &\text{where}~\Pi_U = \sum_{x \in U} \ketbra{x},~ \Pi_S = \sum_{x \in S} \ketbra{x}\,.
\end{xalign}

\begin{figure}[h]
    \centering
    \includesvg[width=0.5\linewidth]{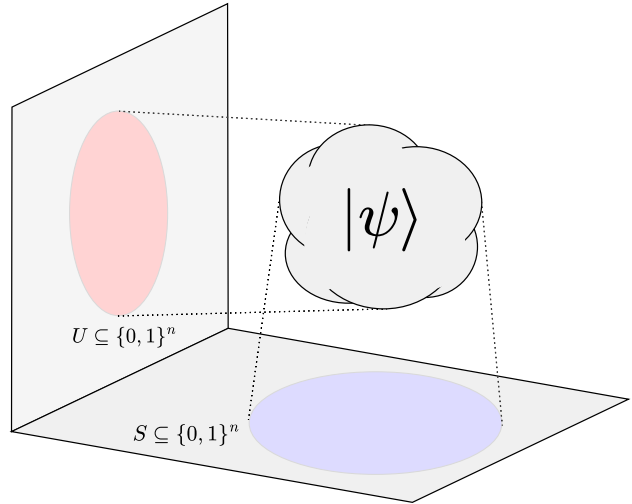}
    \caption{\textit{A cartoon of spectral Forrelation.
    The subsets $S$ and $U$ occupy support in the standard/position and Hadamard/momentum bases, respectively. The sets $S$ and $U$ are $\ge \alpha$-spectrally Forrelated if there exists a state $\ket{\psi}$ such that $\norm{\Pi_U \cdot H^{\otimes n} \cdot \Pi_S \ket{\psi}}^2 \ge \alpha$. Equivalently, there exists a state $\ket{\psi}$ for which the induced classical distributions by measuring in the standard and Hadamard bases are well supported on $S$ and $U$, respectively.}}
\end{figure}

Spectral Forrelation can also be expressed as a property of functions $S,U: \bits^n \to \bits$, where the set identified with a function is the pre-image of $1$. A decision version of the spectral Forrelation problem can be defined for two parameters $\alpha > \beta$. The task is to decide if the pair $(S,U)$ are at least $\geq \alpha$-spectrally Forrelated (yes instance) or at most $\leq \beta$-spectrally Forrelated (no instance), promised that one of the cases holds.
For $\alpha - \beta \geq 1/\poly(n)$, yes instances of spectral Forrelation can be verified with a $n$-qubit quantum witness, namely the top singular vector $\ket{\psi}$ of $\Pi_U \cdot H^{\otimes n}\cdot \Pi_S$ (similar to plain Forrelation~\cite{aaronson2010bqp}, but using $\ket{\psi}$ instead of $|+\rangle^{\otimes n}$). On the other hand, it can be shown that without $|\psi\rangle$, Spectral Forrelation is hard even for quantum query algorithms. This puts Spectral Forrelation in $\QMA\setminus\BQP$. The question remains: is spectral Forrelation outside $\QCMA$?

\subsection{$\QCMA$ algorithms imply good one oracle samplers}
Let us assume, for contradiction, that there exists a $\QCMA$ query algorithm $\mathcal{A}$ for spectral Forrelation requiring a $q = q(n)$ sized classical witness and making $t = t(n)$ oracle queries, where $q(n), t(n)$ are polynomials. A simple transformation can convert $\mathcal{A}$ into an algorithm which separately makes $\le t$ queries to $S$ and $\le t$ queries to $U$. 

To prove the impossibility of a $\QCMA$ algorithm, we restrict our attention to a subset of all spectrally Forrelated pairs (i.e., yes instances): First, we require pairs $(S,U)$ with the property that the oracle $S$ is \emph{sparse}, consisting of approximately $\ell$ non-zero entries where $\ell = 2^{cn}$ for some small constant $c$.  Secondly, we require that $(\Delta, U)$ is a no instance for all subsets $\Delta \subseteq S$ with $\abs{\Delta} \ll \ell$.  We describe pairs $(S, U)$ satisfying this second property as being \emph{strong}, and we will sketch in~\Cref{sec:intro-sampler-too-good-to-be-true} a technique for generating a distribution over sparse and strong yes instances.

Let $w = w(S,U)$, be a witness certifying that $(S,U)$ is a yes instance. In other words, $\mathcal{A}^{(S,U)}(w)$ accepts with high probability\footnote{Note that the perfect completeness of $\QCMA$ algorithms with access to classical oracles was proven by Jordan \emph{et. al.}~\cite{qcma-perfect-completeness} but we do not need that result here.}. Since we take $(S, U)$ to be a strong yes instance and thus have that for any small subset $\Delta$, $(\Delta, U)$ is a no instance of spectral Forrelation, it follows that $\mathcal{A}^{(\Delta, U)}(w)$ accepts with low probability. 

Let us express the algorithm $\mathcal{A}^{(S,U)}(w)$ as a sequence of unitaries interspersed with queries to the $S$ oracle. Assuming that the queries to the $U$ oracle are included in the unitaries $\{V_j\}$, the state of the algorithm $\mathcal{A}$ immediately before its final measurement is given by:
\begin{equation}
    \text{pre-measurement state of}~\mathcal{A}^{(S,U)}(w) = V_t \mathcal{O}_S V_{t-1} \mathcal{O}_S \ldots \mathcal{O}_S V_0 \ket{w, 0}.
\end{equation}
Where we take $\mathcal{O}_S$ to be the phase oracle for the function that checks membership in $S$.  Then the state $V_j V_{j-1} \ldots V_0 \ket{w, 0}$ is equivalent to running the algorithm $\mathcal{A}^{(\emptyset, U)}(w)$ until immediately prior to the $j$-th oracle query. Since $(S,U)$ is a yes instance, $(\emptyset, U)$ is a no instance, and $\mathcal{A}$ must distinguish the two cases, it must be that $\mathcal{A}$ is actually putting significant query weight on points where $S$ and $\emptyset$ differ, namely points in $S$. As a consequence, it follows from a hybrid argument similar to that of Bennett~\emph{et. al.}~\cite{BBBV}, that picking a uniformly random index $j \in [t]$ and measuring the query register of $V_j V_{j-1} \ldots V_0 \ket{w, 0}$ will yield a sample $x_1$ from $S$ with probability $\Omega(t^{-2})$. In other words, there exists an algorithm which only queries $U$ and takes a witness $w$, but produces a sample $x_1$ from $S$ with noticeable probability. This is not surprising, as information about $x_1$ could be encoded in $w$. 

\begin{figure}[h]
\centering
\begin{tikzpicture}[>=latex,scale=1.05,rounded corners=2pt]

\def\cubesize{0.6}
\def\depth{0.15}

\newcommand{\cubeU}[1]{
  \coordinate (A) at (#1,0);
  \coordinate (B) at ($(A)+(\cubesize,0)$);
  \coordinate (C) at ($(A)+(0,\cubesize)$);
  \coordinate (D) at ($(B)+(0,\cubesize)$);
  \coordinate (E) at ($(A)+(\depth,\depth)$);
  \coordinate (F) at ($(B)+(\depth,\depth)$);
  \coordinate (G) at ($(C)+(\depth,\depth)$);
  \coordinate (H) at ($(D)+(\depth,\depth)$);

  \coordinate (M) at ($(A)+(0.375,0.375)$);

  \draw[thick,fill=white] (A)--(B)--(D)--(C)--cycle;
  \draw[thick,fill=gray!10] (C)--(D)--(H)--(G)--cycle;
  \draw[thick,fill=gray!5] (B)--(D)--(H)--(F)--cycle;
  \draw[thick] (A)--(E)--(F)--(B);
  \draw[thick] (E)--(G);
  \draw[thick] (G)--(H);

  \node at (M) {\footnotesize \textit{U}};
}
\fill[blue!5] (-4.8,-2.5) rectangle (-0.1,2.5);
\node at (-2.45,2.7) {\small {Quantum witness}};

\node[draw,rounded corners,fill=white,minimum width=0.6cm] (psi) at (-3.5,0) {$|\psi\rangle$};
\begin{scope}[shift={(-2.55,-0.33)}]\cubeU{0}\end{scope}
\draw[->,thick] (psi.east) -- (-2.6,0);
\draw[->,thick] (-1.7,0) -- (-0.8,0);
\node[above left=-1pt and -2pt] at (-0.9,0) { $x$};
\fill[green!5] (0.3,-2.5) rectangle (5.2,2.5);
\node at (2.75,2.7) {\small {Classical witness}};

\node[draw,rounded corners,fill=white,minimum width=0.6cm] (w) at (0.8,1.62) {$w$};

\def\xC{2.3}
\def\xOut{3.9}
\def\yA{1.3}
\def\yB{-0.1}
\def\yC{-1.5}

\begin{scope}[shift={(\xC,\yA)}]\cubeU{0}\end{scope}
\begin{scope}[shift={(\xC,\yB)}]\cubeU{0}\end{scope}
\begin{scope}[shift={(\xC,\yC)}]\cubeU{0}\end{scope}

\foreach \yy in {\yA,\yB,\yC}{
  \draw[->,thick] (w.east) -- (1.6,\yy+0.32) -- (\xC-0.15,\yy+0.32);
}

\foreach \y/\lbl in {\yA/x_1,\yB/x_2,\yC/x_3}{
  \draw[->,thick] (\xC+\cubesize+0.35,\y+0.32) -- (\xOut,\y+0.32);
  \node[above left=-1pt and -3pt] at (\xOut-0.05,\y+0.32) { $\lbl$};
}

\draw[->,thick]
  (\xC+\cubesize+0.35,\yA+0.32)
    .. controls (4.6,1.0) and (3.4,0.9)
    .. (\xC+0.25,\yB+0.75);
\draw[->,thick]
  (\xC+\cubesize+0.35,\yB+0.32)
    .. controls (4.6,-0.3) and (3.4,-0.4)
    .. (\xC+0.25,\yC+0.757);

\node at (\xC+0.3,-2.0) {$\vdots$};

\end{tikzpicture} \caption{\textit{Adaptive sampling from a classical witness.
A quantum witness $|\psi\rangle$ yields a single sample $x$ upon measurement with an algorithm accessing only $U$. Whereas, a classical witness $w$ can be reused across successive sampling rounds with an adaptive sampler accessing only $U$ and prior samples to generate multiple \emph{distinct} samples.}}
\label{fig:adaptive_sampling}
\end{figure}
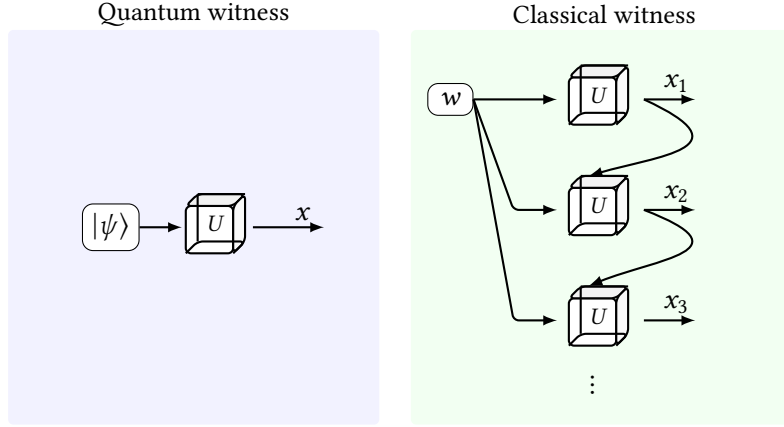

A more surprising fact is that a small modification of this sampler can be used to generate many unique samples\footnote{Note that, as written, there is no guarantee that repeating this experiment with the empty oracle would generate a different sample.}. 
Let us condition on the event that $x_1 \in S$. Now consider picking a uniformly random index $j \in [t]$ but this time measuring the query register of $V_j \mathcal{O}_{\Delta} V_{j-1} \mathcal{O}_{\Delta} \ldots \mathcal{O}_{\Delta} V_0 \ket{w, 0}$ where $\mathcal{O}_{\Delta} = \id - 2 \ketbra{x_1}{x_1}$. This is equivalent to running the algorithm $\mathcal{A}^{\{x_1\}, U}(w)$ until immediately prior to the $j$-th oracle query. Since $(\{x_1\}, U)$ is also a no instance, the hybrid argument between $(S,U)$ and $(\{x_1\}, U)$ yields that the generated sample $x_2$ satisfies
\begin{equation}
    \Pr \left[ x_2 \in S \setminus \{x_1\} ~\middle \vert x_1 \in S \right] \geq \Omega(t^{-2}).
\end{equation}
In other words, conditioned on $x_1 \in S$, this procedure generates a unique second point $x_2 \in S$.
In general, we can repeat this process $v$ times for $v \ll \ell$, each time generating a novel sample from $S$ conditioned on the previous samples being from $S$. Therefore, there exists a sampler that only queries $U$ and produces $v$ unique samples from $S$ with probability $\geq (\Omega(t^{-2}))^v$ when provided the correct witness $w$.  See \Cref{fig:adaptive_sampling} for a cartoon illustration of this iterated sampler. For $v \geq \Omega(q/n)$, where $q$ is the length of the witness, this yields unexpected consequences, as it means that the sampling algorithm cannot simply read the names of samples from the witness; it must also continue to find more samples from its access to $U$. We can make this explicit by constructing a new sampler which guesses the witness $w$ initially at a $2^{-q}$ cost in success probability, yielding the following theorem.

\begin{theorem}[Good samplers from $\QCMA$ algorithms (informal)] \label{thm:informal-removal-of-S}
Assume there exists a classical witness query algorithm $\mathcal{A}$ for spectral Forrelation requiring a $q = q(n)$ sized classical witness and making $t = t(n)$ oracle queries. Let $(S,U)$ be a \emph{strong} yes instance of spectral Forrelation. For all $v = v(n)$ polynomial in $n$, there exists a query algorithm $\mathsf{CumulativeSampler}$ such that $\mathsf{CumulativeSampler}^U$ makes no queries to $S$, $vt$ queries to $U$, and produces $v$ unique samples from $S$ with probability at least
    $\ge 2^{-q} \cdot \Omega(t)^{-2v}$.
\end{theorem}
Our primary goal then is to show that the sampler promised by~\Cref{thm:informal-removal-of-S} is too good to be true, implying that the assumed classical witness query algorithm cannot exist. A first intuition as to why this sampler should not exist is to consider a query algorithm $\widetilde{\mathsf{CumulativeSampler}}$ that only makes $vt$ queries to $S$ and produces $v$ unique samples from $S$. It would be natural to expect that access directly to $S$ (as opposed to some spectrally Forrelated $U$) should only be more useful for outputting samples from $S$. Yet, by the result of Hamoudi and Magniez~\cite{hamoudi2023quantum}, we know an upper bound on the success probability of $\widetilde{\mathsf{CumulativeSampler}}$ is at most $O(t^2 \ell/2^n)^v$. If we compare this upper bound to the supposed $\mathsf{CumulativeSampler}$ guaranteed by Theorem~\ref{thm:informal-removal-of-S}: for $t,\ell$ small relative to $2^n$ and $v\geq \Omega(q/n)$, we have $O(t^2 \ell/2^n)^v\ll2^{-q} \cdot \Omega(t)^{-2v}$. Therefore, if the sampler $\mathsf{CumulativeSampler}$ from~\Cref{thm:informal-removal-of-S} were to actually exist, it would mean that quantum access to any set $U$ spectrally Forrelated with $S$ yields a significantly better sampler than quantum access to the set $S$ itself! This is the first indication that we should be able to prove a $\QMA$ versus $\QCMA$ oracle separation using this insight about samplers. Of course, this is just intuition; we will actually need to prove that query access to $U$ does not help too much in producing points in $S$.

Let us emphasize that we cannot derive an analogous theorem from a quantum witness query algorithm. While we can run the query algorithm with a random witness by using $\id/2^q$ as a proxy for the quantum witness, we do not know how to construct a sampling algorithm that continuously produces samples. The ${\mathsf{CumulativeSampler}}$ in~\Cref{thm:informal-removal-of-S} relies on being able to copy the \emph{same} witness from one iteration to the next, which is impossible for quantum witnesses. Therefore, proving that the result of~\Cref{thm:informal-removal-of-S} is too good to be true does not yield a lower bound for quantum witness query algorithms. This is the fundamental reason that our $\QCMA$-lower bounding technique can separate it from $\QMA$.

\begin{remark}Note that Zhandry~\cite{zhandry2024toward} utilized a similar approach of turning a $\QCMA$ verifier into a sampler, except instead of removing oracle queries to $S$, his approach is to remove oracle queries to $U$ while keeping queries to $S$. He removes queries to $U$ under a general conjecture about the quantum indistinguishability of two oracle distributions whose $k$-wise marginals are close in relative error. Unfortunately, this general conjecture turns out to be false: (plain) Forrelation gives oracles whose $k$-wise marginals are close in relative error to random oracles (as proved by Aaronson~\cite{aaronson2015forrelation}), but the Forrelation algorithm gives a distinguisher\footnote{We thank Uma Girish for pointing this out to us in the initial phases of this work.}. This refutes the general conjecture of Zhandry, but not necessarily the concrete application to removing the oracle $U$. Nevertheless, we opt for a different approach which sidesteps this issue by removing queries to $S$ before analyzing $U$.
\end{remark}

\subsection{A family of strong yes instances}
\label{sec:intro-sampler-too-good-to-be-true}

One method of proving that~\Cref{thm:informal-removal-of-S} is too good to be true is to construct a distribution over strong yes instances and to argue that no sampling algorithm could succeed with the probability guaranteed by~\Cref{thm:informal-removal-of-S} with respect to this distribution. The introduction of strong yes instances is fundamental for two reasons. The first is type-checking; in order to argue about the behavior of the sampler, we need both yes and no instances of the original problem. By considering strong yes instances, we are implicitly considering no instances. Second, being a strong yes instance is a bare-minimum requirement for separating $\QMA$ from $\QCMA$, as yes instances that are not strong have short classical witnesses: suppose $(\Delta, U)$ is a yes-instance for some $\Delta \subset S$ such that $\abs{\Delta} = \poly(n)$, and consider the top eigenvector $\ket{\Psi_\Delta}$ of $\Pi_\Delta \cdot H^{\otimes n} \cdot \Pi_U \cdot H^{\otimes n} \cdot \Pi_\Delta$. Since $\Delta\subset S$, we also have that $\bra{\Psi_\Delta}\Pi_S \cdot H^{\otimes n}\cdot \Pi_U\cdot H^{\otimes n}\cdot \Pi_S\ket{\Psi_\Delta}$ is large. In particular, this means that $\ket{\Psi_\Delta}$ is also a good witness for the instance $(S,U)$. However, $\ket{\Psi_\Delta}$ can be classically described using $O(\abs{\Delta} \cdot \poly(n))$ bits (since its support is limited to $\Delta$, and therefore there are only $\abs{\Delta}$ many amplitudes required to describe it). Therefore, any hope of proving a separation between $\QMA$ and $\QCMA$ relies on studying the behavior of strong yes instances.

We now describe how to sample from a distribution over strong yes instances. Recall that we pick $\ell = 2^{cn}$ for some small constant $c$.  We first construct a (multi)set $S = \{s_1, \ldots, s_\ell\}$ by uniformly randomly sampling $\ell$ elements of $\bits^n$ with replacement. Observe that with all but $1 - \Omega(\ell^2/2^n)$ probability, the elements of $S$ will be distinct. We then construct a distribution over sets $U$ which, with high probability, is spectrally Forrelated with $S$, with $\ket S$ -- the uniform superposition over $S$ -- serving as the witness state. Concretely, we first compute the terms for $y \in \bits^n$,
\begin{equation} \label{eq:def1-of-gamma}
    \gamma_y^{(S)} \defeq \frac{1}{\ell} \sum_{i,j \in [\ell]} (-1)^{y \cdot (s_i \oplus s_j)}.
\end{equation}
Observe that $\gamma_y^{(S)}$ equals $2^n$ times the square of the amplitude $H^{\otimes n}\ket{S}$ places on $y$. We construct the set $U$ by adding each $y \in \bits^n $ to $U$ with independent probability $1 - \frac{1}{2}\e^{-\kappa\gamma_y^{(S)}}$. Here $\kappa$ is a small constant, say $1/10$. 

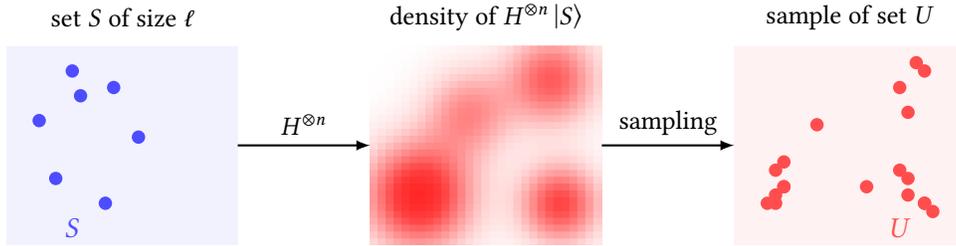
\begin{figure}[h]
 \centering

\begin{tikzpicture}[>=latex,scale=1.1]

\fill[blue!5] (-2,-1.2) rectangle (0.8,1.2);
\node at (-0.6,1.55) {\small set $S$ of size $\ell$};

\foreach \x/\y in {-1.6/0.3,-1.1/0.6,-0.8/-0.7,-0.4/0.1,-1.4/-0.4,
                   -1.2/0.9,-0.7/0.7} {
  \fill[blue!70] (\x,\y) circle (2.3pt);
}
\node[blue!70] at (-1.2,-1.0) {$S$};

\node at (3.8,1.55) {\small density of $H^{\otimes n}\ket{S}$};

\pgfplotsset{
  colormap={whitered}{
    color(0cm)=(white)
    color(0.8cm)=(red!15)
    color(1.6cm)=(red!35)
    color(2.4cm)=(red!55)
    color(3.2cm)=(red!75)
    color(4.0cm)=(red!85)
  }
}

\begin{scope}[shift={(2.4,-1.2)}]
\begin{axis}[
    hide axis,
    view={0}{90},
    width=2.8cm,
    height=2.4cm,
    scale only axis,
    colormap name=whitered,
    colorbar=false,
    samples=30, samples y=30,
    domain=0:2.8,
    y domain=0:2.4,
    zmin=0, zmax=1.8
]
\addplot3 [surf, shader=flat]
  {exp(-2.5*((x-0.6)^2+(y-0.6)^2))
   +0.7*exp(-3.5*((x-2.2)^2+(y-2.0)^2))
   +0.8*exp(-4.5*((x-2.3)^2+(y-0.5)^2))
   +0.5*exp(-3.8*((x-1.2)^2+(y-1.6)^2))};
\end{axis}
\end{scope}

\fill[red!5] (6.8,-1.2) rectangle (9.6,1.2);
\node at (8.2,1.55) {\small sample of set $U$};

\foreach \x/\y in {7.3/-0.6,7.3/-0.3,7.2/-0.7,7.4/-0.5,7.4/-0.5,7.3/-0.7,7.4/-0.2,7.8/0.25,9.0/1.0,9.1/0.9,8.8/0.7,8.4/-0.5,8.9/0.4,9.1/-0.7,9.2/-0.8,9.1/-0.7,8.9/-0.4,8.9/-0.6,8.8/-0.3} {
  \fill[red!70] (\x,\y) circle (2.3pt);
}
\node[red!70] at (8.8,-1.0) {$U$};

\draw[->,thick] (0.8,0) -- (2.4,0)
  node[midway,above]{\small $H^{\otimes n}$};

\draw[->,thick] (5.2,0) -- (6.8,0)
  node[midway,above]{\small sampling};

\end{tikzpicture} 
\caption{\textit{A depiction of how the pair $(S,U)$ is sampled: (a) First, the multiset $S$ is first sampled uniformly randomly; (b) second, for each $y \in \bits^n$, a parameter $\gamma_y^{(S)}$ is calculated; (c) third, a set $U$ is sampled by including $y$ with independent probability $1 - \half \e^{-\kappa \gamma_y^{(S)}}$.}}
\end{figure}

\begin{remark}Zhandry~\cite{zhandry2024toward} constructs $(S,U)$ in a similar manner, but used a Haar-random state on the support of $|S\rangle$ as the witness state, and used a slightly different probability distribution for the sampling of $U$.
\end{remark}

We prove that a random pair $(S,U)$ sampled via this distribution, which we call $\mathsf{Strong}$, is a strong yes instance with overwhelming probability. The proof of this construction is a somewhat involved medley of concentration inequalities and polynomial approximations. An intuition for why this procedure yields strong yes instances with high probability is that we expect the matrix $M=\Pi_S \cdot H^{\otimes n} \cdot \Pi_U \cdot H^{\otimes n} \cdot \Pi_S$ to concentrate tightly around $a\ketbra{S} + b\cdot \id_S$ for constants $a$ and $b$, where $\id_S$ is the identity matrix restricted to the coordinates of $S$. The concentration in the last statement holds even only with respect to the randomness involved in sampling $U$ from $S$. We expect such a concentration since $H^{\otimes n}|S\rangle$ has high amplitude exactly on elements that are more likely to be included in $U$, while for states $|\psi\rangle$ supported on $S$ but orthogonal to $|S\rangle$, we expect the amplitude of $H^{\otimes n}|\psi\rangle$ to be distributed largely independently of the amplitudes in $U$. 

Due to this concentration, with high probability $\ket{S}$ is a quantum witness that $(S,U)$ are $\approx (a+b)$-spectrally Forrelated. Moreover, for any subset $\Delta\subseteq S$, the restriction of $a\ketbra{S} + b\cdot \id_S$ to $\Delta$ is just $\frac{a\abs{\Delta}}{\ell}\ketbra{\Delta}+b\cdot\id_\Delta$, whose top eigenvalue is just $(b + a\abs{\Delta}/\ell)\ll (a+b)$. By choosing our thresholds for yes and no instances appropriately, we have a large family of strong yes instances.

In our actual proof, we do not require proving tight concentration bounds for the entries of the matrix $M$, though we expect that such concentration bounds do hold. Nevertheless, guided by this intuition, we are able to bound the top eigenvalues of $M$ and its restrictions to small subsets $\Delta$.

More specifically, we prove that with probability $1 - 2^{-\Omega(n)}$, a pair $(S,U)$ sampled according to said procedure is a strong yes instance. Combined with the ${\mathsf{CumulativeSampler}}$ of~\Cref{thm:informal-removal-of-S}, we derive that the ${\mathsf{CumulativeSampler}}$ must produce samples from $S$ with probability at least  $2^{-q} \cdot \Omega(T)^{-2v}$ when run on $U$ for a pair $(S,U)$ sampled according to said procedure.

\begin{remark}Looking ahead to our $\QCMA$ lower-bound, choosing $\gamma_z^{(S)}$ proportional to the \emph{squared} amplitude is critical. It turns out that, if we instead choose $\gamma_y^{(S)}$ to be proportional to the (non-squared) amplitude, then the instance $(S,U)$ we obtain is actually in $\BQP$, and so also in $\QCMA$. 
This is because we can approximately synthesize $H^{\otimes n} \ket{S}$ given access to just the sign of $\mel{y}{H^{\otimes n}}{S}$ by generating the state $\frac{1}{\sqrt{2^n}} \sum_y \mathrm{sgn}(\mel{y}{H^{\otimes n}}{S}) \ket{y}$ (as first observed by Irani~\emph{et. al.}~\cite{irani2021quantum}). Therefore, the critical information is embedded in the signs of $\mel{y}{H^{\otimes n}}{S}$. However, by choosing $\gamma_y^{(S)}$ to be proportional to the squared amplitude, this information is not accessible to the verification algorithm. 
A second reason for the importance of squaring is that $\gamma_y^{(S)}$ is invariant under shifts of $S$, which plays a crucial role in our sampling upper bound. 
\end{remark}

\subsection{Sampling success probability upper bounds}

The remainder of the proof is a sampling probability upper bound, which we can express as the following statement. We use the notation $\poly(\cdot)$ to vastly simplify the statement; the technical statement is~\Cref{thm:big-upper-bound-thm}.
\begin{theorem}[Sampling probability upper bound (informal)]
    Let $\mathsf{Strong}$ be the distribution over pairs $(S,U)$ discussed previously. Then any $T$-query algorithm accessing only $U$ produces $v$ distinct samples from $S$ with probability at most
    \begin{equation}
        \qty(\frac{\poly(n,T)}{\poly(\ell)})^{v} + \qty(\frac{\poly(n, T) \sqrt{\ell}}{2^{\Omega(n)}})^{v}.
    \end{equation}
\end{theorem}
\noindent This theorem will contradict the conclusion of~\Cref{thm:informal-removal-of-S} for a choice of $v = \Omega(q)$. \\

To give intuition for the form of this bound, we can consider two edge cases.  First, when $\ell$ is large, it becomes easy to sample points from $S$, as random points are likely to be in $S$.  This gives us the second term in the sum, which becomes large when $\ell$ approaches $2^{O(n)}$.  At the same time, when $\ell$ becomes too small, $U$ becomes more concentrated, potentially revealing information about the structure of $S$.  Our proof implicitly balances these two effects to achieve our sampling upper bound.

An insightful technique for proving upper bounds on the success probability of low-query quantum algorithms is to consider the behavior of the query algorithm on a superposition of possible oracles. This technique was first used in the adversary method of Ambainis~\cite{ambainis-adversary} (to generalize the Bennett~\emph{et. al.}~\cite{BBBV} lower bound for unconstrained search). Instead of viewing the oracles as getting access to unitaries that modify the state of the algorithm, the adversary method treats the oracle as a long vector and each oracle query as a fixed phase kickback unitary on the joint state of the algorithm and oracle. For example, a query to the oracle $U$ can be described by the controlled phase unitary,
\begin{equation}
    \ket{b,x,y} \otimes \ket{\mathrm{tt}_U} \mapsto (-1)^{b\cdot U(y)} \ket{b,x,y} \otimes \ket{\mathrm{tt}_U}
\end{equation}
where $\ket{b,x,y}$ is the state of the algorithm and $\ket{\mathrm{tt}_U}$ is the long vector description of the oracle (here, $\mathrm{tt}$ stands for ``truth table''). 
We note that it is known that this kind of ``phase'' oracle is equivalent (up to a Hadamard transform) to the standard oracle.
In this manner, it is natural to consider the behavior of an algorithm when run on the superposition over oracles. Studying the behavior when run on the superposition is a useful way of arguing query lower bounds and sampling probability upper bounds over the randomness in the oracle distribution and the randomness of the algorithm. A central challenge in proving quantum query lower bounds is designing a suitable perspective on the superposition over oracles that is clean enough to prove lower bound statements. Zhandry's compressed oracle technique~\cite{zhandry2019record} is one method for effectively describing the superposition over oracles. Hamoudi and Magniez~\cite{hamoudi2023quantum} were some of the first to use the compressed oracle technique to prove sampling probability upper bounds for problems such as unconstrained search and collision finding.

\subsection{A bosonic perspective} To prove the desired sampling upper bound in this particular scenario, we introduce a new compression technique. We construct a compression of the superposition over oracles $(S,U)$ by expressing the oracle in terms of bosons. This perspective will naturally clean up much of the technical calculations as well as provide a physical perspective on the query algorithm. This view of the superposition over oracle pairs $(S,U)$ can be interpreted as a ``first quantization'' of compressed oracle techniques, which may be of independent interest. To understand this bosonic perspective in detail, it is helpful to initially ignore the oracle $U$ and instead concentrate on only constructing a purification of uniformly sampling a multiset $S$ of size $\ell$.

There are two natural ways to express the multiset $S$: we can express the set as a vector in $(\{0,1\}^n)^\ell$, each coordinate corresponding to one point in $S$, and the value at that coordinate telling us the value of that point. This representation is, however, not unique, as permuting the vector does not change the multiset.
Alternatively, we can express the multiset as a vector in $\ZZ_{\geq 0}^{2^n}$ of 1-norm $\ell$ with the $x$'th entry representing how many times $x$ appears in $S$. The second perspective can be seen as tossing $\ell$ indistinguishable balls into $2^n$ bins with each toss being uniformly random. \\

\begin{figure}[h]
\centering
\begin{tikzpicture}[scale=1, line join=round, line cap=round]

\def\cubesize{1.0}
\def\sep{1.6}      %
\def\depth{0.4}

\newcommand{\cube}[1]{
    \coordinate (A) at (#1,0);
    \coordinate (B) at ($(A)+(\cubesize,0)$);
    \coordinate (C) at ($(A)+(0,\cubesize)$);
    \coordinate (D) at ($(B)+(0,\cubesize)$);
    \coordinate (E) at ($(A)+(\depth,\depth)$);
    \coordinate (F) at ($(B)+(\depth,\depth)$);
    \coordinate (G) at ($(C)+(\depth,\depth)$);
    \coordinate (H) at ($(D)+(\depth,\depth)$);
    \draw[thick,fill=white] (A)--(B)--(D)--(C)--cycle;
    \draw[thick,fill=gray!10] (C)--(D)--(H)--(G)--cycle;
    \draw[thick,fill=gray!5] (B)--(D)--(H)--(F)--cycle;
    \draw[thick] (A)--(E)--(F)--(B);
    \draw[thick] (E)--(G);
    \draw[thick] (G)--(H);
}

\foreach \i in {0,...,5}{
    \cube{\i*\sep}
}

\node at (6.25*\sep+0.5*\cubesize,0.5*\cubesize) {$\cdots$};

\cube{7.2*\sep}

\foreach \i in {1,...,5}{
    \foreach \j in {1,2}{
        \fill ($(0,0)+(\j*0.4,\i*0.15)$) circle (2pt);
    }
}

\foreach \a/\b in {0.3/0.5,0.7/0.5}{
    \fill ($(2*\sep,0)+(\a*\cubesize,\b*\cubesize)$) circle (2pt);
}

\foreach \a/\b in {0.3/0.5,0.7/0.5}{
    \fill ($(4*\sep,0)+(\a*\cubesize,\b*\cubesize)$) circle (2pt);
}

\fill ($(7.2*\sep,0)+(0.5*\cubesize,0.5*\cubesize)$) circle (2pt);

\end{tikzpicture} \caption{\textit{A depiction of the bosonic Fock state $\ket{\psi} = \ket{10, 0, 2, 0, 2, 0, \ldots, 1}$.}}
\end{figure}
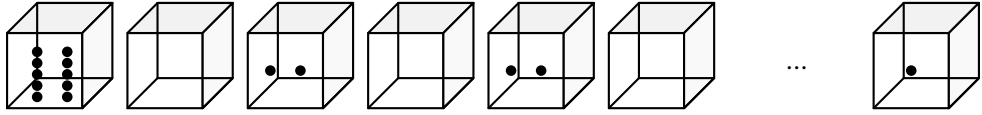

The quantum mechanical analog of a multiset is a collection of \emph{bosons}. Consider a system of $\ell$ bosons, each of which occupies a ``state'' from $\{0,1\}^n$. A characteristic of bosons is that any number of bosons may occupy the same state. Thus, the state of the $\ell$ bosons\footnote{Whereas, for a fermionic system, by the exclusion principle, two fermions cannot occupy the same state, and so a fermionic system corresponds to an ordinary set.} is exactly described by a multiset $S$ over $\{0,1\}^n$. Then, the two perspectives above give two different ways to represent a bosonic system.

We now give a brief primer on bosonic systems. Bosonic systems are described in terms of \emph{modes}, each of which corresponds to an independent quantum degree of freedom. In physics, a mode may be defined in position space, where it is associated with a localized site or spatial region, or in momentum space, where it is associated with a plane-wave excitation. These two descriptions are related by a Fourier transform, so that switching between position and momentum modes is analogous to changing bases in a Hilbert space.  
In the position basis, one specifies a set of operators $\{\hat a_x, \hat a_x^\dagger\}$ that annihilate or create a boson at position $x$, respectively. The number of bosons in a particular (position) mode is unrestricted, with the number operator $\hat n_x = \hat a_x^\dagger \hat a_x$ measuring the occupation of (position) mode $x$. This representation is natural when considering local interactions or spatially constrained dynamics.  
In the momentum basis, one works instead with operators $\{\tilde{a}_x, \tilde{a}_x^\dagger\}$ that annihilate or create excitations of definite momentum $x$, respectively. Analogously, the number of bosons in a momentum mode is unrestricted, with the number operator $\tilde{n}_x = \tilde{a}_x^\dagger \tilde{a}_x$ measuring the occupation of momentum mode $x$. As one might suspect, the total number of bosons in the position basis $\hat N = \sum_x \hat n_x$ equals the total number of bosons in the momentum basis $\tilde{N} = \sum_x \tilde{n}_x$. The momentum description is particularly convenient for systems with translation invariance, where the state of the system takes a simple diagonal form in momentum space.

In this work, we consider a simplified ``computer science'' perspective on bosons. Our system consists of $2^n$ modes indexed by elements of $\bits^n$. Sometimes it will be convenient to index them using the isometry $[2^n] \equiv \bits^n$; so the 0-momentum mode is equivalent to the $0^n$-momentum mode. Instead of relating the position and momentum creation/annihilation operators by the Fourier transform, we define our momentum creation/annihilation operators in terms of the Hadamard transform: 
\begin{equation}
\tilde{a}_x = \frac{1}{\sqrt{2^n}} \sum_{z \in \bits^n} (-1)^{x\cdot z} \hat a_z, \quad \tilde{a}_x^\dagger = \frac{1}{\sqrt{2^n}} \sum_{z \in \bits^n} (-1)^{x\cdot z}\hat a_z^\dagger. \label{eq:fourier-transform-of-boson-ops}
\end{equation}
In this equation, the indexing variables $x, z$ are elements of $\bits^n$ and $x \cdot z$ equals the inner product of the vectors over $\FF_2$.
In~\Cref{sec:prelim-bosons}, we elaborate on the mathematics of bosons. For now, however, simply observe that the $0$-momentum mode creation operator, $\tilde{a}_{0}^\dagger$, is the uniform superposition over all position creation operators as it equals $\frac{1}{\sqrt{2^n}} \sum_z \hat a_z^\dagger$. Therefore, the uniform superposition of a single boson in a random position mode is a single boson in the 0-momentum mode. This generalizes to: the uniform superposition of $\ell$ bosons each in an independently random position mode is $\ell$ bosons in the 0-momentum mode.
We can interpret this as a computer science version of the \emph{Heisenberg uncertainty principle}~\cite{heisenberg-uncertainty-principle}: certainty about the momentum of a set of bosons implies maximal uncertainty about the positions of the bosons.

Lastly, a convenient basis for studying bosonic states is the \textit{Fock} basis, which records the number of bosons in each mode. One can write down Fock basis states in either the position or momentum basis. The bosonic (position) Fock basis is a collection of \emph{orthonormal} states of the form $\ket{\ell_0, \ell_1, \ldots, \ell_{2^n-1}}$ with each $\ell_x \in \ZZ_{\geq 0}$ where the total number of bosons is $\sum_x \ell_x$. In this setting, the annihilation and creation operators can be defined as the operators such that:
\begin{xalign}[eq:boson-action-on-fock]
     \ket{\ell_0, \ldots, \ell_x - 1, \ldots, \ell_{2^n-1}} &= \frac{\hat a_x}{\sqrt{\ell_x}} \ket{\ell_0, \ldots, \ell_x, \ldots, \ell_{2^n-1}} \\
    \ket{\ell_0, \ldots, \ell_x + 1, \ldots, \ell_{2^n-1}} &= \frac{\hat a_x^\dagger}{\sqrt{\ell_x + 1}} \ket{\ell_0, \ldots, \ell_x, \ldots, \ell_{2^n-1}}.
\end{xalign}

\subsection{The hardness of sampling without queries}

Having set up this machinery, we observe that a purification of sampling a uniformly random multiset of size $\ell$ is exactly the state of $\ell$ bosons in the 0-momentum mode. And having defined a suitable purification, we can progress to sampling probability bounds. In the bosonic perspective, guessing elements of $S$ is equivalent to identifying position modes where the algorithm believes bosons reside. A first step to proving that~\Cref{thm:informal-removal-of-S} is too good to be true, is showing that an algorithm that makes 0 queries to the oracle cannot identify with good probability position modes which contain bosons.

This ``baby'' problem is easy enough to prove via classical combinatorics, but we will purposefully resolve it by appealing to more quantum mechanical techniques. 
Observe that, by the relationship of position and momentum operators~(\eqref{eq:fourier-transform-of-boson-ops}), the state of a single boson in a fixed momentum mode $\tilde a_x^\dagger \ket \vac$ is equal to $\sum_y (-1)^{x \cdot y} \hat{a_y}^\dagger \ket \vac$, a superposition of being in every position mode, with phases corresponding to $x$. Therefore, the probability that any position guess $y$ is correct is $2^{-n}$. This is one example of a Heisenberg uncertainty principle~\cite{heisenberg-uncertainty-principle} with a more general form being that if the state of a single boson is supported on at most $r$ momentum modes $x_1, \ldots, x_r$, i.e., the state equals $\sum_{i = 1}^r \alpha_i \tilde{a}_{x_i}^\dagger \ket \vac$, then the probability that any position guess $y$ is correct is $\leq r/2^n$. 
In greater detail, what we can actually prove is for any $z \in \bits^n$, $\ev{\hat n_z}{\psi} = r \ell/2^n$ for any $\ell$-boson state supported on at most $r$ modes. 
Recall that $\hat n_z$ is the operator \textit{counting} the number of bosons at location $z$. We can define $\Pi_{\hat n_z > 0}$ as the projector onto having a non-zero number of bosons in location $z$. Observe that $\hat n_z \ge \Pi_{\hat n_z > 0}$. Then, the probability that a guess of location $z$ is correct is equal to $\ev{\Pi_{\hat n_z > 0}}{\psi} \leq \ev{\hat n_z}{\psi} = r\ell/2^n$. 
From here, we can conclude that the first guess has probability at most $r \ell/2^n$ of being correct, irrespective of which position mode $z$ is guessed.

In this argument, we made a conceptual shift that turns out to be quite useful mathematically. Instead of studying the expectation of the state with respect to $\Pi_{\hat n_z > 0}$, we chose to study the expectation with respect to $\hat n_z$. This is mathematically equivalent to applying Markov's inequality ($\Exp[\Pi_{X > 0}] = \Pr[X > 0] \leq \Exp[X]$ for any random variable $X \geq 0$). This weaker bound will be sufficient for our argument. This argument naturally generalizes: to upper bound the probability of finding bosons at each of $z_1,\ldots,z_v$ for distinct indices $z_1,\ldots,z_v$ (which is equal to the expectation with respect to $\Pi_{\hat{n}_{z_1} > 0} \ldots \Pi_{\hat{n}_{z_v} > 0}$), the Markov inequality upper bound lets us instead bound the expectation with respect to $\hat n_{z_1} \ldots \hat n_{z_v}$. Next, observe that $\hat n_{z_1} \ldots \hat n_{z_v}$ equals $\hat a_{z_1}^\dagger \ldots \hat a_{z_v}^\dagger\hat a_{z_v} \ldots \hat a_{z_1}$ when the indices are distinct. Therefore, the expectation of $\hat n_{z_1} \ldots \hat n_{z_v}$ is equal to $\norm{\hat a_{z_v} \ldots \hat a_{z_1} \ket{\psi}}^2$, the norm of the state after applying $v$ annihilation operators.

To analyze the norm of applying two annihilation operators $\hat a_{z_2}\hat a_{z_1}$ when the state of the oracle is $\ell$-bosons in the 0-momentum mode, we can study how the multiplicative norm decreases with each sequential operator application. We observe that the normalized state after annihilating a boson at location $z_1$ will again only be supported on bosons in the 0-momentum mode: annihilating this boson did not affect the remaining bosons, and so the resulting state will be $\ell-1$ bosons in the 0-momentum mode. It follows that the following annihilation decreases the norm multiplicatively by at least $(\ell - 1)/2^n$. We can continue this observation proving an upper bound of $\leq (\ell/2^n)^v$ for the norm of the vector after $v$ annihilations. This proves a bound on the success probability of guessing $v$ locations without making any queries at all. If our initial state was supported on $r$ distinct momentum modes, the same argument would have produced an upper bound of $\leq (r \ell/2^n)^v$.

\subsection{Understanding the action of queries to the $U$ oracle}

To prove an upper bound on the sampling probability after some queries to $U$, we need to extend the prior sampling probability upper bound to a broader set of initial states. The previous argument was particular to the state of $\ell$ bosons in the 0-momentum mode. To extend the argument, it is first illustrative to understand what the superposition of the algorithm and oracle registers looks like after making $T = \poly(n)$ queries to the $U$ oracle. What we discover is that the state of the oracle after $T$ queries can be described as what we coin a \emph{quasi-even condensate} and unpack below.

First, the phrase \textit{condensate}. We borrow the term condensate from many-body physics, where it denotes a regime in which a macroscopic fraction of bosons occupy a single-particle mode. In our case, we use it to refer to a system where almost all the bosons are in the 0-momentum mode. Concretely, for us, a momentum Fock state is an $r$-condensate if at least $\ell - r$ of the bosons are in the 0-momentum mode. Equivalently, at most $r$ bosons are not in the 0-momentum mode. An $r$-condensate may be a superposition of momentum Fock states that are $r$-condensates. This work considers states where $\ell = 2^{cn}$ and we will consider $r = \poly(n)$, so only a negligible fraction of the bosons are not in the $0$-momentum mode.

Second, the phrase \textit{quasi-even}. For a momentum Fock state, we say that the state is $o$-quasi-even if at most $o$ of the momentum modes (except the 0-mode) have an odd number of bosons in them. A general state is $o$-quasi-even if it is entirely supported on momentum Fock states which are $o$-quasi-even. A priori, the definition of quasi-even isn't nearly as motivated as the definition of a condensate (which has a natural physical interpretation). However, upon analysis of queries to the $U$ oracle, studying how quasi-even a state is will become apparent.

Both being a condensate and being quasi-even are properties in the momentum basis. Therefore, we can define projectors onto the momentum Fock states that satisfy them. Since the projectors are both diagonal in the momentum Fock basis, they commute. So, the definition of a state as a $(r,o)$-quasi-even condensate is well-defined. Roughly speaking, our characterization theorem shows that the post-query state $\ket*{\psi_{\mathrm{PQ}}}$ satisfies the following: for all $\iota > 0$, there exists another state $\ket*{\psi'}$ such that (a) $\norm{\ket*{\psi_{\mathrm{PQ}}} - \ket*{\psi'}} \leq \iota$ and (b) $\ket*{\psi'}$ is a $(r,o)$-quasi-even condensate for $r = \poly(n,T, \log(1/\iota))$ and $o \ll v/4$ with overwhelming probability (i.e., $(1 - \poly(T) / \sqrt{\ell})^{O(v)}$) as long as $\poly(T) \ll \ell$. 

This particular characterization of the post-query state is important because we are additionally able to show that the sampling success probability of any algorithm for which the post-query state is such a quasi-even condensate decays exponentially fast with $v$. These two ingredients combine to prove our sampling probability upper bound. We have not yet answered two fundamental questions: (1) why are post-query states effectively quasi-even condensates and how do we prove it, and (2) why is there a sampling probability upper bound for quasi-even condensates.

The answer to the first question comes from understanding the behavior of a single query to the oracle $U$ at $y$. What we can formally prove is that a query at $y$ applies a polynomial in the exponential function of the ``double $y$-momentum hopping operator'', which is defined as:
\begin{equation}
    \tilde{\H}_y \defeq \frac{1}{\ell} \sum_{x, x' \in \bits^n} \tilde{a}_{x \oplus y}^\dagger \tilde{a}_{x' \oplus y}^\dagger \tilde{a}_x \tilde{a}_{x'}.
\end{equation}
As the name may suggest, the action of $\tilde{\H}_y$ on a momentum Fock state is to pick two modes $x, x'$ which contain bosons and ``hop'' each of the bosons by $y$. The action of the $\tilde{\H}_y$ moves around momentum Fock states, but perhaps surprisingly, it is diagonal in the position basis and its action can be described in terms of the coefficients $\gamma_z^{(S)}$ which were defined in~\eqref{eq:def1-of-gamma}.

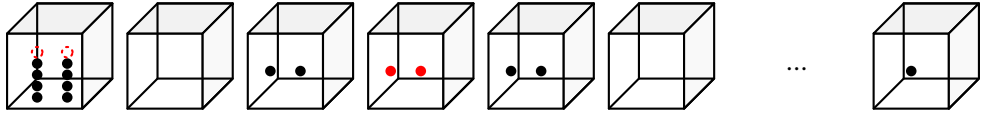
\begin{figure}[h]
\centering
\begin{tikzpicture}[scale=1, line join=round, line cap=round]

\def\cubesize{1.0}
\def\sep{1.6}
\def\depth{0.4}

\newcommand{\cube}[1]{
  \coordinate (A) at (#1,0);
  \coordinate (B) at ($(A)+(\cubesize,0)$);
  \coordinate (C) at ($(A)+(0,\cubesize)$);
  \coordinate (D) at ($(B)+(0,\cubesize)$);
  \coordinate (E) at ($(A)+(\depth,\depth)$);
  \coordinate (F) at ($(B)+(\depth,\depth)$);
  \coordinate (G) at ($(C)+(\depth,\depth)$);
  \coordinate (H) at ($(D)+(\depth,\depth)$);
  \draw[thick,fill=white] (A)--(B)--(D)--(C)--cycle;
  \draw[thick,fill=gray!10] (C)--(D)--(H)--(G)--cycle;
  \draw[thick,fill=gray!5] (B)--(D)--(H)--(F)--cycle;
  \draw[thick] (A)--(E)--(F)--(B);
  \draw[thick] (E)--(G);
  \draw[thick] (G)--(H);
}

\foreach \i in {0,...,5}{\cube{\i*\sep}}
\node at (6.25*\sep+0.5*\cubesize,0.5*\cubesize) {$\cdots$};
\cube{7.2*\sep}

\foreach \i/\j in {1/1,2/1,3/1,4/1,1/2,2/2,3/2,4/2}{
  \fill ($(0,0)+(\j*0.4,\i*0.15)$) circle (2pt);
}
\foreach \a/\b in {5/1,5/2}{
  \draw[red,dotted, thick] ($(0,0)+(\b*0.4,\a*0.15)$) circle (2pt);
}

\foreach \a/\b in {0.3/0.5,0.7/0.5}{
  \fill ($(2*\sep,0)+(\a*\cubesize,\b*\cubesize)$) circle (2pt);
}

\foreach \a/\b in {0.3/0.5,0.7/0.5}{
  \fill[red] ($(3*\sep,0)+(\a*\cubesize,\b*\cubesize)$) circle (2pt);
}

\foreach \a/\b in {0.3/0.5,0.7/0.5}{
  \fill ($(4*\sep,0)+(\a*\cubesize,\b*\cubesize)$) circle (2pt);
}

\fill ($(7.2*\sep,0)+(0.5*\cubesize,0.5*\cubesize)$) circle (2pt);

\end{tikzpicture} \caption{\textit{A depiction of $\tilde{a}_{0 \oplus 3}^\dagger \tilde{a}_{0 \oplus 3}^\dagger \tilde{a}_0 \tilde{a}_0 \ket{\psi}$ which is one component in the sum $\tilde{\H}_3 \ket{\psi}$.}}
\end{figure}
In slightly more detail, recall that we choose to include $y \in U$ with (independent) probability $1 - \e^{-\kappa\gamma_y^{(S)}}/2$ where $\gamma_y^{(S)}$ was defined in~\eqref{eq:def1-of-gamma} and $\kappa$ is a small constant (think $1/10$). The use of the exponential functions will have further technical implications for the rest of the proof, but for intuition, observe that the first-order Taylor approximation of this probability equals\footnote{An astute reader might question why the constants $1/2$ and $\kappa$ appear. Both of these constants are necessary for our proof techniques, but are probably not necessary.}~$1/2 + \kappa\gamma_y^{(S)}/2$. On the other hand, recall that shifts, or ``hops'', in the momentum basis are diagonalized in the position basis. For analogous reasons, $\tilde{\H}_y$ is a diagonal matrix in the position Fock basis. 
Furthermore, for position Fock basis state $\ket{\mathrm{tt}_S}$ described by the multiset $S$ of size $\ell$, the corresponding diagonal entry is $\ev{\tilde{\H}_y}{\mathrm{tt}_S} = \gamma_y^{(S)} - 1$. See~\eqref{eq:action-of-single-hop-operator} for details.  Thus, we see that the diagonal entries of $\tilde{\H}_y$ in the position Fock basis are closely related to the probability that $z$ is included in $U$. This relationship allows us to express the effect of querying $U$ at $z$ in terms of $\tilde{\H}_y$.

What we will show is that a query to a random $U$ according to the prescribed description can be expressed as an exponential operator in the terms $-\gamma_y^{(S)}$, which by the prior relation is an exponential operator in the terms $-\tilde{\H}_y$.
Proving this to be true is significantly more challenging, and we save most of the details for the technical content of this paper.
Of the numerous challenges required to prove this statement, the first is that for a fixed $S$ oracle, we need to analyze the result of querying a purification of the distribution of $U$ oracles. This is because evaluating whether the samples guessed by the algorithm are correct does not require $U$; therefore, the quantity we are actually interested in studying is the reduced density matrix on the algorithm's register and the register of the oracle encoding the $S$ oracle -- i.e., we can take the partial trace over the register encoding the $U$ oracle. To make calculating this partial trace easier, we encode queries to the $U$ oracle using another compressed oracle. By doing so, we can express the action of queries to the $U$ oracle on the algorithm and $S$ oracle registers in terms of \emph{Kraus operators}. These Kraus operators can be defined as a function of $\tilde{\H}_y$ for the locations $y$ queried by the algorithm. 

This in turn leads us to the second challenge: the Kraus operators (formally defined in~\eqref{eq:kraus-operators}) will be exponential functions in $- \tilde{\H}_y$. The exponential function appears here because we include $y \in U$ with probability $1 - \e^{-\kappa \gamma_y^{(S)}}/2$. Recall that our intended goal was to prove that the action of a query can be well approximated by a polynomial in $\tilde{\H}_y$. However, we can observe from the definition that the values $\gamma_y^{(S)}$ can lie in the range $[0, \ell]$. But, any polynomial approximation of $\e^{- \kappa \gamma_y^{(S)}}$ that is $\eps$-accurate on the entire region $[0,\ell]$ will require a very large degree --- for example, a Taylor series approximation will require degree $\Omega(\ell \log(1/\eps))$. Furthermore, a polynomial of such degree in $\tilde{\H}_y$ applied to the initial state will no longer be a condensate, and we don't know of techniques for proving sampling upper bounds for states which are not condensates.
The point being that an approximation on the entire range $[0,\ell]$ is untenable. However, it is also unnecessary.
Instead, what we observe is that each $\gamma_y^{(S)}$ for $y \neq 0$ is approximately distributed as the square of a normal Gaussian. Furthermore, since our analysis is, in some sense, computing the average success probability over all possible $(S,U)$ pairs, what is actually required is that the polynomial approximation of the exponential function is good \emph{proportional} to the distribution over $\gamma_y^{(S)}$. In other words, since most of the mass is around 1, the polynomial approximation needs to be very strong in this region, but the approximation can be exponentially bad at large values, as the probability mass of $\gamma_y^{(S)}$ is likewise exponentially small. We show that applying ``flat'' polynomial approximations~\cite{bakshi2024learning} to the exponential function\footnote{An astute reader might question why we choose to use an exponential function to define the probability of including $y \in U$. Indeed, we too spent much time attempting to prove the result without the use of an exponential function. If we had been able to use a polynomial function, then there would be no need to consider a suitable approximation. However, we were unable to find a low-degree polynomial that is within $[0,1]$ on the range $[0,\ell]$ and has sufficient signal to provide a $\QMA$ algorithm for this problem. Therefore, we required using the exponential function to guarantee that the probability of including $y \in U$ was well defined. In turn, this necessitated the use of a flat approximation. We suspect that a truncation function (for example, $\max\{1, \epsilon \gamma_y^{(S)}\}$ would have also sufficed, provided there exists a good flat approximation of this function.} developed by Narayanan~\cite{narayanan2024improved} suffices to give a $\poly(n,t,\log(1/\iota))$-degree polynomial that well-approximates the action of a query for our means.  See \Cref{subsec:polynomial_overview} for a technical overview of the polynomial approximation we take.

\subsection{A characterization of post-query states}

We promote this observation to prove that the action of the queries can be approximated to $\iota$-precision by a polynomial of degree $\poly(n,T,\log(1/\iota))$ in the operators $\tilde{\H}_y$. 
This will prove that the state is close to the subspace of $r$-condensates for $r = \poly(n,T,\log(1/\iota))$ as each action of the double-momentum hopping operator $\tilde{\H}_y$ will move at most 2 bosons from the 0-momentum subspace. 

However, recall that we start with $\ell = 2^{cn}$ bosons in the 0-momentum mode. Therefore, any $\poly(n)$ query algorithm, in expectation, moves an imperceptible number of the bosons from the 0-momentum mode. Again, this is the motivation behind the nomenclature choice of ``condensate''. Had we been able to prove that the state was an $r$-condensate for $r \ll v$, the proof would be over as we could apply an uncertainty principle. Unfortunately, this is not the case, as the sampler we construct makes $T = vt$ queries, if we started from a $\QCMA$ algorithm that made $t$ queries. 
What we will be able to prove is that with each query, the probability that the state becomes less ``quasi-even'' is negligible in $n$. In total, it is very unlikely that the post-query state is not a quasi-even condensate. To understand why this is the case, let us imagine a simplified setting where the algorithm makes a query, the action on the oracle register is equal to applying the double-hopping operator $\tilde{\H}_y$. This is an obvious simplification as the operator is not unitary, but ignore this for now.  We think of this operator, in words, as the following operation: pick two random bosons at modes $x$ and $x'$, shift them to be at $x \oplus y$ and $x' \oplus y$, and multiply the norm of the state by the number of bosons already at modes $x \oplus y$ and $x' \oplus y$.  Then it is clear to see that on a condensate (i.e., a state such that the vast majority of bosons are in the $0$-momentum mode), the double-hopping operator is almost entirely characterized by its restriction onto the $0$- and $y$-momentum modes (i.e., shifting bosons into or out of the condensate).  If this really was an exact characterization, then the state of a quantum algorithm querying this non-physical operation would have exactly $2$ bosons in every momentum mode that was queried by the algorithm, and every non-zero momentum mode would be perfectly paired up.  The real double-hopping operator has some probability (we bound this by $\poly(r) / \sqrt{\ell}$) on each query to affect a boson that is not in the condensate, but since this number is negligible compared to the number of queries the algorithm makes, we can prove a strong bound on the probability that we observe $\ll v/4$ many momentum modes occupied by an odd number of bosons. 

The analysis is not as simple, though, as the exact action of a query is not the double-hopping operator, but consists of terms looking like the exponential of the double-hopping operator: $\e^{-\kappa \tilde{\H}_y}$.  This presents new challenges to analyzing the effects of queries, which we resolve by using techniques from perturbation theory. Having shown that the double-hopping operator mostly does not change the number of odd indices in the condensate, we interpret it as a small perturbation (i.e., the part of the double-hopping operator that would change the number of odd indices) added to a large operator (i.e., the part of the double-hopping operator that preserves the evenness of all entries).  Applying the Dyson series for the exponential function allows us to bound the effect of $\e^{-\kappa \tilde{\H}_y}$ on the number of even entries in a similar manner as before.    

\subsection{Momentum conservation and boson pairs}

One might look at the property of being a quasi-even condensate outlined in the previous subsection and wonder why it \emph{should} say anything about the probability that a sampler succeeds at sampling many points from the oracle.  To see this, it will be instructive to think about an analog of Noether's theorem~\cite{Noether1918} to our oracle.  In classical mechanics, Noether's theorem roughly states that conservation laws correspond to symmetries in a system.  In our case, one such conservation law is the total momentum of the system.  In particular, for all $y$, the double-hopping operator $\tilde{\H}_y$ maps a system with total momentum $0$ (where we define total momentum to use addition over $\mathbb{F}_2$) to a system that also has total momentum $0$.  Analogous with classical mechanics, this corresponds exactly to translational invariance.  In other words, the state of the system does not change if we apply the shift operator, $\mathsf{Shift}_{x} \hat{a}_{z} \mathsf{Shift}_x = \hat{a}_{z \oplus x}$.  Thus, the conservation of momentum trivially implies a bound on the probability that any algorithm outputs a single point of $S$.  

However, making this observation about the system as a whole is too brittle on its own to bound the probability of sampling more than one point.  At a proof level, an algorithm that successfully guesses one point from $S$ will change the total momentum of the system by effectively annihilating one of the bosons.  Even worse, an algorithm that queries the entire truth table of $U$ would expect to be able to output many points of $S$ given a single point in $S$.  Thus, we must use a more fine-grained feature of the states of an algorithm to provide an upper bound on the sampling probability.  

One way to reconcile this is to make the following observation: translational invariance applies to \emph{any} collection of bosons that have $0$-total momentum, simply because applying the shift operator corresponds to applying a phase in the momentum basis, and having $0$-total momentum is equivalent to the phases canceling each other out.  Thus, whenever bosons are paired up (i.e., there is another boson in the same momentum mode), the pair of bosons will obey translational invariance, and (in a loose sense, since all bosons are identical) the sampler is very unlikely to sample one of the bosons from this pair.  This property is far more robust than the original observation about $0$-total momentum, as successfully sampling a boson affects at most one of the paired up bosons, although our sampler upper bound proof does not use an inductive argument.  This critical observation is the reason why having almost all of the momentum modes occupied by an even number of bosons allows us to prove a sampler upper bound.  

While the intuition that paired up bosons should be hard to find, our proof that algorithms whose purified states are supported on quasi-even condensates are hard to sample from does not use an inductive argument.  Instead, we directly upper bound the spectral norm of any operator of the form $\hat{n}_{z_1}\ldots \hat{n}_{z_v}$ on the subspace of quasi-even condensates using the max row $1$-norm, and use the quasi-even and condensate properties to carefully bound the constructive interference that quasi-even condensates can generate.

\subsection{Putting it all together}

All together, we have argued that the existence of a $\QCMA$ algorithm for spectral Forrelation implies a sampler which morally guesses the positions of bosons given only query access to their momentum information. We then prove, by considering the query action on the purification of all oracles, that such a sampler cannot perform well unless it makes a superpolynomial number of queries. 

The proof technique described in this paper is the technique we arrived at after considerable research and challenges. We highlight the challenges at the end of the result in our concluding remarks~(\Cref{sec:concluding-remarks}).

\section{History of the $\QMA$ versus $\QCMA$ problem}

The question of whether $\QMA$ equals $\QCMA$ first appeared in the survey of Aharanov and Naveh~\cite{aharonov2002quantum} with the first indication of a separation given by Aaronson and Kuperberg~\cite{aaronson2007quantum}. The following is a brief review of the progress made on the problem since then: 
An early candidate classical oracle separation was given by Lutomirski~\cite {lutomirski2011component}, but the candidate lacked a proof.  More recently, a number of results have made progress towards the goal of a classical oracle separation by proving separations under different restrictions on how the oracle is accessed.  Fefferman and Kimmel~\cite{fefferman2015quantum} showed a separation assuming that the oracle is an ``in-place permutation oracle'', a non-standard model where the oracle irreversibly permutes the input state; the resultant object was ``less quantum'' than a reflection about a Haar-random state but still inherently quantum.~\cite{fefferman2015quantum} also presents the first techniques for proving query lower bounds for $\QCMA$; they lift an $\AM$ lower bound for set size estimation of Goldwasser and Sipser~\cite{goldwasser-sipser} to a $\QCMA$ lower bound for set size. Unfortunately, as is the case for many $\QCMA$ lower bounds, the same lower bound will apply for $\QMA$, thereby not providing a separation. The corresponding $\QMA$ lower bound for set size estimation was proven later using Laurent polynomial techniques by Aaronson \emph{et. al.}~\cite{aaronson_et_al:LIPIcs.CCC.2020.7}. Future results starting from Natarajan and Nirkhe~\cite{natarajan2024distribution} showed $\QMA$ versus $\QCMA$ separations for weakened notions of a prover or verifier. Natarajan and Nirkhe build on the prior works for set size estimation to construct a set size estimation problem with an efficient $\QMA$ algorithm by adding graph structure. For example,~\cite{natarajan2024distribution} proved the separation assuming the witness was only a function of some portion of the oracle -- this can be equivalently expressed as a distribution testing problem. Recently, Agarwal and Kundu~\cite{agarwal2025cautionarynotequantumoracles} have shown that one should be cautious when dealing with distributional oracles or unitary oracles as separations with respect to such oracles may not imply classical oracle separations. Secondly, a line of work has used quantum advantage relative to unstructured oracles~\cite{yamakawa2024verifiable} to separate $\mathsf{QMA}$ from $\mathsf{QCMA}$ in settings where the quantum verifier was limited: \cite{liu2022non,li2023classical} gave a separation assuming the verifier can only make classical oracle queries, and more recently \cite{ben2024oracle} gave a separation which allows the verifier to make quantum queries, but assumes the adaptivity of the queries is sub-logarithmic.  

Motivated by this seemingly dual requirement of having to apply quantum query complexity tools to highly structured oracles, a pair of works, by Zhandry~\cite{zhandry2024toward} (which this work takes some inspiration from) and Liu, Mutreja, and Yuen~\cite{liu2024qma}, showed connections between the $\mathsf{QMA}$ versus $\mathsf{QCMA}$ problem and \emph{pseudorandomness} against quantum adversaries.  
\cite{liu2024qma} improved the lower bound analysis of~\cite{natarajan2024distribution} and proposed a conjecture about the pseudorandomness of $\delta$-dense permutations, which, if proven, demonstrates another classical oracle separation between $\mathsf{QMA}$ and $\mathsf{QCMA}$.  
Previously,~\cite{guo2021unifying} proved that the $\delta$-dense conjecture would imply that any quantum algorithm making queries to a random oracle can be simulated by an efficient classical one, a major open question in quantum query complexity~\cite{aaronson2009need}.  A similar open problem would need to be resolved in order to make the separation in~\cite{liu2024qma} unconditional. We nevertheless believe that the results and the statistical conjecture of \cite{liu2024qma} remain interesting, even in light of this result, since proving that the oracle of \cite{liu2024qma} works to separate $\mathsf{QCMA}$ and $\mathsf{QMA}$ (either through their conjecture about $\delta$-dense permutations or by different means) would provide an alternate oracle separation that sheds light on the Aaronson-Ambainis conjecture.

\section{Observations and open questions}

\begin{enumerate}[leftmargin=1.2em]
    \item \textbf{Black-box separations and the two basis thesis}~Recently, Ma and Natarajan~\cite{ma2025basessufficeqma1completeness} identified a $\QMA_1$-complete family of local Hamiltonians where every term is 6-local and either diagonal in the standard or Hadamard bases. The collection of local Hamiltonian terms diagonal in the standard basis forms a constraint satisfaction problem (CSP). Likewise, the local Hamiltonian terms diagonal in the Hadamard basis form a second CSP. If we let $S$ be the solutions to the first CSP and $U$ the solutions to the second CSP, Ma and Natarajan's result can be expressed as the statement: ``Deciding 1 vs $1 - 1/\poly(n)$ spectral Forrelation is $\QMA_1$-complete even when $S$ and $U$ are the solution sets to CSPs''. Previously, Cubitt and Montanaro~\cite{doi:10.1137/140998287} had proven that the local Hamiltonian problem where all terms are of the type $\alpha_{ij} X_i \otimes X_j$ or $\beta_{ij} Z_i \otimes Z_j$ is $\QMA$-complete for a completeness-soundness gap of $1/\poly(n)$. However, Cubitt and Montanaro's Hamiltonians were necessarily \emph{frustrated} as they were built from perturbation theory gadgets. 
    
    The result presented here is an oracular variant of these two results.
    Ma and Natarajan's result also shows that it is $\QMA$-complete to decide the $1 - 1/\exp(n)$ vs $1 - 1/\poly(n)$ spectral Forrelation problem even when $S$ and $U$ are the solution sets to CSPs. Whereas this result shows that the problem of deciding $59/100$ vs $57/100$ spectral Forrelation for general sets $S$ and $U$ is not in $\QCMA$. Therefore, if $\QCMA$ were to equal $\QMA$, our black-box separation concretely says that the $\QCMA$ algorithm \textit{must depend on the structure} of the two CSPs. This is the analog of how the unconstrained search problem lower bound for $\BQP$~\cite{BBBV} proves that if a $\BQP$ algorithm exists for $\NP$, it must depend on the structure of the CSPs.

    Furthermore, our oracle separation does not have perfect completeness, which, to the best of our knowledge, all previous candidate oracle separations did. And we do not know a technique for adapting this protocol to have perfect completeness, as the state $H^{\otimes n} \ket{S}$ has some support on every basis vector. 
    
    More broadly, our results fall into a larger family of quantum computation results that satisfy the ``two-basis thesis''---that, computationally speaking, it suffices to consider computation or Hamiltonian terms that are either in the standard or Hadamard basis. Other results that satisfy the two-basis thesis include the BB84 protocol~\cite{bb84}, Weisner's quantum money scheme~\cite{weisner10.1145/1008908.1008920}, Aaronson and Christiano's money scheme~\cite{aaronson2012quantum}, the Mermin-Peres magic square game~\cite{mermin1990simple,peres1990incompatible}, Mahadev's measurement protocols~\cite{Mahadev:EECS-2018-88}, and quantum codes such as the NLTS Hamiltonian construction of Anshu, Breuckmann, and Nirkhe~\cite{nlts10.1145/3564246.3585114}.

    \item \textbf{Upgrading quantum oracle separations}~Variations of the Aaronson and Kuperberg~\cite{aaronson2007quantum} oracle are incredibly prevalent in quantum complexity theory and quantum cryptography. Having demonstrated that we can suitably replace the~\cite{aaronson2007quantum} quantum oracle with a classical oracle in this particular setting, we posit that more results are due for upgrades to classical oracles. In particular, we identify the question of separating $\QMA$-search from $\QMA$-decision with respect to a classical oracle as incredibly pertinent; the quantum oracle result was proven by Irani \emph{et. al.}~\cite{irani2021quantum}.
    
    \item \textbf{Cryptographic primitives.}~More broadly, quantum oracles are equally prevalent in quantum complexity theory and quantum cryptography. For example, many complexity-theoretic and cryptography \emph{separations} are proven by means of quantum oracles. One notable example is Kretschmer's separation~\cite{kretschmer:LIPIcs.TQC.2021.2} between quantum pseudorandomness and (for example) one-way functions. Very recently, there has been an explosion of unitary oracle separations separating various notions of quantum cryptography~\cite{EC:CheColSat25,EC:BosCheNeh25,EC:BMMMY25,C:GolZha25,EPRINT:Barhoush25,EPRINT:GLMY25,EPRINT:AnaGulLin25}.
    Often, for the cryptographic use cases, there hasn't been an easy classical oracle replacement. Can we use the techniques introduced in this work to make these improvements?

    \item \textbf{The computational complexity of clonability}~Nehoran and Zhandry~\cite{nehoran2024computational} introduce the concept of $\mathsf{ClonableQMA}$, which is the class of decision problems decidable with a quantum witness that is also efficiently clonable. \textit{It is not difficult to see that our proof also extends to separate $\mathsf{ClonableQMA}$ from $\QMA$ with respect to a classical oracle, as the sampler generated in~\Cref{thm:informal-removal-of-S} can be constructed given the clonability.} The complexity class $\mathsf{ClonableQMA}$ was identified as the complexity-theoretic generalization of many cryptographic tasks that build on the idea that some states are hard to clone while still easy to verify. Using the constructed oracle separation, can we work backwards and identify new cryptographic protocols that can be proven secure with respect to a classical oracle?

    \item \textbf{Boolean function analysis}~Liu, Mutreja, and Yuen~\cite{liu2024qma} identify an inherent connection between $\QMA$ vs $\QCMA$ and the Aaronson-Ambainis conjecture~\cite{aaronson2009need}, a major open question in boolean function analysis. There are two versions of the Aaronson-Ambainis conjecture: one in terms of the influence of boolean functions and the other in terms of quantum query algorithms. Our result does not directly address either version of the Aaronson-Ambainis conjecture, but we hope that it might offer a new interesting perspective with which the problem might be tackled. 
    
    \item \textbf{Matching upper bounds.} We've shown that spectral Forrelation is hard for $\QCMA$. However, we do not know what the optimal $\BQP$ or $\QCMA$ algorithms for this oracle are. Similarly, the problem of sampling points from $S$ given oracle access to $S$ and $U$ (or just $U$) seems like an interesting quantum query complexity problem in its own right.  What is the best sampling algorithm for producing $v$ problems in $S$ from access to both $S$ and $U$?

    \item \textbf{Proof improvements}~Lastly, we admit that many of the components of our proof are likely suboptimal. We made many decisions in our proof that could be modified to make the proof simpler. One particular decision that stands out is to include an element $y$ in the set $U$ with probability $1 - \e^{-\kappa \gamma_y^{(S)}}/2$. Is there a better choice of function that makes the proof simpler? We chose the exponential function since we could find a family of good polynomial approximations to the functions, and previously Zhandry~\cite{zhandry2024toward} used a similar function since it was amenable to integration. However, we do not know if this was optimal, as it caused other challenges. Other decisions, such as analyzing the bosonic system in terms of quasi-even condensates, could also be optimized. Is there a better generalization (instead of quasi-even condensates) for which we can prove sampling probability upper bounds? We suspect and hope that a simpler reformulation of this proof will be found. 
\end{enumerate}

\section{Outline of the paper}

This paper is broken up into multiple parts. In the remainder of~\Cref{part:introduction}, we introduce relevant notations, definitions, and formalize the definitions of the complexity classes $\QMA$, $\QCMA$, and oracle separations. We, however, defer introducing the quantum mechanics of bosons to~\cref{sec:prelim-bosons} of~\cref{part:sampler-upper-bound}. Next,~\Cref{part:qcma-to-sampler} and~\Cref{part:sampler-upper-bound,part:poly-query-implies-qec} play dual roles to each other, together proving the impossibility of an efficient query algorithm with a polynomial-sized witness for spectral Forrelation in the property-testing regime. 

\Cref{part:qcma-to-sampler} proves that an efficient $\QCMA$ algorithm for the spectral Forrelation problem implies a sampling probability \emph{lower bound} for a particular sampling task, namely sampling points from $S$ given oracle access to $U$, the heavy points of the Hadamard transform of $\ket{S}$.  The proof highly resembles the hybrid method originally used by \cite{BBBV}, with the key idea being to use the fact that a classical witness can be reused without degrading its quality to get arbitrarily many samples from $S$.

Then,~\Cref{part:sampler-upper-bound,part:poly-query-implies-qec} prove a strictly smaller sampling \emph{upper bound} for that same task.  \Cref{part:sampler-upper-bound} introduces a broad family of states, which we call quasi-even condensates, that (approximately) satisfy a large number of translational symmetries.  Using a combinatorial argument, we bound the maximum success probability of any sampler whose purifying register is a quasi-even condensate.  Finally, \Cref{part:poly-query-implies-qec} shows that after only making a few queries to $U$, the purifying register of every query algorithm will be supported on quasi-even condensates.  \Cref{part:poly-query-implies-qec} introduces a new compressed oracle technique for sparse oracles, and combines it with flat polynomial approximations to the exponential and tools from perturbation theory.  

Finally,~\Cref{part:concluding-remarks} brings together all the important results from the previous parts to prove a property testing oracle separation which can be lifted to a complexity class separation. We end~\Cref{part:concluding-remarks} with concluding remarks about how we came to this proof.
\section{Preliminaries}

\paragraph{Organizational notes}
Some of the proofs in this paper are naturally lengthy and require smaller technical/mathematical lemmas to prove. Notationally, we write these required lemmas in inline boxed environments. The intention is that a first pass reading of the paper can effectively skip over these lemmas by only reading the main listed lemmas/theorems in the boxed environments. 

\subsection{Mathematical notation}
\label{sec:math_notation}

The following notations are used in this work. Most are standard, but we reiterate them for the sake of clarity.
\begin{itemize}
\item $\norm{\cdot}$ or $\norm{\cdot}_{\mathrm{op}}$ for a matrix will refer to the operator norm of a matrix, unless specified otherwise, and $\norm{\cdot}_1$ denotes the matrix $1$-norm, $\norm{X}_1 = \Tr[\abs{X}]$. $\norm{\cdot}$ for a vector will refer to the $2$-norm unless specified otherwise.
\item When using the notation $\prod_{i = a}^{b} X_{i}$ for $X_i$ that do not commute, the product expands from left to right starting from $a$.  For example, when $a < b$, it expands as $X_{a} X_{a+1}\ldots X_{b}$.  When $b < a$, it expands as $X_{a} X_{a-1}\ldots X_{b}$.
\item Given a $x \in \{0, 1\}^{n}$, we define $1_{x}$ to be the length $2^{n}$ tuple with entries indexed by $n$-bit strings, with $1$ in the index corresponding to $x$, and $0$ elsewhere.
\item For two functions $f, g : \NN \rightarrow \RR$, we say that $f = O(g)$ and $g = \Omega(f)$ if $\exists~n_0 \geq 0$ and $c > 0$ such that $f(n) \leq c g(n)$ for all $n \geq n_0$. 
\item We employ the notation $\delta_{x,y}$ for the indicator function of $x = y$ and we also use $\delta(b)$ for the indicator function when $b$ is a boolean predicate such as $b = (x \in A)$.
\item The binomial function ${a \choose b}$ can be extended naturally to non-integer valued $a$. In this paper, we will only need to use $a = 1/2$. For $b > 0$,
\begin{equation}
    \binom{1/2}{b}
\defeq \frac{(1/2)(1/2 - 1)(1/2 - 2)\cdots(1/2 - b + 1)}{b!}\,.
\end{equation}
We also define $\binom{1/2}{0} \defeq 1$.  An important fact that we use is that the absolute value of this is always $< 1$ for $b > 0$.
\end{itemize}

\paragraph{Indexing}

Throughout most of this paper, we will be indexing over the set $\bits^n$; this set is isomorphic to the set $[N]$ where $N = 2^n$ under the lexicographic ordering. It is sometimes convenient to switch indexing sets between $\bits^n$ and $[N] \defeq \{0, \ldots, 2^n - 1\}$ for notational simplicity.

\subsection{Quantum query complexity}

\paragraph{Quantum query algorithms}
For this paper, we employ the following definitions of quantum query circuits/algorithms. These are standard definitions in the literature and are restated here for convenience. We use the term ``an oracle of size $n$'' to refer to a function $\Oo: \bits^n \rightarrow \bits$. This term matches notions from query complexity. We also use the phrase ``an oracle'' to refer to a function $\Oo: \bits^* \rightarrow \bits$ for a function mapping arbitrary strings to bits.

\begin{definition}[Quantum query circuit/algorithm]
For this paper, we define a \emph{quantum query algorithm} as an inputless quantum circuit that interacts coherently with a boolean oracle function $\Oo : \{0,1\}^n \to \{0,1\}$.
The oracle is accessed via a \emph{phase query gate}
$\ket{b, x} \;\longmapsto\; (-1)^{b\cdot \Oo(x)} \ket{b, x}$, which acts on an $n+1$-qubit query register.  The algorithm is described by an alternating sequence of unitaries (drawn from any fixed universal gate set) and oracle gates. Here $t$ denotes the \emph{number of oracle queries} made by the algorithm and $n$ denotes the size of the boolean oracle function. After all gates are applied, a designated output qubit is measured in the computational basis to determine acceptance.  
The initial state is assumed to be all qubits as $\ket{0}$ and all intermediate unitaries may act on an arbitrary (but finite) number of qubits. This definition naturally extends to multiple oracles (as will be the use case in this paper).
\end{definition}
This model captures the standard query-complexity viewpoint: there is no explicit input string—only access to the oracle $\Oo$ of a particular input size $n$. Furthermore, we don't place any restrictions on the complexity of the interleaved unitaries, and the overall circuit size is not of importance beyond the number of queries to $\mathcal{O}$.
We note that the phase query gate described here is equivalent to other standard notions of oracle access, including the bit-flip oracle $\ket{b, x} \;\longmapsto\; \ket{b \oplus \Oo(x), x}$.

We note that a pair of functions $(S, U): \bits^{n-1} \rightarrow \bits$ can be represented as a single function on $n$ bits by 
\begin{equation}
    \mathcal{O}(b, x) = \begin{cases}
        S(x)& \text{if}\quad b = 0\,,\\
        U(x)& \text{if}\quad b = 1\,.
    \end{cases} 
\end{equation}
Furthermore, controlled-access to the ``combined unitary'' can be constructed from controlled-access to both $S$ and $U$.
Therefore, for the remainder of this paper (and the statement of~\Cref{thm:main}), we use the notation of accessing two oracles as it is equivalent and notationally easier to digest.

\begin{definition}[Quantum query algorithm with witness]
A quantum query algorithm may also receive an auxiliary \emph{witness}. There are two main types of witnesses we consider.
\begin{itemize}
    \item  A \emph{quantum witness} is a state $\ket{\psi}$ on $q$ qubits.
    \item A \emph{classical witness} is a bit string $w \in \{0,1\}^q$, treated as a computational-basis state.
\end{itemize}
The definition is identical to the previous except the input state is now $\ket{\psi} \otimes \ket{0 \ldots 0}$ or $\ket{w} \otimes \ket{0 \ldots 0}$. The algorithm’s acceptance probability may depend on both the oracle and the witness.
\end{definition}

\paragraph{Quantum complexity theory}

The previous definitions of quantum query circuits/algorithms are defined in terms of a parameter $n$, which specifies the input length for oracle queries. Most of the paper will involve proving property-testing lower bounds for such algorithms. Only in the end will we need to reconcile these definitions with quantum complexity theory. \\

In complexity theory, we are interested in families of quantum algorithms that run on all possible input sizes. Notationally, $n$ is often used as the length of the input, which we want to classify as a yes or no instance of a decision problem. An oracular complexity class is one defined in terms of a function $\Oo: \bits^* \rightarrow \bits$ (or equivalently a family of functions $\Oo: \bits^n \rightarrow \bits$ for each integer length $n$). \\

To formally define the quantum complexity classes $\BQP^\Oo$, $\QCMA^\Oo$, and $\QMA^\Oo$, we will need to define the notion of a uniform quantum oracle algorithm. For the purposes of this result, we will only need to define $\P$-uniform quantum oracle algorithms. This is due to the result of Yao~\cite{yao-qc}, which proved that quantum complexity classes can be defined in terms of $\P$-uniform families of quantum circuits.

\begin{definition}[Quantum oracle circuits] We define a {family of quantum oracle circuits/algorithms} $\{\Aa_n\}_{n \ge 1}$,
where the index $n$ corresponds to the length of the explicit input to the computational problem.
Each $\Aa_n$ can be represented concretely as a quantum circuit consisting of:
\begin{itemize}
    \item elementary quantum gates drawn from a fixed, complete gate set (for example, Hadamard, phase, and controlled-NOT),
    \item \emph{oracle phase gates} providing coherent access to the Boolean oracles
    \begin{equation}
    \Oo_k : \{0,1\}^{k} \to \{0,1\}~\text{by}~\ket{b, x} \mapsto (-1)^{b\cdot \Oo_k(x)} \ket{b, x}~\text{for}~x \in \bits^k, b \in \bits\,.
    \end{equation}
\end{itemize}
This can be viewed as the circuit model definition of accessing $\Oo$ at various lengths. \\

The family $\{\Aa_n\}$ is $\P$-\emph{uniform} if there exists a deterministic polynomial-time Turing machine $M$ that, on input $1^n$, outputs a full classical description of the circuit $\Aa_n$.  
Because $M$ runs in time polynomial in $n$, the resulting circuit $\Aa_n$ must satisfy the following polynomial bounds for some polynomial functions:
\begin{itemize}
    \item it takes as input a classical input of size $n$ and consists of unitary gates followed by the measurement of a single qubit for a binary output,
    \item it contains at most $\poly(n)$ gates (either oracle or elementary), and
    \item the largest length $k$ which it can query $\Oo_k$ is at most $\poly(n)$.
\end{itemize}
This ensures that the circuit family represents an efficiently describable quantum algorithm operating within polynomial resources. 
\end{definition}

Using the definition of $\P$-uniform quantum oracle algorithms, we define the standard oracle quantum complexity classes. The natural extension of $\BQP$ is the following.
\begin{definition}[Oracle $\BQP$]
    A promise language $\Ll^\Oo = (\Ll_\yes, \Ll_\no)^\Oo \subseteq \{0,1\}^*$ is in $\BQP^{\Oo}$ if there exists a $\P$-uniform family of quantum oracle circuits $\Aa_n$ such that for every input $x$ of length $n = \abs{x}$,
\begin{xalign}
    \text{(Completeness)} & \quad x \in \Ll_\yes \implies \Pr[\Aa_n^\Oo(x)\text{ accepts}] \ge \tfrac{2}{3}, \\
    \text{(Soundness)} & \quad x \in \Ll_\no \implies \Pr[\Aa_n^\Oo(x)\text{ accepts}] \le \tfrac{1}{3}.
\end{xalign}
\end{definition}
 The pair of constants $2/3$ and $1/3$ is not important. Standard parallel repetition techniques can be used to choose different constants or even functions as close as $(1/2 + 1/\poly(n), 1/2 - 1/\poly(n))$ or as far apart as $(1 - 2^{-\poly(n)}, 2^{-\poly(n)})$. Next, we define the generalizations of $\QCMA$ and $\QMA$.
\begin{definition}[Oracle $\QCMA$]
    A promise language $\Ll^\Oo = (\Ll_\yes, \Ll_\no)^\Oo \subseteq \{0,1\}^*$ is in $\QCMA^{\Oo}$ if there exists a $\P$-uniform family of quantum oracle circuits $\Aa_n$ with $\Aa_n$ accepting a witness of length $q(n)$, such that for every input $x$ of length $n = \abs{x}$,
\begin{xalign}
    \text{(Completeness)} & \quad x \in \Ll_\yes \implies \exists~w \in \bits^{q(n)} \st \Pr[\Aa_n^\Oo(x,w)\text{ accepts}] \ge \tfrac{2}{3}, \\
    \text{(Soundness)} & \quad x \in \Ll_\no \implies \forall~\tilde w \in \bits^{q(n)},~\Pr[\Aa_n^\Oo(x, \tilde w)\text{ accepts}] \le \tfrac{1}{3}.
\end{xalign}
\end{definition}
\begin{definition}[Oracle $\QMA$]
    A promise language $\Ll^\Oo = (\Ll_\yes, \Ll_\no)^\Oo \subseteq \{0,1\}^*$ is in $\QMA^{\Oo}$ if there exists a $\P$-uniform family of quantum oracle circuits $\Aa_n$ with $\Aa_n$ accepting a witness of length $q(n)$, such that for every input $x$ of length $n = \abs{x}$,
\begin{xalign}
    \text{(Completeness)} & \quad x \in \Ll_\yes \implies \exists~\ket{\psi} \in \qubits{q(n)} \st \Pr[\Aa_n^\Oo(x,\ket{\psi})\text{ accepts}] \ge \tfrac{2}{3}, \\
    \text{(Soundness)} & \quad x \in \Ll_\no \implies \forall~\ket*{\tilde \psi}\in \qubits{q(n)},~\Pr[\Aa_n^\Oo(x, \ket*{\tilde \psi})\text{ accepts}] \le \tfrac{1}{3}.
\end{xalign}
\end{definition}

\subsection{Sets, multisets, and functions}

Throughout the paper, we use the following notation to refer to sets, multisets, and the functions associated with them.  
\begin{itemize}
    \item Given a multiset $S$, we use the notation $\{x_1, \ldots, x_{\ell}\}$ to refer to the \textit{unordered} collection of elements in $S$, where the elements $x_1, \ldots, x_{\ell}$ may be non-distinct.  
    \item We abuse notation and take $S(x)$ to be the indicator function $\delta(x \in S)$, which is $1$ if $x$ is in $S$ with any multiplicity, and $0$ otherwise.
    \item We will use the notation $\mathcal{O}_{S}$ to refer to the phase oracle for the function $S(\cdot)$, i.e., the unitary that acts as $\mathcal{O}_{S} \ket{b, x} = (-1)^{b \cdot S(x)} \ket{b, x}$.
    \item When applying operations between sets, like $S \setminus T$, we refer to the operation of first mapping $S$ and $T$ to the set of distinct elements in $S$ and $T$ respectively, and then outputting the operation applied to the resulting sets.  We also use the notation $T \subseteq S$ to mean that $T$ is a subset of the set of distinct elements of $S$.
    \item We also define a projector $\Pi_{S}$ for a multiset $S$ to be the projection onto basis states $x$ in $S$, treating $S$ as a set, i.e., $\sum_{x \in S} \ketbra{x}$.
    \item Throughout the paper, we use the terminology ``oracle access to a multiset $S$'' to mean query access to $\mathcal{O}_{S}$, and may write an algorithm querying $\mathcal{O}_{S}$ as $\mathcal{A}^{S}$ for brevity.  Algorithms may receive access to multiple multisets, in which case we write $\mathcal{A}^{(S_1, S_2)}$.  
    \item We will also write, for a multiset $S = \{s_1, \ldots, s_{\ell}\} \subseteq \bits^n$, the state $\ket{S} = \frac{1}{\sqrt{\ell}} \sum_{i = 1}^{\ell} \ket{s_i}$ in the Hilbert space $\qubits{n}$.  Note that when $S$ has no multiplicities greater than $1$, this is a normalized state, but otherwise it may be unnormalized.  This is as opposed to the state $\ket{\mathrm{tt}_{S}}$ in the Hilbert space $\CC^{\ZZ_{\ge 0}^{2^n}}$, the classical description of the multiset.  
\end{itemize}

\subsection{Spectral Forrelation}

\begin{definition}[Spectral Forrelation] We say two subsets $S, U \subset \bits^n$ are $\alpha$-spectrally Forrelated if
\begin{align}
    \alpha = \norm{\Pi_U\cdot H^{\otimes n}\cdot \Pi_S}_\mathrm{op}^2 = \max_{\norm{\ket{\psi}} = 1} \norm{\Pi_U\cdot H^{\otimes n}\cdot \Pi_S \ket{\psi}}^2 .
\end{align}
Here, the projectors $\Pi_{U}$ and $\Pi_{S}$ are the projectors onto the subspaces spanned by the basis vectors $\ket{x}$ for $x \in U$ and $x \in S$, respectively.
\end{definition}

\begin{remark}
    First note that this definition can be extended to multisets $S$ and $U$ by taking the set of elements included in $S$ and $U$.  In this way, spectral Forrelation is defined to be a property of pairs of multisets as well.

    This definition can also be extended to functions/oracles $S, U: \bits^n \rightarrow \bits$ by taking the sets $S$ and $U$ to be the pre-images of $1$ for the functions $S$ and $U$ respectively. In this way, spectral Forrelation is also defined as a property of a pair of oracles. 
\end{remark}

\begin{remark}
    The spectral Forrelation of any two oracles is a number $\in [0,1]$. By definition, if either oracle is the constant function $0$ (i.e., corresponds to the empty set), then the spectral Forrelation is $0$.
\end{remark}

In this paper, we will use the notation ``at least $\alpha$-spectrally Forrelated'' to mean that a pair of subsets are $\beta$-spectrally Forrelated for some $\beta \geq \alpha$, and similarly for ``at most $\alpha$-spectrally Forrelated''.  We will typically refer to a pair $(S, U)$ that is at least $59/100$-spectrally Forrelated as a yes instance of spectral Forrelation, and a pair $(S, U)$ that is at most $57/100$-spectrally Forrelated as a no instance.  

The following is implicit in \cite{zhandry2024toward}, but we re-prove it here.  

\begin{theorem}[$\QMA$~containment]
\label{thm:spectral_forrelation_qma_containment}
For any $\alpha > \beta$, there is a $O(1/(\alpha - \beta)^2)$ quantum query algorithm with a $n$-qubit quantum witness that, given oracle access to oracles $S, U: \bits^n \to \bits$, accepts with probability at least $2/3$ if they are  at least $\alpha$-spectrally Forrelated (yes instances), and probability at most $1/3$ if they are at most $\beta$-spectrally Forrelated  (no instances).
\end{theorem}

\begin{proof}
    There is a simple verifier that accepts yes instances with probability $\ge \alpha$ and accepts no instances with probability $\le \beta$:
    Let $\ket{\psi}$ be the $n$-qubit quantum witness. The following quantum circuit describes the verifier.

    $$
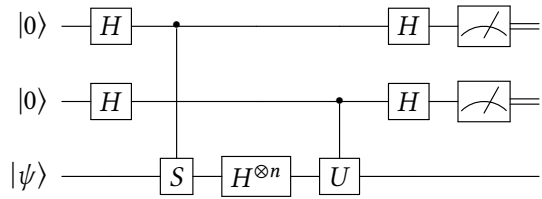
\Qcircuit @C=1em @R=1.2em {
    & \lstick{\ket{0}} & \gate{H} & \ctrl{2} & \qw      & \qw  &\gate{H}    & \meter & \cw \\
    & \lstick{\ket{0}} & \gate{H} & \qw      & \qw      & \ctrl{1} &\gate{H} & \meter & \cw \\
    & \lstick{\ket{\psi}} 
        & \qw
        & \gate{S} 
        & \gate{H^{\otimes n}} 
        & \gate{U} 
        & \qw & \qw & \qw
}
$$ ~

The verifier accepts if both measurements equal 1. It is easy to check that the probability that both measurements output 1 is exactly $\norm{\Pi_U \cdot H^{\otimes n} \cdot \Pi_S \ket{\psi}}^2$. By definition, there exists a $\ket{\psi}$ such that this probability is $\ge \alpha$ for yes instances and for all $\ket{\psi}$, this probability is $\le \beta$ for no instances. Using the Marriott-Watrous amplification~\cite{Marriott2005}, using $O(1/(\alpha - \beta)^2)$ queries, we can convert the verifier into one that decides spectral Forrelation with completeness-soundness of $(2/3, 1/3)$.
\end{proof}
Note that the previous circuit is effectively measuring the input witness with projectors $\Pi_S$ and $H^{\otimes n} \cdot \Pi_U \cdot H^{\otimes n}$. The Marriott-Watrous amplification protocol~\cite{Marriott2005} will consist of alternating these two measurements and using the classical outputs from the measurements to decide the problem.

\newpage
\part{From $\QCMA$ algorithms to samplers}
\label{part:qcma-to-sampler}

The goal of \Cref{part:qcma-to-sampler} is to show that there is a reduction from a $\QCMA$ algorithm that decides spectral Forrelation (i.e., whether two oracles of length $n$ are at least $59/100$-spectrally Forrelated or at most $57/100$-spectrally Forrelated, given a classical witness), to a sampler such that for all $(S, U)$ that are at least $59/100$-spectrally Forrelated, the sampler outputs many points from $S$, only querying $U$.  

\Cref{sec:qcma-to-one-sampler} will describe a general reduction for taking $\QCMA$ algorithms that query two oracles and decide between two families of (pairs of) oracles, and producing samplers that query one of the two oracles.  The reduction will work for families of oracles satisfying a special property. Given a family of yes and no instances defined on pairs of oracles $(S, U)$, we say that $(S, U)$ is a \emph{strong yes instance} if $(S, U)$ is a yes instance, and for any small $\Delta \subseteq S$, $(\Delta, U)$ is a no instance.  This section will show that any $\QCMA$ algorithm that distinguishes between families of pairs of yes and no oracle instances of this problem can be used to sample points from strong yes instances.  Finally, \Cref{sec:strong_yes} shows that if we define yes instances to be sets $(S, U)$ that are at least $59/100$-spectrally Forrelated, and no instances to be sets $(S, U)$ that are at most $57/100$-spectrally Forrelated, there is a procedure for sampling from a large family of strong yes instances.  These constants ($59/100$ and $57/100$) are artifacts of the proof, and not necessarily fundamental to spectral Forrelation.

\section{Constructing samplers from strong yes instances} \label{sec:qcma-to-one-sampler}

Assume the existence of a $t$-query quantum query algorithm $\mathcal{A}^{(S,U)}$ with $q$-bit classical witness which solves the spectral Forrelation problem. Formally, assume the quantum query algorithm $\mathcal{A}^{(S,U)}$ has the following properties:
\begin{enumerate}
    \item (Completeness) If $(S,U)$ is at least $59/100$-spectrally Forrelated, there exists a witness $w \in \bits^q$ such that $\mathcal{A}^{(S,U)}(w)$ accepts with probability at least $2/3$. 
    \item (Soundness) If $(S,U)$ is at most $57/100$-spectrally Forrelated, then for all witnesses $\tilde w \in \bits^q$, $\mathcal{A}^{(S,U)}(\tilde w)$ accepts with probability at most $1/3$.
\end{enumerate}
In this section, we describe pairs $(S, U)$ that are at least $59/100$-spectrally Forrelated as yes instances of spectral Forrelation, and pairs $(S, U)$ that are at most $57/100$-spectrally Forrelated as no instances of spectral Forrelation.  

As described in the preliminaries, oracle access to the multisets $S$ and $U$ can be described as the linear extensions of the following maps
    \begin{xalign}
        \ket{b, x}\ket{z} &\overset{\mathcal{O}_S}{\mapsto} (-1)^{b \cdot S(x)} \ket{b, x}\ket{z}\,, \\
        \ket{b, x}\ket{z} &\overset{\mathcal{O}_U}{\mapsto} (-1)^{b \cdot U(x)} \ket{b, x}\ket{z}\,.
\end{xalign}
where $b$ is a control bit, $x$ describes the input to the oracle, and $z$ describes the state of the remainder of the system. The algorithm $\mathcal{A}^{(S,U)}$ can be thought of starting from a state $\ket{w}\ket{0\ldots0}$ and applying a sequence of general unitaries interlaced with $t$ queries to $S$ and $t$ queries to $U$. The algorithm $\mathcal{A}^{(S, U)}$ concludes by measuring the first qubit in the standard basis.

\subsection{Sampling from $S$} 
We can treat the algorithm $\mathcal{A}^{(S, U)}$ as $(\mathcal{A}^{U})^{S}$ --- i.e., an algorithm $\mathcal{A}^{U}$, consisting of standard unitary gates, as well as $U$ gates, that makes queries the $S$ oracle.  When expressed this way, the state of the algorithm $(\mathcal{A}^{U})^{S}$ immediately before its final measurement is given by
\begin{equation}
\label{eq:qcma_algorithm_state_sampler}
    V_t \mathcal{O}_S V_{t-1} \mathcal{O}_S \ldots \mathcal{O}_S V_0 \ket{w, 0}
\end{equation}
with the unitaries $\{V_j\}$ including queries to $U$.  Consider the following algorithm, which takes as input a witness $w$ and a set $\Delta \subseteq \bits^{n}$. The algorithm $\mathsf{Sampler}^U(w, \Delta)$ is a sampler for generating an additional sample from $S$ given a prior subset $\Delta$ of found points.  Recall that for a multiset $S$ and set $\Delta$, we are taking $S \setminus \Delta$ to be the set of elements that appear with multiplicity at least $1$ in $S$ that do not appear in $\Delta$.  

\begin{mathinlay}
    \textbf{Query algorithm} $\mathsf{Sampler}^U(w, \Delta)$:
    \begin{enumerate}
            \item Sample $j \leftarrow \{0, \ldots, t-1\}$ uniformly randomly.
            \item Compute the state $V_j \mathcal{O}_{\Delta} V_{j-1} \mathcal{O}_{\Delta} \ldots \mathcal{O}_{\Delta} V_0 \ket{w,0}$, where $\mathcal{O}_{\Delta}$ is the unitary defined by
            \begin{equation}
                \ket{b,x}\ket{z} \overset{\mathcal{O}_{\Delta}}{\mapsto} (-1)^{b\cdot\delta(x \in \Delta)} \ket{b,x}\ket{z}\,.
            \end{equation}
            \item Measure the state in the standard basis for output $(b, x,z)$.
            \item If $x \in \Delta$, output the alphabetically first symbol not in $\Delta$. Else, output $x$.
    \end{enumerate}
\end{mathinlay}

A minor adjustment can be made to ensure that $\mathsf{Sampler}$ makes exactly $t$ queries to $U$. At a high level, the idea behind the sampler is that for $\mathcal{A}^{U}$ to distinguish between $S$ and $\Delta$, it must be querying points in $S$ that are not in $\Delta$, since the oracles $f_{\Delta}$ and $S$ are identical outside of those points.  We prove the following claim about $\mathsf{Sampler}$. The proof closely resembles the hybrid method of \cite{BBBV}.

\begin{lemma}
\label{lem:repeated-sampler}

Let $(S,U)$ be a yes instance of the oracle problem (i.e., at least $59/100$-spectrally Forrelated).  Let $\Delta\subseteq S$ be a subset of the set of elements in $S$ such that $(\Delta, U)$ is a no instance (i.e., at most $57/100$-spectrally Forrelated). Let $w$ be a witness that causes the query algorithm $\mathcal{A}$ to accept $(S, U)$ with probability at least $2/3$. Then, the algorithm $\mathsf{Sampler}^U(w, \Delta)$ makes $t$ queries to the oracle $U$, no queries to the oracle $S$, and produces a sample from $S \setminus \Delta$ with probability at least $ \left(\frac{1}{36t^2}\right)$.

\end{lemma}
\begin{proof}
    We begin by defining a sequence of hybrids.  Recall that the verification algorithm's state immediately before its final measurement is given by \eqref{eq:qcma_algorithm_state_sampler}.  Let the $j$-th hybrid state $\ket*{h_{j}(w)}$ be define as
    \begin{equation}
    \ket*{h_j(w)} \defeq V_{t} \mathcal{O}_{S} \ldots \mathcal{O}_{S} \underbrace{V_{j} \mathcal{O}_\Delta V_{j-1} \mathcal{O}_\Delta \ldots V_1 \mathcal{O}_\Delta V_0 \ket{w,0}}_{\defeq\  \ket{\psi_j(w)}}\,,
\end{equation}
where we call the state $\ket{\psi_{j}(w)}$ the $j$-th prefix state.  Intuitively, the prefix state $\ket{\psi_{j}(w)}$ corresponds to running the algorithm $\mathcal{A}$ with oracles $(\Delta, U)$ until the $(j+1)$-th query, and the hybrid state $\ket*{h_j(w)}$ corresponds to replacing the first $j$ queries to the $S$ oracle with queries to the $\Delta$ oracle.  

Then, $\ket*{h_{0}(w)}$ corresponds to running $\mathcal{A}$ on $(\Delta, U)$ up until the final measurement and $\ket*{h_t(w)}$ corresponds to running $\mathcal{A}$ on $(S, U)$ up until the final measurement.  Since $(\Delta, U)$ is a no instance of spectral Forrelation and $(S, U)$ is a yes instance of spectral Forrelation, there is a measurement that accepts $\ket*{h_0(w)}$ with probability at most $1/3$ and accepts $\ket*{h_t(w)}$ with probability at least $2/3$, namely to measure the first qubit in the computational basis.  Therefore, we have that
\begin{xalign}
    \frac{1}{3} &\leq \Tr\left[\ketbra{0} \left(\ketbra{h_t(w)} - \ketbra{h_0(w)}\right)\right]\\
    &\leq \frac{1}{2} \norm{\ketbra{h_t(w)} - \ketbra{h_0(w)}}_{1}\\
    &\leq \norm{\ket*{h_t(w)} - \ket*{h_0(w)}}\,.
\end{xalign}
A proof of these inequalities can be found in e.g.~\cite{wilde2016quantum}.  By the triangle inequality, we have that
\begin{xalign}
    \frac{1}{3} &\leq \norm{\ket{h_t(w)} - \ket{h_0(w)}} \\
    &\leq \sum_{j = 1}^t \norm{\ket*{h_j(w)} - \ket*{h_{j-1}(w)}} \\
    &\leq \sum_{j = 0}^{t-1} \norm{ (\mathcal{O}_{\Delta} - \mathcal{O}_S) \ket*{\psi_j(w)}} \\
    &= 2 \sum_{j = 0}^{t-1} \norm{ \Gamma \ket*{\psi_j(w)}}
\quad\text{where}~\Gamma \defeq  \sum_{x \in S \setminus \Delta} \ketbra{1, x} \otimes \id\,.
\end{xalign}
Here, we use the fact that the $j$ and $(j+1)$'st hybrids differ only by the oracle query immediately after the $j$-th prefix state, and we recall that $S$ is the controlled phase flip oracle that applies a sign of $-1$ when $b = 1$ and $x \in S$, so $\mathcal{O}_{\Delta} - \mathcal{O}_{S}$ is exactly twice the projector onto strings of the form $\ket{1, x}$ for $x \in S \triangle \Delta = S \setminus \Delta$ as $\Delta \subseteq S$.  Here, recall that operations like set minus and symmetric difference act on $S$ and $\Delta$ as if they were sets.  For each $j \in \{0, \ldots, t-1\}$, express
\begin{equation}
    \ket*{\psi_j(w)} = \sum_{x} \beta_{x}^{(j)} \ket{x} \otimes \ket*{\psi_j(x,w)}~\text{for}~\beta_{x} \in \mathbb{R}^{+}. \label{eq:psi-decomposition}
\end{equation}
Here, the control bit $b$ is captured in $\ket{\psi_{j}(x, w)}$.  Then, using the Cauchy-Schwarz inequality, we have that,
\begin{equation}
    \frac{1}{6} \leq \sum_{j = 0}^{t-1} \sqrt{\sum_{x \in S \setminus \Delta} (\beta_{x}^{(j)})^2} \leq \sqrt{t} \sqrt{\sum_{j = 0}^{t-1} \sum_{x\in S \setminus \Delta} (\beta_{x}^{(j)})^2}\,.
\end{equation}
This implies that,
\begin{equation}
    \frac{1}{36t^2} \leq \frac{1}{t} \sum_{j = 0}^{t-1} \sum_{x\in S \setminus \Delta} (\beta_{x}^{(j)})^2\,. \label{eq:prob-of-output-coplants}
\end{equation}
The right hand side is exactly the probability that $\mathsf{Sampler}^{U}(w, \Delta)$ outputs a $x$ from $S \setminus \Delta$. 
\end{proof}

\subsection{Witness-free sampler} 
The algorithm $\mathsf{Sampler}^U(w, \Delta)$ requires a witness $w$ to run. Furthermore, conditioned on $\mathsf{Sampler}^U(w, \Delta)$ outputting a novel point from $S$, we can iterate the sampler to generate multiple points. This generates a notion of a Cumulative Sampler
which also requires a witness. However, we can simply guess the witness $w$ at a cost to the success probability of the algorithm:

\begin{mathinlay}
    \textbf{Query algorithm} $\mathsf{CumulativeSampler}^U$:
    \begin{enumerate}
        \item Sample a random $w \in \bits^q$.
        \item Initialize $\Delta \leftarrow \emptyset$.
       \item For $v$ rounds, run $\mathsf{Sampler}(w, \Delta)$ with fresh independent randomness. Append output $x$: $\Delta \leftarrow \Delta \cup \{x\}$.
       \item Output the resulting $\Delta$.
    \end{enumerate}
\end{mathinlay}

Notice that the algorithm $\mathsf{CumulativeSampler}^U$ takes no witness as an input. Meaning the same sampler is producing samples from $S$ with the aforementioned probability for \emph{every} yes instance $(S,U)$ such that $(\Delta, U)$ are no instances for small subsets $\Delta$. We formalize this observation in the following theorem.

\begin{theorem}[Good samplers from $\QCMA$ algorithms]
\label{thm:formal-removal-of-S}
Assume there exists a quantum query algorithm $\mathcal{A}$ with classical witness for spectral Forrelation instances such that for instances of size $n$, $\mathcal{A}$ takes a $q$-sized classical witness and makes $t$ oracle queries. Let $(S,U)$ be a yes instance of spectral Forrelation and $v \in \NN$ be such that $(\Delta, U)$ is a no instance of spectral Forrelation for all subsets $\Delta \subset S$ with $\abs{\Delta} \leq v$. There exists a query algorithm $\mathsf{CumulativeSampler}$ (with implicit dependence on $v$) such that for $\mathsf{CumulativeSampler}^U$ makes 0 queries to $S$, $vt$ queries to $U$ and produces $v$ unique samples from $S$ with probability at least
\begin{equation}
    \ge 2^{-q} \cdot \left(\frac{1}{36t^2}\right)^v\,. \label{eq:formal-prob-bound}
\end{equation}
\end{theorem}
\begin{proof}
    Define the following events with respect to the running of $\mathsf{CumulativeSampler}$. Let $G$ be the event that the witness $w$ sampled is a good witness for the pair $(S,U)$. Second, let $E_1 \ldots, E_v$ be the events that the corresponding guesses are in $S$. By construction, the samples are distinct. The probability they are all correct is
    \begin{xalign}
        \Pr[E_v, \ldots, E_1] &\geq \Pr[G] \Pr[E_v, \ldots, E_1 | G ]  \\
        &\geq \Pr[G] \cdot \prod_{j = 1}^v \Pr[E_j | E_{j-1}, \ldots, E_1, G] \\
        &\geq 2^{-q} \cdot \left(\frac{1}{36t^2}\right)^v\,.
    \end{xalign}
    Here, we apply \Cref{lem:repeated-sampler}, which tells us that the probability that the $j$-th sample is in $S$, conditioned on the first $j-1$ samples being in $S$ is at least $\frac{1}{36t^2}$.  
\end{proof} %
\section{Strong yes instances for spectral Forrelation}

\label{sec:strong_yes}
 \Cref{thm:formal-removal-of-S} demonstrates that there exists a sampler algorithm with good success probability that, when given oracle access to the $U$ oracle in a yes instance $(S,U)$ of spectral Forrelation, with the additional property that $(\Delta, U)$ is a no instance of spectral Forrelation for all subsets $\Delta \subset S$ such that $\abs{\Delta} \leq v$, outputs $v$ points from $S$. We define such yes instances as \emph{strong} yes instances.
\begin{definition}[Strong yes instance]
    For any $t_1,t_2,v$ with $t_1<t_2$, we will say that a pair $(S,U)$ is a $(t_1,t_2,v)$-strong yes instance if
    \begin{enumerate}
        \item (completeness): $(S,U)$ are at least $t_2$-spectrally Forrelated.
        \item (soundness): For any subset $\Delta \subset S$ such that $\abs{\Delta} \leq v$, $(\Delta, U)$ are at most $t_1$-spectrally Forrelated.
    \end{enumerate}
\label[definition]{def:strong-instance}
\end{definition}
Intuitively, disproving the consequence derived in~\Cref{thm:formal-removal-of-S} will require constructing many strong yes instances and proving that no small query sampler can be successful for all these yes instances. To create yes instances, we want to sample a random multiset $S$ of size $\ell$, and then construct (with high probability) a set $U$ that is Forrelated with $S$. For this, define
    \begin{equation}
        \gamma_y^{(S)} = \qty(\frac{1}{\sqrt{\ell}} \sum_{i \in [\ell]} (-1)^{y \cdot s_i})^2 = \frac{1}{\ell}\sum_{i,j} (-1)^{y \cdot(s_i + s_j)} = 1 + \frac{1}{\ell} \sum_{i \neq j} (-1)^{y \cdot (s_i + s_j)}\,. \label{eq:def-of-gamma-s}
    \end{equation}
Observe that $\gamma_y^{(S)}$ equals $2^n \cdot \abs{\mel{y}{H^{\otimes n}}{S}}^2$, where $\ket{S}$ is the superposition over $S$, weighted by the multiplicities of the elements, divided by $\sqrt{\ell}$.  When $S$ is a set (i.e., has no multiplicities that are not $0$ or $1$), this is a normalized state, but if it is a multiset it may not be normalized. The following lemma proves that there exists a distribution that is overwhelmingly supported on strong yes instances. In conjunction with~\Cref{thm:formal-removal-of-S}, it gives that a $\QCMA$ algorithm implies a sampler for a particular distribution. In the next section (\Cref{sec:bosons}), we will prove that such a sampler cannot exist.

\begin{definition}[The $\mathsf{Strong}$ distribution]\label[definition]{def:def-of-strong-dist}
    Let the distribution $\mathsf{Strong}_{\kappa}$ over pairs $(S,U)$ be defined by
    \begin{enumerate}[label=(\alph*)]
        \item sampling the multiset $S = \{s_1, \ldots, s_\ell\}$ of size $\ell$ by sampling each $s_i$ uniformly randomly from $\bits^n$,
        \item followed by, sampling a set $U$ by adding each point $y \in \bits^n \setminus \{0^n\}$ to $U$ with probability $1 - \e^{-\kappa\gamma_y^{(S)}}/2$ where $\gamma_y^{(S)}$ is defined in~\eqref{eq:def-of-gamma-s}. 
    \end{enumerate}
    The distribution $\mathsf{Strong}= \mathsf{Strong}_\kappa$ can also be seen as a distribution over oracle pairs $(S,U)$ by considering the indicator oracles for the multisets.
\end{definition}
Note that in the distribution $\mathsf{Strong}_{\kappa}$, $U$ is always a set, even though $S$ may be have multiplicities.  
Second, observe that the point $y = 0^n$ has a corresponding probability of $1 - \e^{\kappa \ell}/2$ which is an invariant irrespective of the instance. Therefore, querying $y$ provides no information for a verification algorithm about the instance that it couldn't have achieved by flipping a coin. 
\begin{lemma}[The strong yes property]\label[lemma]{lem:strong-instance}
    For all $\kappa \in [0, 1]$, $\rho \geq 0$ such that $t_1 < t_2$ in~\eqref{eq:t1-t2-def}, a pair $(S,U)$ sampled $\sim \mathsf{Strong}_{\kappa}$ is a $(t_1,t_2,v)$-strong yes instance, except with probability at most \begin{equation}\ell^6 2^{-n} +2\ell^2 \exp\qty(-\frac{\rho^2 2^n}{2\ell^2}),\end{equation} where
    \begin{xalign}[eq:t1-t2-def]
        t_1&=\frac{1+\kappa}{2}+\frac{v}{\ell}+\rho\,,\\
        t_2&=\frac{1+3\kappa}{2}-\frac{15\kappa^2}{4}-\frac{5\kappa}{\ell}-\rho\,.
    \end{xalign}
\end{lemma}
 In particular, if $500 \le \ell \ll 2^{n/6}$, by setting $\rho = \frac{2\ell^2}{2^{n}} \ln\left(\frac{2^{n}}{2\ell^{4}}\right)$ such that $2\ell^2 \exp\qty(-\frac{\rho^2 2^{n}}{2\ell^2}) = \ell^6 2^{-n}$ and choosing $\kappa=1/10$, and we have that except with probability at most $2\ell^6 2^{-n}$, $(S,U)$ is a $(57/100,59/100,\ell/100)$-strong yes instance\footnote{We emphasize that the constants here are not particularly important; any $t_2 > t_1 + 1/\poly(n)$ would have sufficed.}.   Further note that for some choices of $\kappa, \rho$, for example, when $t_1 \geq 1$ and $t_2 \leq 0$, the lemma trivially holds, and when $t_1 \geq t_2$, it is not very useful.   
Before we give the proof of~\Cref{lem:strong-instance}, we remark on the choice of adding $y$ with probability a function of $\gamma_y^{(S)}$.

\begin{remark}
    It is important that the oracle $U$ is constructed to have $y$ added with probabilities that are functions of $\gamma_y^{(S)} = 2^{n} \abs{\mel*{y}{H^{\otimes n}}{S}}^2$.  For certain choices of $U$ where the probability of including $y$ in $U$ is a function of $\mel*{y}{H^{\otimes n}}{S}$, or in particular its sign, there are ways to synthesize the state $\ket{S}$ using oracle access to $U$, see e.g.~\cite{irani2021quantum}.  
\end{remark}

\begin{proof}[Proof of~\Cref{lem:strong-instance}] For a list $S=\{s_1,\ldots,s_\ell\}$ and a positive integer $k$, we say an equation $s_{i_1}\oplus\ldots\oplus s_{i_{2k}}=0^n$ is ``trivial'' if the index multiset $\{i_1,\ldots,i_{2k}\}$ contains each index with even multiplicity. Notice that such trivial equations hold independent of the values of $s_i$. We say that $S$ is ``good'' if for $k=1,2$ or $3$, the only identities $s_{i_1}\oplus\ldots\oplus s_{i_{2k}}=0^n$ that hold are the trivial identities.
\begin{mathinlay}
\begin{claim}Except with probability $\ell^6/2^n$ over the choice of random $S$ of size $\ell$, $S$ is good. 
\end{claim}
\begin{proof}There are at most $\ell^6+\ell^4+\ell^2$ equations we need to consider. Since we only need to handle non-trivial equations, it is actually straightforward to bound the number of non-trivial equations of $2,4,6$ terms by the number of ways to choose $2$, $4$ or $6$ elements from $S$, which we can then upper bound as ${\ell \choose 6} + {\ell \choose 4} + {\ell \choose 2} \leq \ell^6$. For each non-trivial equation, over the choice of random $S$, the value of $s_{i_1}\oplus\cdots\oplus s_{i_{2k}}$ is just a random element in $\{0,1\}^n$. As such, the probability that the non-trivial equation holds is exactly $2^{-n}$. Union-bounding over all non-trivial equations gives the claim.
\end{proof}
\end{mathinlay}
Since the probability that $S$ is not good is very small, we will only calculate the probability that $(S,U)$ is a strong yes instance conditioned on $S$ being good.
For a matrix $M$ and a subset $\Delta$, let $M_{[\Delta]}$ denote the principal submatrix of $M$ obtained by discarding all rows and columns outside of $\Delta$.
Let $M^{S,U}=\Pi_{S} \cdot H^{\otimes n} \cdot \Pi_{U} \cdot H^{\otimes n} \cdot \Pi_{S}$. Notice that all the rows and columns of $M^{S,U}$ indexed by $x\notin S$ will be 0, so with a slight abuse notation, we will think of $M^{S,U}$ as the the sub-matrix of $H^{\otimes n} \cdot \Pi_{U} \cdot H^{\otimes n}$.

Completeness then corresponds to showing that the largest eigenvalue of $M^{S,U}$ is $\ge t_2$. To prove this, it suffices to construct a state $\ket{\psi}$ such that $\ev{M^{S,U}}{\psi} \ge t_2$. We use $|\psi\rangle$ equaling the uniform superposition over $S$. Soundness, on the other hand, corresponds to proving that $\|M^{S,U}_{[\Delta]}\|\leq t_1$ for any subset $\Delta\subseteq S$ s.t. $\abs{\Delta} \leq v$.

Let $M^S=\Exp_U[M^{S,U}]$ where the expectation is over $U$ sampled based on $S$. We will use concentration statements to relate the maximal eigenvalues of $M^{S,U}$ and its principal submatrices to those of $M^S$, which allows us to focus on $M^S$.

\begin{mathinlay}
\begin{claim}\label[claim]{claim:expectedMentries} For all $S$, $\displaystyle \Pr_U\left[\max_{x,x'}\left|M^S_{x,x'}-M^{S,U}_{x,x'}\right|\right]> \rho]\leq 2\ell^2 \e^{-\rho^2 2^n/2}$.
\end{claim}
\begin{proof}Fix $S$. For $y\in\{0,1\}^n$, let $U_y$ denote the random variable which is 1 if $y\in U$ and 0 otherwise. For $x,x'\in S$, write 
\begin{xalign}
    M^{S,U}_{x,x'}&=\frac{1}{2^n}\sum_{y\in U}(-1)^{(x\oplus x')\cdot y}=\frac{1}{2^n}\sum_{y\in\{0,1\}^n}(-1)^{(x\oplus x')\cdot y}U_y\,,\\
    M^S_{x,x'}&=\Exp_U[M^{S,U}_{x,x'}]=\frac{1}{2^n}\sum_{y \in \bits^{n}}(-1)^{(x\oplus x')\cdot y}\Exp_U[U_y]\,.
\end{xalign}
Observe that, since $S$ is fixed, $M^{S,U}_{x,x'}$ is the sum of $2^n$ independent (but not identical) random variables with each $\in \{\pm 2^{-n}\}$. The independence lets us apply Hoeffding's inequality, showing that
\begin{equation}
\Pr_u\qty[\abs{M^{S,U}_{x,x'}-M^S_{x,x'}}\geq\rho]\leq 2\e^{-\rho^2 2^n/2} \,.
\end{equation}
\Cref{claim:expectedMentries} follows by union-bounding over all $\ell^2$ pairs $(x,x') \in S^2$.
\end{proof}
\begin{claim}\label[claim]{claim:expectedMeigenvalue}For any $S$, except with probability at most  $2\ell^2 \e^{-\rho^2 2^n/2\ell^2}$ over the choice of $U$, we have that:
\begin{xalign}
    &\|M^{S,U}\|\geq \|M^S\| - \rho\,, \\
    &\text{and for any subset}~\Delta\subseteq S,~\|M^{S,U}_{[\Delta]}\|\leq \|M^S_{[\Delta]}\| +\rho\,.
\end{xalign}
\end{claim}
\begin{proof}We invoke \Cref{claim:expectedMentries} with $\rho'=\rho/\ell$ to bound $|M^{S,U}_{x,x'}-M^S_{x,x'}|\leq \rho'=\rho/\ell$ for all $x,x'$, except with probability $2\ell^2 \e^{-\rho^2 2^n/2\ell^2}$. Under this condition, by the Gershgorin circle theorem\footnote{ The Gershgorin circle theorem states that for a complex $W = (W_{xy})$ matrix, every eigenvalue of $W$ lies in the union of the discs $D(W_{xx}, R_x)$ where $R_x = \sum_{y: y \neq x} \abs{W_{xy}}$.},  $\|M^{S,U}-M^S\|_{\mathrm{op}}\leq \ell\rho'=\rho$ and $\|M^{S,U}_{[\Delta]}-M^S_{[\Delta]}\|_{\mathrm{op}}\leq |\Delta|\rho'\leq \ell\rho'=\rho$ for any subset $\Delta\subseteq S$. \Cref{claim:expectedMeigenvalue} then follows by the triangle inequality.
\end{proof}
\end{mathinlay}
 It remains to lower-bound $\norm{M^S}_{\mathrm{op}}$ and upper-bound $\norm{M^S_{[\Delta]}}_{\mathrm{op}}$. To do this, we compute approximate matrices in PSD ordering.
Define
\begin{xalign}
 A^S& \defeq H^{\otimes n} \cdot {\sf Diag}\left[\left(\frac{1}{2}+\frac{\kappa}{2}\gamma_y^{(S)}-\frac{\kappa^2}{4}(\gamma_y^{(S)})^2\right)_y\right] \cdot H^{\otimes n}\,,\\
  B^S& \defeq H^{\otimes n} \cdot {\sf Diag}\left[\left(\frac{1}{2}+\frac{\kappa}{2}\gamma_y^{(S)}\right)_y\right] \cdot H^{\otimes n}\,.
\end{xalign}
Here, the notation $\mathsf{Diag}[\left(f(y)\right)_{y}]$ denotes the diagonal matrix where for all $y \in \bits^{n}$, the $(y, y)$'th entry of the matrix is $f(y)$.  
\begin{mathinlay}
\begin{claim}\label[claim]{claim:psdsandwich}For any $S$ and any subset $\Delta\subseteq\{0,1\}^n$ (including $\Delta=\{0,1\}^n$), we have that $A^S_{[\Delta]}\preccurlyeq M^S_{[\Delta]} \preccurlyeq B^S_{[\Delta]}$,
and in particular $\|A^S_{[\Delta]}\| \leq \|M^S_{[\Delta]}\| \leq \|B^S_{[\Delta]}\| $.
\end{claim}
\begin{proof}Recall that $\Exp_U[\Pi_U]={\sf Diag}[(1-\frac{1}{2}\e^{-\kappa\gamma_y^{(S)}})_y]$. We can use the small Taylor series expansions of $\e^{-x}$ to bound, for all non-negative real $x$,
\begin{equation}\frac{1}{2}+\frac{\kappa}{2}x-\frac{\kappa^2}{4}x^2\leq 1-\frac{1}{2}\e^{-\kappa x}\leq \frac{1}{2}+\frac{\kappa}{2}x
\end{equation}
for all non-negative real $x$. For diagonal matrices, as PSD ordering is equivalent to the ordering of the diagonal entries:
\begin{equation}
    {\sf Diag}\left[\left(\frac{1}{2}+\frac{\kappa}{2}\gamma_y^{(S)}-\frac{\kappa^2}{4}(\gamma_y^{(S)})^2\right)_y\right]\preccurlyeq {\sf Diag}\left[\left(1-\frac{1}{2}\e^{-\kappa\gamma_y^{(S)}}\right)_y\right]\preccurlyeq {\sf Diag}\left[\left(\frac{1}{2}+\frac{\kappa}{2}\gamma_y^{(S)}\right)_y\right]\,.
\end{equation}
As PSD ordering is preserved under transformations $M\mapsto C^\dagger MC$, it follows that $A^S\preccurlyeq M^S \preccurlyeq B^S$. As PSD ordering is also preserved for all principal submatrices, this proves \Cref{claim:psdsandwich}.
\end{proof}
\end{mathinlay}
Now we are ready to bound the maximal eigenvalues of (principal submatrices of) $A^S$ and $B^S$ and, in turn, the maximal eigenvalues of $M^S$. 
Recall that for all good $S$, all $\ell$ elements in the multiset $S$ are distinct and so are the elements in the sumset $S\oplus S = \{x \oplus y : x, y \in S. x \neq y\}$.

\noindent \paragraph{Bounding the top eigenvalue of $B^S$.} For $x \in \bits^n$, let $\delta_x$ be the indicator function for $x=0^n$. Observe that:
\begin{xalign}
    B^S_{x,x'}&=\frac{1}{2^n}\sum_y (-1)^{(x\oplus x')\cdot y}\left(\frac{1}{2}+\frac{\kappa}{2}\gamma_y^{(S)}\right)\\
    &=\frac{1}{2}\left(\frac{1}{2^n}\sum_y (-1)^{(x\oplus x')\cdot y}\right)+\frac{\kappa}{2\ell}\left(\frac{1}{2^n}\sum_y (-1)^{(x\oplus x')\cdot y}\sum_{x_0,x_0'\in S}(-1)^{(x_0\oplus x_0')\cdot y}\right)\\
    &=\frac{1}{2}\delta_{x\oplus x'}+\frac{\kappa}{2\ell}\sum_{x_0,x_0'\in S}\delta_{x\oplus x'\oplus x_0\oplus x_0'}\,.
\end{xalign}
For diagonal entries $x=x'$, $\delta_{x\oplus x'}=1$. Moreover, for diagonal entries $\delta_{x\oplus x'\oplus x_0\oplus x_0'}=0$ unless $x_0=x_0'$; there are exactly $\ell$ such cases since $S$ is good. Thus, the diagonal entries all equal $B^S_{x,x}=\frac{1}{2}+\frac{\kappa}{2}$.
Meanwhile, for off-diagonal entries $x\neq x'$, $\delta_{x\oplus x'}=0$. Moreover, because $S$ is good, $\delta_{x\oplus x'\oplus x_0\oplus x_0'}=0$ except for the two cases $(x_0,x_0')=(x,x')$ or $(x_0',x_0)=(x,x')$. Thus, the off diagonal entries all equal $B^S_{x,x'}=\frac{\kappa}{\ell}$.
Observe, then, that the matrix $B^S_{[\Delta]}$ can be more succinctly expressed as
\begin{equation}
    B^S_{[\Delta]} = \qty( \half + \frac{\kappa}{2} - \frac{\kappa}{\ell}) \id_\Delta + \frac{\kappa \abs{\Delta}}{\ell} \ketbra{\Delta}
\end{equation}
where $\ket{\Delta} = \frac{1}{\sqrt{\Delta}} \sum_{x \in \Delta} \ket{x}$, the uniform superposition over $\Delta$. $\ket{\Delta}$ is thus the top eigenvector.
Therefore, 
\begin{equation}
    \norm{B^S_{[\Delta]}}_{\mathrm{op}} = \ev{B^S_{[\Delta]}}{\Delta} \leq \frac{1 + \kappa}{2} + \frac{\abs{\Delta}\kappa}{\ell}\,.
\end{equation}

\noindent \paragraph{Bounding the top eigenvalue of $A^S$.} This term is slightly more complicated due to the quadratic term in $\gamma_y^{(S)}$. We have:
\begin{xalign}
    A^S_{x,x'}&=\frac{1}{2^n}\sum_y (-1)^{(x\oplus x')\cdot y}\left(\frac{1}{2}+\frac{\kappa}{2}\gamma_y^{(S)}-\frac{\kappa^2}{4}(\gamma_y^{(S)})^2\right)\\
    &=B_{x,x'}^S-\frac{\kappa^2}{4\ell^2}\left(\frac{1}{2^n}\sum_y \sum_{x_0,x_1,x_0',x_1'}(-1)^{(x\oplus x')\cdot y}(-1)^{(x_0\oplus x_0')\cdot y}(-1)^{(x_1\oplus x_1')\cdot y}\right)\\
    &=B_{x,x'}^S-\frac{\kappa^2}{4\ell^2}\sum_{x_0,x_1,x_0',x_1'}\delta_{x\oplus x'\oplus x_0\oplus x_0'\oplus x_1\oplus x_1'}\,.
\end{xalign}
We have already evaluated $B_{x,x'}^S$, and instead focus on evaluating the final term. For diagonal entries $x=x'$, since $S$ is good, there are only two ways for $\delta_{x\oplus x'\oplus x_0\oplus x_0'\oplus x_1\oplus x_1'}=\delta_{x_0\oplus x_0'\oplus x_1\oplus x_1'}$ to be non-zero:
\begin{itemize}
    \item $x_0=x_0'$ and $x_1=x_1'$. There are $\ell^2$ such terms.
    \item $x_0\neq x_0'$, and either $(x_1,x_1')=(x_0,x_0')$ or $(x_1',x_1)=(x_0,x_0')$. There are $\ell\times (\ell-1)\times 2$ such terms.
\end{itemize}
This gives a total of $3\ell^2-2\ell$ such terms. As a consequence, we have 
\begin{equation}
A^S_{x,x}=B^S_{x,x}-\frac{\kappa^2}{4\ell^2}(3\ell^2-2\ell)=\frac{1}{2}+\frac{\kappa}{2}-\frac{3\kappa^2}{4}+\frac{\kappa^2}{2\ell}\,.
\end{equation}
For off-diagonal entries $x\neq x'$, the only way for $x\oplus x'\oplus x_0\oplus x_0'\oplus x_1\oplus x_1'$ to be 0 is for one of $x_0,x_0',x_1,x_1'$ to be $x$, another to be $x'$, and the remaining two must be equal. This again follows from the goodness of $S$. There are $4\times 3$ ways of choosing which of the four elements are equal to $x$ and $x'$, and $\ell$ ways to choose the remaining pair. This slightly over-counts, since for example $(x,x',x_0,x_0',x_1,x_1')=(x,x',x,x,x,x')$ would be counted 3 times, for $x_0=x$, $x_0'=x$ and $x_1'=x$. There are 8 terms of this form (4 where there is one $x'$ and 3 $x$ among $(x_0,x_0',x_1,x_1')$, and 4 where there is one $x$ and 3 $x'$). Each term of this form is over-counted 2 extra times (for a total of 3 times). The correct number of terms with $x\oplus x'\oplus x_0\oplus x_0'\oplus x_1\oplus x_1'=0$ is therefore: $12\ell-16$.
This means the off-diagonal entries are equal to:
\begin{equation}A^S_{x,x'}=B^S_{x,x'}-\frac{\kappa^2}{4\ell^2}(12\ell-16)=\frac{\kappa}{\ell}-\frac{3\kappa^2}{\ell}+\frac{4\kappa^2}{\ell^2}\,.\end{equation}
Similar to the calculation for $B^S_\Delta$,
\begin{equation}
    A^S = \left[\frac{1}{2}
+ \kappa \!\left(\frac{1}{2} - \frac{1}{\ell}\right)
+ \kappa^{2} \!\left(-\frac{3}{4} + \frac{7}{2\ell} - \frac{4}{\ell^{2}}\right)\right] \id_S + \left[ \kappa - 3\kappa^2+\frac{4\kappa^2}{\ell} \right] \ketbra{S}\,.
\end{equation}
Therefore, the top eigenvector is indeed $\ket{S}$ and 
\begin{xalign}
\norm{A^S}_{\mathrm{op}} &= \ev{A^S}{S} \\
&=
\frac{1}{2}
+ \kappa\!\left(\frac{3}{2}-\frac{1}{\ell}\right)
+ \kappa^{2}\!\left(-\frac{15}{4}+\frac{15}{2\ell}-\frac{4}{\ell^{2}}\right) \\
&\geq \frac{1}{2}+\frac{3\kappa}{2}-\frac{15\kappa^2}{4}-\frac{5 \kappa}{\ell}\,.
\end{xalign}
Combining the previous Hoeffding's inequality with these bounds on the spectral norm completes the proof.\end{proof}

\newpage

\part{Sampling probability upper bound for quasi-even condensates}
\label{part:sampler-upper-bound}

The proof components in~\Cref{part:qcma-to-sampler} combine to prove the existence of a $\QCMA$ algorithm for spectral Forrelation implies sampling probability \emph{lower bound}. More specifically,~\Cref{sec:qcma-to-one-sampler} showed how to transform a $\QCMA$ algorithm for deciding between yes (at least $59/100$-spectrally Forrelated) and no (at most $57/100$-spectrally Forrelated) instances into a successful sampler that samples points from strong yes instances, and~\Cref{lem:strong-instance} in~\Cref{sec:strong_yes} constructed a distribution $\mathsf{Strong}_\kappa$ which is, with overwhelming probability, supported on strong yes instances. For the rest of the paper, we will choose $\kappa = 1/10$ and use $\mathsf{Strong} = \mathsf{Strong}_{1/10}$ for brevity.

Together, we have proven that if there exists a $\QCMA$ decision algorithm, then there exists a lower bound on the success probability of a sampler at guessing points from $S$ given query access to $U$ when $(S,U)$ is sampled according to the definition of~\Cref{lem:strong-instance}. The remainder of the proof is to marry this probability lower bound with a contradictory probability upper bound. In this section, we will use compressed oracle/purification techniques defined in terms of bosons.
The main theorem for~\Cref{part:sampler-upper-bound} is~\Cref{thm:big-upper-bound-thm}, and it will be stated after we introduce the bosonic framework in~\Cref{sec:prelim-bosons}.

Before starting, we would like to emphasize that the proofs and bounds we generate are almost certainly far from tight. We will loosely bound terms for legibility, and we acknowledge that the resulting bounds will not be tight. Nevertheless, they will be sufficient to contradict~\Cref{thm:formal-removal-of-S}.

\section{Quantum mechanics of bosons}

\label{sec:prelim-bosons}

As stated in the Introduction, our proof will require understanding purifications of quantum oracles as bosonic systems.  This section will provided the facts and definitions needed to understand the rest of the proof. For the purposes of this note, it suffices to consider a bosonic system with $2^n$ modes with the modes indexed by $n$-bit vectors. For notational convenience, it is also useful to be able to index the modes with $\{0, \ldots, 2^n - 1\}$ under the standard isometry. In particular we refer to the $0^n$-mode as the 0-mode.

\subsection{A natural basis}
The bosonic (position) Fock basis is a collection of \emph{orthonormal} states of the form $\ket{\ell_0, \ell_1, \ldots, \ell_{2^n-1}}$ with each $\ell_x \in \NN$. This is the basis state corresponding to $\ell_0$ bosons in the 0-mode, $\ell_1$ bosons in the 1-mode, etc. The total number of bosons is $\sum_x \ell_x$. 

We will exclusively consider states with $\ell$ total bosons; however, it is mathematically helpful to be able to add and remove bosons at will to describe the transformation from one state of $\ell$ bosons to another. \\

\subsection{The second quantization}
Let $\ket \vac \defeq \ket{0, \ldots, 0}$ be the vacuum state representing $0$ bosons in the system. Abiding by the traditional notation from physics, let $\hat a_x$ and $\hat a_x^\dagger$ be the annihilation and creation operators for a boson in the $x$-th mode, respectively.  These are defined by their action on the position Fock basis as follows
\begin{xalign}
\label{eq:creation_annihilation_definition}
    \hat{a}_x \ket{\ell_0, \ldots, \ell_{x}, \ldots, \ell_{2^{n}-1}} &= \sqrt{\ell_{x}} \ket{\ell_0, \ldots, \ell_{x}-1, \ldots, \ell_{2^{n}-1}}\quad\text{and,}\\
    \hat{a}_x^{\dagger} \ket{\ell_0, \ldots, \ell_{x}, \ldots, \ell_{2^{n}-1}} &= \sqrt{\ell_{x}+1} \ket{\ell_0, \ldots, \ell_{x}+1, \ldots, \ell_{2^{n}-1}}\,.
\end{xalign}

We will refer to these as the annihilation and creation operators for bosons in the \emph{position} basis. Note that these operators are \emph{not} unitary.  Very roughly, the factor $\sqrt{\ell_x}$ corresponds to the fact that the bosons are indistinguishable and so we do not know which of the $\ell_x$ bosons were annihilated. Likewise, for creation. In calculations about bosons, it is useful to oscillate between the Fock representation and the creation/annihilation perspective. In physics, the two representations are referred to as first and second quantizations, respectively. 

It follows from the definition of the position creation operator that 
\begin{equation}
    \frac{1}{\sqrt{\prod_{x = 0}^{2^n-1} \ell_x!} }\prod_{x = 0}^{2^n-1} \left(\hat a_x^\dagger\right)^{\ell_x} \ket{\vac} = \ket{\ell_0, \ldots, \ldots, \ell_{2^n-1}}\,.
\end{equation}
The commutation relations for bosonic position operators are given by
\begin{equation}
    [\hat a_x, \hat a_y^\dagger] = \hat a_x \hat a_y^\dagger - \hat a_y^\dagger \hat a_x 
= \delta_{xy}\,, \qquad [\hat a_x, \hat a_y] = [\hat a_x^\dagger, \hat a_y^\dagger] = 0\,. \label{eq:bosonic-commutation} \end{equation}
One may expect all annihilation and creation operators to commute since, for example, annihilating a boson at $x$ seems independent of creating a boson at $y$. This is true generally, except for annihilating and creating at the same mode $x$. Intuitively, the reason for a lack of commutation is that if there are zero bosons at mode $x$, annihilation actually does nothing since there is nothing to annihilate. This causes an asymmetry since annihilation-before-creation may have nothing to annihilate, but annihilation-after-creation will always have something to annihilate.

\subsection{A momentum basis}
We can also define the annihilation and creation operators in the \emph{momentum} basis by the Hadamard transform. Note, that this is our \textit{computer science} interpretation. Usually, the transform from position to momentum basis is given by the quantum Fourier transform over the group $\ZZ_{2^n}$.
\begin{xalign}
    \tilde{a}_y &\defeq \frac{1}{\sqrt{2^n}} \sum_{x \in \bits^n} (-1)^{x\cdot y} \hat a_x\,, \\
    \tilde{a}_y^\dagger &\defeq \frac{1}{\sqrt{2^n}} \sum_{x \in \bits^n} (-1)^{-x\cdot y} \hat a_x^\dagger\,.
\end{xalign}
The commutation relations for momentum operators can be derived to be the analogs of~\cref{eq:bosonic-commutation}:
\begin{equation} \label{eq:bosonic-commutation2}
    [\tilde a_x, \tilde a_y^\dagger] = \tilde a_x \tilde a_y^\dagger - \tilde a_y^\dagger \tilde a_x 
= \delta_{xy}\,, \qquad [\tilde a_x, \tilde a_y] = [\tilde a_x^\dagger, \tilde a_y^\dagger] = 0\,. 
\end{equation}
Having defined momentum annihilation and creation operators, we can derive that there exists a second Fock basis that can be used to describe states. This is the momentum Fock basis, and it is a collection of orthonormal states $\ket{\ell_0, \ell_1, \ldots, \ell_{2^n-1}}$ where the integers $\ell_x$ describe how many bosons are in each \emph{momentum} mode. It is the exact analog of the position Fock basis.

\subsection{Number operators}
Additionally, we can define the position and momentum number operators as $\hat n_x \defeq \hat a_x^\dagger \hat a_x$ and $\tilde{n}_x \defeq \tilde{a}_x^\dagger \tilde{a}_x$, respectively.  These are diagonal matrices in the position and momentum Fock bases that multiply a Fock basis state by the \emph{number} of bosons in the $x$'th mode --- i.e., $\hat{n}_{0} \ket{\ell_{0}, \ldots} = \ell_{0} \ket{\ell_{0}, \ldots}$, hence the name. These have the following commutation relations with the creation and annihilation operators.
\begin{xalign}
    \hat n_x \hat a_x &= \hat a_x^\dagger \hat a_x \hat a_x = (\hat a_x \hat a_x^\dagger -1)\hat a_x = \hat a_x(\hat a_x^\dagger \hat a_x - 1) = \hat a_x (\hat n_x - 1)\,, \\
    \hat n_x \hat a_x^\dagger &= \hat a_x^\dagger \hat a_x \hat a_x^\dagger = \hat a_x^\dagger (1 + \hat a_x^\dagger \hat a_x) = \hat a_x^\dagger (1 + \hat n_x)\,.
\end{xalign}
Intuitively, the number of bosons \emph{after} annihilating is just one less than the number \emph{before} annihilating (and vice-versa for creating). We also define the total number operator $\hat N = \sum_x \hat n_x$. One can verify that $\hat N = \tilde N$.

\subsection{A random bosonic setup}

A typical problem in (classical) combinatorics might start with ``place $\ell$ indistinguishable balls uniformly randomly into $N$ boxes''. The quantum mechanical interpretation is to uniformly place $\ell$ bosons into $N=2^n$ modes. One exact purification of this setup is to consider the state with $\ell$ bosons in the 0-momentum mode, which equals
\begin{equation}
\label{eq:bosonic_initial_state}
    \frac{1}{\sqrt{\ell!}} \left(\tilde{a}_0^\dagger\right)^\ell \ket \vac = \frac{1}{\sqrt{\ell! \cdot 2^{n\ell}}} \sum_{x_1, \ldots, x_\ell} \left(\prod_{i = 1}^\ell \hat a_{x_i}^\dagger \right) \ket \vac
\end{equation}
with the right side demonstrably equivalent to creating $\ell$ bosons uniformly randomly into the vacuum.
Observe that the $\sqrt{\ell!}$ additional normalization term might seem strange initially; however, it is indeed necessary. Indeed, this is a consequence of the bosons being indistinguishable, and therefore not ``knowing'' the order in which they are created. This results in overcounting, which is fixed by dividing by $\sqrt{\ell!}$.
Furthermore, suppose we measure this state in the position Fock basis and interpret the measurement as a multiset.  Note that the probability of measuring a particular multiset $S$, which contains element $x$, $\ell_x$ many times, when measuring \eqref{eq:bosonic_initial_state} can be calculated as
\begin{equation}
    \Pr[S] = \frac{\ell!}{2^{n\ell}\prod_{x} \ell_{x}!}\,,
\end{equation}
which is exactly the uniform distribution over multisets of size $\ell$.  

\subsection{Bosonic Hilbert space}

Bosonic systems on $2^n$ modes are states in an infinite-dimensional Hilbert space. As there is no bound on the number of bosons in a bosonic system, this Hilbert space is infinite-dimensional. However, the Hilbert space can be expressed as the direct sum of finite-dimensional Hilbert spaces by restricting to a fixed number of bosons:
\begin{equation}
    \Hh_{\mathrm{boson}} = \bigoplus_{\ell = 0}^\infty \Hh_{\mathrm{boson}}^{(\ell)}~\text{where}~\Hh_{\mathrm{boson}}^{(\ell)} = \mathrm{span}~\text{of Fock states of $\ell$ bosons}\,.
\end{equation}
Since $\hat N = \tilde N$, the number of bosons in the position and momentum bases is the same; therefore, the transformation mapping between the position and momentum Fock bases (i.e., the Hadamard transform) is block-diagonal with respect to this decomposition. 

In this work, we will restrict ourselves to considering bosonic systems with a fixed number, $\ell$, of bosons. Therefore, the states are in the space $\Hh_{\mathrm{boson}}^{(\ell)}$. Both the position and momentum Fock bases of this Hilbert space can be indexed by non-negative integer tuples of length $2^n$ with total sum $\ell$. Observe that this is isomorphic to the set of truth-tables of multisets $S \subseteq \bits^n$ of size $\ell$.  In the rest of the paper, we identify the states $\ket{\mathrm{tt}_{S}}$ with the position Fock basis state with bosons in locations given by the elements of $S$.

\section{Sampler upper bound statement and organization}
\label{sec:thm-statement}

This section will state the main theorem of \Cref{part:sampler-upper-bound}, and give some intuition for how the proof will go, and the organization of the rest of the part. The goal of~\Cref{part:sampler-upper-bound} is to prove a probability \emph{upper bound} on the success of any polynomial query sampler when run on the $\mathsf{Strong}$ distribution defined in~\Cref{def:def-of-strong-dist}.  This is the counterpart to \Cref{thm:formal-removal-of-S}, which proves a contradictory lower bound on the success probability, assuming a $\QCMA$ algorithm for spectral Forrelation exists. 

\subsection{Theorem statement}
The main theorem of this part is the following.
\begin{theorem}[Sampling probability upper bound]
\label{thm:big-upper-bound-thm}
    For all $v$, for all quantum algorithms $\Aa^U$ accessing an oracle $U$ and outputting $v$ distinct outputs, while making $t$ queries per output, if a pair of oracles $(S,U)$ are sampled according to distribution $\mathsf{Strong}$ (defined in~\Cref{def:def-of-strong-dist}), then the probability that all $v$ of the outputs of $\Aa^U$ are elements of $S$ is at most
    \begin{equation}
        \leq 2\left(\frac{4v((vt)^{30} + v (vt)^{20})\sqrt{\ell}}{2^{n/4}}\right)^{v} + \left(\left(\frac{(vt)^{4}}{\ell^{1/32}}\right)^{v} + e^{-5vt}\right)^2\,.
    \end{equation}
\end{theorem}

In this theorem, the exponents are not important and are the consequence of loose bounding for legibility. What matters is that for $v, t = \poly(n)$ and $\ell = 2^{cn}$, the numerators are significantly smaller than the denominators. Therefore, as $v$ grows, this quantity decreases exponentially fast. We will actually prove a slightly stronger statement, namely we will show a bound for any algorithm that first makes $T$ queries to the oracle $U$ and then outputs $v$ guesses. For $T = vt$, this is a strictly larger class of algorithms than those that make $t$ queries per guess. The rationale for studying this stronger model is that it natively handles the complexities of the memory of the algorithm between guesses. 

\subsection{Proof overview and intuition}

As suggested in the Introduction, we will heavily use the bosonic framework to analyze the success probability of a sampler algorithm. The bosonic framework will be used to analyze a uniform superposition over the pairs $(S,U)$ sampled from the $\mathsf{Strong}$ distribution defined in~\Cref{def:def-of-strong-dist}. 
By considering a uniform superposition over pairs $(S,U)$, we can conveniently analyze the average-case success probability of the sampler.

Our proof strategy for proving~\Cref{thm:big-upper-bound-thm} is reminiscent of ideas introduced by Hamoudi and Magniez~\cite{hamoudi2023quantum}. We will define a family of subspaces $\{\mathsf{QEC}_{(r, o)}\}$ indexed by positive integers $r$ and $o$, with the property that if the state after $T$ queries is mostly contained in $\mathsf{QEC}_{(r, o)}$ for $r \leq \poly(n)$ and $o \leq v/4$, then the success probability bound for the guessing algorithm is enough to prove~\Cref{thm:big-upper-bound-thm}. Once we have established that states supported on the subspace $\mathsf{QEC}_{(r, o)}$ have a low sampling success probability, we will prove that the state of all $T$-query algorithms querying the purified $(S, U)$ oracles is supported almost entirely on the $\mathsf{QEC}_{(r, o)}$ subspace for some $r$ that is polynomial in $T$.  Putting these together we will have a sampling probability upper bound.

The notation $\mathsf{QEC}_{(r, o)}$ refers to something we call a \emph{$(r,o)$-quasi-even condensate}. Informally, a $(r,o)$-quasi-even condensate is a state in the span of momentum Fock states where the number of \emph{odd} number operators is $\le o$ and the number of momentum modes that are non-zero is at most $r$.  See the following remark for more details about this choice of naming.

\begin{remark}[Nomenclature]
\label{remark:condensate-nomenclature}
We borrow the term condensate from many-body physics, where it denotes a regime in which a macroscopic fraction of bosons occupy a single-particle mode—typically the zero-momentum mode—giving rise to collective, mean-field–like behavior. In our setting, we use condensate analogously to indicate that most bosons remain in the zero-momentum mode, with only a small number of excitations occupying other modes. The modifier quasi-even reflects that the occupations across modes are almost parity-symmetric: only a few modes have odd occupation numbers.

 Hence, a quasi-even condensate refers to a subspace of Fock states that are both condensate-like (dominated by the zero-momentum mode) and nearly even under mode-parity. This terminology is not meant in a thermodynamic sense, but rather as a descriptive analogy capturing the structure of states that remain close to an even-parity condensate configuration with a few parity defects.
\end{remark}

\subsection{Organization}

We begin by proving a sampling probability upper-bound for quasi-even condensates in~\Cref{sec:bad_subspace}. Then, in~\Cref{sec:bosons} we prove how to represent the state of an algorithm that has made a few queries to the oracle $U$ when run in superposition over all pairs $(S,U) \sim \mathsf{Strong}$ (defined in~\Cref{def:def-of-strong-dist}). Next, in~\Cref{part:poly-query-implies-qec} we prove that the state of an algorithm that has made polynomial in $n$ many queries is a quasi-even condensate. This proves that we can apply the previously derived sampling probability upper-bound and completes the proof. %

\section{Sampler upper bounds for quasi-even condensates}
\label{sec:bad_subspace}

This section will define states that are quasi-even condensates and prove an upper bound on the sampling success probability of all algorithms whose states are supported on quasi-even condensates. 
In future sections, we prove that the post-query state of all polynomial query samplers is close to a quasi-even condensate.

Recall that the larger goal is to prove an upper bound on the success probability of a query algorithm on a random instance from $\mathsf{Strong}$ (defined in~\Cref{def:def-of-strong-dist}).  We can imagine that the query algorithm is split into two steps: (a) the first is the $T$ queries of interaction with the oracle $U$ and (b) the measurement with respect to $S$ of the guesses.
Interaction with a $(S,U) \sim \mathsf{Strong}$ can be studied by first purifying the distribution over oracles and interacting coherently with each pair $(S,U)$ in superposition. By purifying, we mean that queries to the oracle $U$ are replaced by the linear extension of the unitary
\begin{equation}
    \ket{b, x, y} \ket{\mathrm{tt}_S} \ket{\mathrm{tt}_U} \mapsto (-1)^{b \cdot U(y)} \ket{b, x,y} \ket{\mathrm{tt}_S} \ket{\mathrm{tt}_U}
\end{equation}
where $\mathrm{tt}_{(\cdot)}$ is our notation for the truth table of the corresponding oracle.
Checking the $v$ guesses output by the algorithm is equivalent to measuring the final (entangled) state of the algorithm and oracle with respect to 
\begin{equation} \label{eq:pi_success}
    \Pi_{\mathrm{succ}} \defeq \sum_{\substack{z_1, \ldots, z_v \in (\bits^n)^{v}\\\text{ distinct}}} \ketbra{z_1, \ldots, z_v} \otimes \qty(\sum_{S: z_1, \ldots, z_v \in S} \ketbra{\mathrm{tt}_S}).
\end{equation}
The success probability of this sampler over the distribution of random oracles $(S,U)$ is equal to the probability that measuring $\Pi_{\mathrm{succ}}$ on post-query state of the sampler, when run on the purification of the oracles, accepts.  We now define $(r, o)$-quasi-even condensates.

\begin{definition}[Quasi-even condensates] \label[definition]{def:low-collision-subspace}
    Let $u = (u_x)_{x \in \bits^n}$ be a tuple of non-negative integers such that $\sum_{x} u_{x} = \ell$, representing a momentum Fock state of $\ell$ bosons. Then we say that $u$ describes an \emph{$(r, o)$-quasi-even condensate} if 
    \begin{enumerate}
        \item (Condensate) $u_0 \geq \ell - r$. I.e., most of the bosons are in the 0-mode.
        \item (Quasi-even) At most $o$ many $u_{x}$, except for $u_0$, are odd.
    \end{enumerate}
    We define \emph{$(r, o)$-quasi-even condensates} to be any state in the span of momentum Fock states corresponding to quasi-even condensate tuples, or $\mathrm{span}\{~\ket{u} : u \text{ is } (r, o)\text{-quasi-even condensate}~\}$.
    We further define $\mathsf{QEC}_{(r, o)}$ to be the projector onto this subspace. Additionally, we define projectors $\mathsf{Con}_r$ and $\mathsf{QE}_o$ as the projectors onto $r$-condensates and $o$-quasi-even states, respectively. 
    We also define $\mathsf{QE}_{= o}$ and $\mathsf{QE}_{\geq o}$ to be the projectors onto states with exactly $o$ odd $u_{x}$'s (excluding $u_0$) and $\geq o$ many odd $u_{x}$'s (excluding $u_0$). All of these projectors are diagonal in the momentum Fock basis and therefore commute. By definition,
    \begin{equation}
        \mathsf{QEC}_{r,o} = \mathsf{Con}_r \cdot \mathsf{QE}_o =  \mathsf{QE}_o \cdot \mathsf{Con}_r.
    \end{equation}
\end{definition}
Note that all of the previously defined projectors require a state of exactly $\ell$ bosons.
Within this subspace, $\mathsf{Con}_0$ is the projector onto the state that has all $\ell$ bosons in the $0$-momentum mode, i.e., $\ket{\ell, \ldots}$, and $\mathsf{Con}_{\ell}$ is the projector onto all states with $\ell$ bosons. The following theorem is the main result of this section, which shows a sampling probability upper bound for $(r,v/4)$-quasi-even condensates when $r \ll \ell$.

\begin{theorem}[Quasi-even condensate probability upper bound] \label{thm:main-sampling-upper-bound}
Let $\Pi_{\mathrm{succ}}$ be the previously defined success operator and $\mathsf{QEC}_{(r,v/4)}$ be the projector onto $(r,v/4)$-quasi-even condensates, which only acts on the $\reg{S}$ register. Then, \begin{equation}
    \norm{ \mathsf{QEC}_{(r,v/4)} \cdot \Pi_{\mathrm{succ}} \cdot \mathsf{QEC}_{(r,v/4)} }_{\mathrm{op}} \leq 2 \left(\frac{4v (r^3 + vr^2) \sqrt{\ell}}{2^{n/4}}\right)^{v}\,.
\end{equation}
\end{theorem}
The key lemma required to prove~\Cref{thm:main-sampling-upper-bound} is the following. It provides an upper bound on the success probability of a quasi-even condensate that makes distinct guesses $z_1, \ldots, z_v$. A priori,~\Cref{lem:ham-intermediate} may seem strange as it calculates the maximum eigenvalue of a product of number operators on the space of quasi-even condensates. However, we will prove that this is sufficient for proving an upper bound on the maximum eigenvalue of $\Pi_{\mathrm{succ}}$ within the space of quasi-even condensates. This proof will rely on the fact that the upper bound proven in~\Cref{lem:ham-intermediate} is independent of the choice of guess locations.

\begin{lemma}
\label[lemma]{lem:ham-intermediate}
For distinct coordinates $z_1, \ldots, z_v \in \bits^n$, %
\begin{equation}
\norm{\mathsf{QEC}_{(r,v/4)} \cdot \hat n_{z_1} \ldots \hat n_{z_v} \cdot \mathsf{QEC}_{(r,v/4)}}_{\mathrm{op}} \leq 2 \left(\frac{4v (r^3 + vr^2) \sqrt{\ell}}{2^{n/4}}\right)^{v}\,.
\end{equation}
\end{lemma}
\noindent We first prove that~\Cref{lem:ham-intermediate} implies~\Cref{thm:main-sampling-upper-bound} and then finish this section with the proof of~\Cref{lem:ham-intermediate}.
\begin{proof}[Proof of~\Cref{thm:main-sampling-upper-bound}]
Observe that we can reformulate $\Pi_{\mathrm{succ}}$ as the following:
\begin{equation} 
    \Pi_{\mathrm{succ}} = \sum_{\substack{z_1, \ldots, z_{v} \in (\bits^n)^{v} \\ \mathrm{distinct}}} \ketbra{z_1, \ldots, z_v} \otimes \Pi_{z_1, \ldots, z_v}\,.
\end{equation}
where define $\Pi_{z_1, \ldots, z_v}$ to be the projection onto states { $\ket{\mathrm{tt}_S}$} that have at least one boson in position modes $z_1, \ldots z_v$.
Next, we note that for distinct $z_1, \ldots, z_v$, the annihilation operators commute and we have that 
    \begin{equation}
        \left(\hat a_{z_1}^{\dagger} \ldots \hat a_{z_v}^{\dagger}\hat a_{z_1} \ldots \hat a_{z_v}\right) = \hat n_{z_1} \ldots \hat n_{z_v} \succcurlyeq \Pi_{z_1, \ldots, z_v}\,.
    \end{equation}
This statement is conceptually equivalent to applying Markov's inequality: that $\Pr[X > 0] \leq \Exp[X]$ for non-negative random variable $X$. Therefore, it follows that 
\begin{xalign}&\norm{\mathsf{QEC}_{(r,v/4)} \cdot \Pi_{\mathrm{succ}} \cdot 
\mathsf{QEC}_{(r,v/4)}}_{\mathrm{op}} \leq \norm{\mathsf{QEC}_{(r,v/4)} \cdot \Lambda_{\mathrm{succ}} \cdot \mathsf{QEC}_{(r,v/4)}}_{\mathrm{op}} \\
&\label{eq:def-lambda-success-repeat}\text{where~}
    \Lambda_{\mathrm{succ}} \defeq \sum_{\substack{z_1, \ldots, z_v \\ \mathrm{distinct}}} \ketbra{z_1, \ldots, z_v} \otimes \left(\hat a_{z_1}^{\dagger} \ldots \hat a_{z_v}^{\dagger}\hat a_{z_1} \ldots \hat a_{z_v}\right)\,.
\end{xalign}
Now we will prove the following bound, applying \Cref{lem:ham-intermediate}.
\begin{equation}
    \norm{ \mathsf{QEC}_{(r,v/4)} \cdot \Lambda_{\mathrm{succ}} \cdot \mathsf{QEC}_{(r,v/4)} }_{\mathrm{op}} \leq 2 \left(\frac{4v (r^3 + vr^2) \sqrt{\ell}}{2^{n/4}}\right)^{v}\,.
\end{equation}
Consider any state $\ket{\varphi}$ supported on $\mathsf{QEC}_{(r,v/4)}$ and write it in its decomposition based on guesses:
    \begin{equation}
        \ket{\varphi} = \sum_{z_1, \ldots, z_v} \alpha_{z_1, \ldots, z_v} \ket{z_1, \ldots, z_v} \otimes \ket*{\varphi_{z_1, \ldots, z_v}}
    \end{equation}
    where $\ket*{\varphi_{z_1, \ldots, z_v}}$ is the remainder of the state (normalized). Then,
    \begin{xalign}
        \ev*{\Lambda_{\mathrm{succ}}}{\varphi} &= \sum_{\substack{z_1, \ldots, z_v \\ \text{distinct}}} \abs{\alpha_{z_1, \ldots, z_v}}^2 \cdot \ev{\hat n_{z_1} \ldots \hat n_{z_v}}{\varphi_{z_1, \ldots, z_v}}\\
        &\leq \max_{\substack{z_1, \ldots, z_v \\ \text{distinct}}} \ev{\hat n_{z_1} \ldots \hat n_{z_v}}{\varphi_{z_1, \ldots, z_v}} \\
        &\leq 2 \left(\frac{4v (r^3 + vr^2) \sqrt{\ell}}{2^{n/4}}\right)^{v}.
    \end{xalign}
 Here, the final line is the application of \Cref{lem:ham-intermediate}.
\end{proof}

It only remains to prove~\Cref{lem:ham-intermediate}, which we prove next.
\begin{proof}[Proof of~\Cref{lem:ham-intermediate}]
    Recall that an upper bound on the spectral norm of a Hermitian matrix is the max $1$-norm over all rows.  Our goal is to study $\hat n_{z_1} \ldots \hat n_{z_v}$ within the space defined by quasi-even condensates. So, we restrict our attention to the row indexed by $u$ in the momentum Fock basis where $u$ is a $(r,o)$-quasi-even condensate tuple. Therefore, the goal is to bound
    \begin{xalign}[eq:decomposition-into-numb-and-bound]
        &\sum_{(r,o)-\text{QEC tuple}~w} \abs{\mel*{w}{\hat n_{z_1} \ldots \hat n_{z_v}}{u}} \\
        &\hspace{2mm}= \sum_{d \geq 0} \qty(\sum_{\substack{(r,o)-\text{QEC tuple}~w \\ \abs{w - u} = 2d}} \abs{\mel*{w}{\hat n_{z_1} \ldots \hat n_{z_v}}{u}} ) \\
        &\hspace{2mm}\le \sum_{d \geq 0} \qty( \max_{\substack{(r,o)-\text{QEC tuple}~w \\ \abs{w - u} = 2d}} \abs{\mel*{w}{n_{z_1} \ldots \hat n_{z_v}}{u}}) \cdot \#\qty{w : \substack{(r,o)-\text{QEC tuple}~w \\ \abs{w - u} = 2d}}\,.
    \end{xalign}
As the previous equation suggests, we will bound the terms in the previous equation using two additional claims.  The first is a bound on the entry of the matrix corresponding to a tuple $w$, at distance $2d$ from a fixed tuple $u$.  

\begin{claim}
\label[claim]{claim:decay-bound-condensate}
    Fix two $(r,o)$-quasi-even condensate tuples $u$, $w$ such that $\abs{u - w} = 2d$ with $\ell \geq 2d$.  
    Then the following holds for all distinct $z_1, \ldots, z_v$.
    \begin{equation}\label{eq:overlapbound}
        \abs{\mel*{u}{\hat n_{z_1}\ldots \hat n_{z_v}}{w}} \leq \begin{cases} \displaystyle
            v! (2r)^{v + d/2} \frac{\ell^{v-d/2}}{2^{nv}} & \text{if}~d \leq v, \\ 0 &\text{if}~d > v. \end{cases}
    \end{equation}
\end{claim}
The proof of this claim is deferred to after the proof of~\Cref{lem:ham-intermediate}.
Our second claim is an upper bound on the number of quasi-even condensates $w$ at distance $2d$ from our initial quasi-even condensate $u$ for all values of $d$.  

\begin{claim}
   \label[claim]{claim:counting-bound-condensate}
       For every $(r, o)$-quasi-even condensate tuple $u$ and $d \geq 0$, the number of $(r, o)$-quasi-even condensate tuples $w$ that are at exactly distance $2d$ from $u$ is upper bounded by
       \begin{equation}\label{eq:condensates_at_distance_d}
           (2^{n+1})^{d/2 + o} (r + d/2 + o)^{d/2 + o}\,.
   \end{equation}
\end{claim}
The proof of this claim is also deferred to after the proof of~\Cref{lem:ham-intermediate}.
To finish the proof, we return to~\eqref{eq:decomposition-into-numb-and-bound} to bound $\sum_{(r,o)-\text{QEC tuple}~w} \abs{\mel*{w}{\hat n_{z_1} \ldots \hat n_{z_v}}{u}}$. We explain the derivation of each subequation after the statement.

\begin{xalign}
    &\sum_{(r,o)-\text{QEC tuple}~w} \abs{\mel*{w}{\hat n_{z_1} \ldots \hat n_{z_v}}{u}} \\
    &\hspace{10mm}\leq \sum_{d \geq 0} \qty( \max_{\substack{(r,o)-\text{QEC tuple}~w \\ \abs{w - u} = 2d}} \abs{\mel*{w}{\hat n_{z_1} \ldots \hat n_{z_v}}{u}}) \cdot \#\qty{w : \substack{(r,o)-\text{QEC tuple}~w \\ \abs{w - u} = 2d}}\\
    &\hspace{10mm}\leq v! (2r)^{v} \sum_{d = 0}^{v} (2r)^{d/2} \frac{\ell^{v - d/2}}{(2^n)^{v}} \left( (2^{n+1})^{d/2 + o} (r + 1 + d/2 + o)^{d/2 + o} \right)\label{line:65c}\\
    &\hspace{10mm}= v! (2r)^{3v/2}\left(\frac{\ell}{2^n}\right)^{v} (2^n)^{o} \sum_{d = 0}^{v} \left(\sqrt{\frac{2^{n+1}}{\ell}}\right)^{d} (r + 1 + d/2 + o)^{d/2 + o}\label{line:65d}\\
    &\hspace{10mm}\leq v! (2r)^{3v/2}\left(\frac{\ell}{2^n}\right)^{v} (2^n)^{o} \sum_{d = 0}^{v} \left(\sqrt{\frac{2^{n+1}}{\ell}}\right)^{d} (r +  1 +v/2 + o)^{d/2 + o}\label{line:65e}\\
    &\hspace{10mm}= v! (2r)^{3v/2}\left(\frac{\ell}{2^n}\right)^{v} (2^n)^{o} (r + v/2 + o)^{o} \sum_{d = 0}^{v} \left(\sqrt{\frac{2^{n+1}}{\ell}}\right)^{d} (r + 1 + v/2 + o)^{d/2}\label{line:65f}\\
    &\hspace{10mm}\leq 2 v! (2r)^{3v/2}\frac{\ell^{v}}{(2^n)^{v - o}} \left(\sqrt{\frac{2^{n+1}}{\ell}}\right)^{v} (r + 1 + v/2 + o)^{v/2 + o}\label{line:65g}\\
    &\hspace{10mm}\leq 2^{2v + 1} v! (r^{3v/2}) (r + 1 + v/2 + o)^{v/2 + o} \frac{\ell^{v/2}}{(2^n)^{v/2 - o}}\,.\label{line:65h}
\end{xalign}
Here, to derive~\eqref{line:65c}, we use the bounds in \Cref{claim:decay-bound-condensate} and \Cref{claim:counting-bound-condensate}.~\Eqref{line:65d} factors out the terms that do not depend on $d$,~\eqref{line:65e} uses that $d\leq v$, and~\eqref{line:65f} again factors out terms that do not depend on $d$. Then we observe that the sum in~\eqref{line:65f} is a geometric series; since the summand 
\begin{equation}
\frac{\sqrt{2^{n+1}}}{  \ell} (r + 1 + v/2 + o) \ge 2
\end{equation}
and, therefore, the series sums to at most twice the largest term in the series (i.e., when $d = v$), giving~\eqref{line:65g}.  Re-arranging terms again gives us the bound in~\eqref{line:65h}
Setting $o = v/4$, we finish the proof, as we have a bound of
\begin{xalign}
    \norm{ \mathsf{QEC}_{(r,o)} \cdot \hat n_{z_1} \ldots \hat n_{z_v} \cdot \mathsf{QEC}_{(r,o)}}_{\mathrm{op}} &\leq \max_{(r, o)-\text{QEC tuple}~u}\sum_{(r,o)-\text{QEC tuple}~w} \abs{\mel*{w}{\hat n_{z_1} \ldots \hat n_{z_v}}{u}} \\
    &\leq \frac{2^{2v + 1} v! (r^{3v/2}) (r + 3v/4 + 1)^{3v/4} \ell^{v/2}}{(2^n)^{v/4}} \\
    &\leq 2 \left(\frac{4v (r^3 + vr^2) \sqrt{\ell}}{2^{n/4}}\right)^{v}\,.
\end{xalign}
Here, we use the crude upper bounds $r^{3/2} \leq r^{2}$ and $3v/4 + 1\leq v$ to remove the constants in the expression.
\end{proof}
Now we prove the additional claims. 
\begin{mathinlay}
\begin{claim*}[Claim~\ref{claim:decay-bound-condensate} restated]
    Fix two $(r,o)$-quasi-even condensate tuples $u$, $w$ such that $\abs{u - w} = 2d$ with $\ell \geq 2d$.  
    Then the following holds for all distinct $z_1, \ldots, z_v$.
    \begin{equation}\label{eq:overlapbound-restated}
        \abs{\mel*{u}{\hat n_{z_1}\ldots \hat n_{z_v}}{w}} \leq \begin{cases} \displaystyle
            v! (2r)^{v + d/2} \frac{\ell^{v-d/2}}{2^{nv}} & \text{if}~d \leq v, \\ 0 &\text{if}~d > v. \end{cases}
    \end{equation}
\end{claim*}
\begin{proof}
    Recall that $\hat a_{z_i} = \frac{1}{\sqrt{2^n}} \sum_{w} (-1)^{w \cdot z_i} \tilde{a}_{w}$.  Then we can write
    \begin{xalign}
        \abs*{\mel*{u}{\hat n_{z_1} \ldots \hat n_{z_v}}{w}} &= \abs*{\mel*{u}{\hat a_{z_1}^{\dagger}\ldots \hat a_{z_v}^{\dagger} \hat a_{z_1}\ldots \hat a_{z_v}}{w}} \\
        &= \abs{\frac{1}{2^{nv}} \sum_{\substack{\alpha_1, \ldots, \alpha_v \\ \beta_1, \ldots, \beta_v}} \qty(\prod_{i = 1}^{v} (-1)^{z_i \cdot (\alpha_i \oplus \beta_i)}) \mel*{u}{\tilde{a}_{\alpha_{1}}^{\dagger}\ldots \tilde{a}_{\alpha_{v}}^{\dagger}\tilde{a}_{\beta_1} \ldots \tilde{a}_{\beta_v}}{w} }\\
        &\leq \frac{1}{2^{nv}} \sum_{\substack{\alpha_1, \ldots, \alpha_v \\ \beta_1, \ldots, \beta_v}}  \abs{\mel*{u}{\tilde{a}_{\alpha_{1}}^{\dagger}\ldots \tilde{a}_{\alpha_{v}}^{\dagger}\tilde{a}_{\beta_1} \ldots \tilde{a}_{\beta_v}}{w}}\,. \label{eq:expansion-gram-matrix}
    \end{xalign}
Next, if $d > v$, observe that any term in~\eqref{eq:expansion-gram-matrix} corresponds to subtracting $1$ from momentum modes $\beta_1, \ldots, \beta_v$, and adding $1$ to momentum modes $\alpha_1, \ldots, \alpha_v$ starting from the quasi-even condensate. Furthermore, the term is non-zero if any only if they correspond to a sequence of additions and subtractions that map $w$ to $u$. By assumption, $u$ and $w$ differ by $2d > 2v$ edits in terms of the 1-norm, and so there is no sequence of $v$ many additions $\alpha_1, \ldots, \alpha_v$ and subtractions $\beta_1, \ldots, \beta_v$ mapping $w$ to $u$, and thus every term in the sum is $0$.

Now we switch to the case when $d \leq v$.  We first bound each of the non-zero terms in the sum.  Fix a choice $\alpha_1, \ldots, \alpha_v$ and $\beta_1, \ldots, \beta_v$, and consider the term
\begin{equation}
    \abs{\bra{u} \tilde{a}^{\dagger}_{\alpha_1} \ldots \tilde{a}^{\dagger}_{\alpha_v} \tilde{a}_{\beta_1} \ldots \tilde{a}_{\beta_v} \ket{w}}\,.
\end{equation}
As noted before, this is only non-zero if subtracting $1$ from the modes $\beta_1, \ldots, \beta_v$ and adding $1$ to the modes $\alpha_1, \ldots, \alpha_{v}$ maps from the Fock state $w$ to $u$.  To bound the term, we notice that when we apply an annihilation operator to a mode $z$, the norm can increase multiplicatively by at most $\sqrt{w_{z}}$, and whenever $z \neq 0$, we can bound $w_{z}$ by $r$, the total number of non-zero bosons.  Next, since $w,u$ both are exactly $\ell$ boson states, at most half of their difference can be accounted for by the 0-mode. Since $w$ and $u$ differ by $2d$, therefore at least $d$ of the combined collection of $\beta$'s and $\alpha$'s must be non-zero, meaning at most $2v-d$ of them can be 0. The $\leq 2v - d$ annihilation operators corresponding to the 0-mode have operator norm bounded by $\sqrt{\ell}$, and the remaining $d$ operators are bounded by $\sqrt{r}$ (when restricting to the subspace of $(r, o)$-quasi-even condensates).  Thus, we have the following upper bound:
\begin{equation}
    \abs{\bra{u} \tilde{a}^{\dagger}_{\alpha_1} \ldots \tilde{a}^{\dagger}_{\alpha_v} \tilde{a}_{\beta_1} \ldots \tilde{a}_{\beta_v} \ket{w}} \leq \ell^{v - d/2} r^{d/2}\,.
\end{equation}

The final step in the proof is to bound the total number of non-zero terms in the sum.  To do this, we make two observations.  We first note that we can bound the number of choices of $\beta_1, \ldots, \beta_v$ that correspond to a non-zero term in the sum by $\binom{r+1}{v}$.  This is because whenever $\beta_i$ acts on a mode that is unoccupied in $w$, it maps $\ket{w}$ to $0$, and there are only $r+1$ many occupied modes in $w$.  

Then, we notice that for any $u,w$ and $\beta_1,\cdots,\beta_v$, there is a unique multiset of creation operators  $\{\alpha_1,\cdots,\alpha_v\}$ such that $\bra{u} \tilde{a}^{\dagger}_{\alpha_1} \ldots \tilde{a}^{\dagger}_{\alpha_v} \tilde{a}_{\beta_1} \ldots \tilde{a}_{\beta_v} \ket{w}\neq 0$. The number of possible ordered lists of $\alpha_1,\cdots,\alpha_v$ is then at most $v!$.  Combining our bounds, we have
\begin{xalign}
    \abs*{\mel*{u}{\hat n_{z_1} \ldots \hat n_{z_v}}{w}} &\leq \frac{1}{2^{nv}} \sum_{\substack{\alpha_1, \ldots, \alpha_v \\ \beta_1, \ldots, \beta_v}}  \abs{\mel*{u}{\tilde{a}_{\alpha_{1}}^{\dagger}\ldots \tilde{a}_{\alpha_{v}}^{\dagger}\tilde{a}_{\beta_1} \ldots \tilde{a}_{\beta_v}}{w}}\\
    &\leq \frac{1}{2^{nv}} \ell^{v - d/2} r^{d/2} \sum_{\substack{\alpha_1, \ldots, \alpha_v \\ \beta_1, \ldots, \beta_v}} \delta\left(\abs{\bra{u} \tilde{a}^{\dagger}_{\alpha_1} \ldots \tilde{a}^{\dagger}_{\alpha_v} \tilde{a}_{\beta_1} \ldots \tilde{a}_{\beta_v} \ket{w}} \neq 0\right)\\
    &= \frac{1}{2^{nv}} \ell^{v - d/2} r^{d/2} v! \binom{r+1}{v}\\
    &\leq v! (2r)^{v + d/2} \frac{\ell^{v - d/2}}{2^{nv}}\,.
\end{xalign}
Here, we apply our bound on the magnitude of $\abs{\mel*{u}{\tilde{a}_{\alpha_{1}}^{\dagger}\ldots \tilde{a}_{\alpha_{v}}^{\dagger}\tilde{a}_{\beta_1} \ldots \tilde{a}_{\beta_v}}{w}}$, then apply our bound on the number of non-zero terms in the sum, and finally re-arrange terms to get the desired expression.
\end{proof}
\end{mathinlay}

Finally, we prove the second of the additional claims.  The proof of the lemma is more combinatorial in nature than the other proofs in this section, so we provide some intuition before the proof.

One way to bound the number of quasi-even condensates at distance $2d$ from a fixed condensate $u$ would be to think about tuples $e$ (for ``error'' relative to $u$) with sum $0$ and $1$-norm $2d$ (representing the distance between $u$ and another quasi-even condensate) and upper bound the ways to construct $e$.  Of course, simply taking all vectors of norm $2d$ would imply an upper bound of $\sim (2^{n})^{2d}$, i.e., to add $d$ units of positive difference, and $d$ units of negative distance to any of the $2^{n}$ entries.  However, this does not use the fact that our condensates are mostly paired up.  To take advantage of this, we might imagine that we split $e = 2e_{\mathrm{even}} + e_{\mathrm{odd}}$, where $e_{\mathrm{odd}}$ is a vector consisting of $\{-1, 0, +1\}$ entries.  Here, we again run into a number of problems.  First, even ignoring $e_{\mathrm{odd}}$, the number of ways to create $e_{\mathrm{even}}$ is still roughly $\sim(2^{n})^{d}$ (i.e., add $d/2$ positive and negative units to the $2^{n}$ units), and secondly, there are many choices for $e_{\mathrm{even}}$ and $e_{\mathrm{odd}}$, because anytime an entry of $e_{\mathrm{odd}}$ is $+1$, we can add $1$ to the corresponding entry of $e_{\mathrm{even}}$ and switch $e_{\mathrm{odd}}$ to be $-1$.  To fix these issues requires first noticing that we can only negative units to entries of $u$ that are non-zero, reducing the number of ways to generate $e_{\mathrm{even}}$ to $\sim (2^{n})^{d/2} \cdot r^{d/2}$, and describing a ``canonical'' way to assign $e_{\mathrm{odd}}$.  These complications require some careful and lengthy accounting, which the following claim handles.

\begin{mathinlay}
    \begin{claim*}[Claim~\ref{claim:counting-bound-condensate} restated]
         For every $(r, o)$-quasi-even condensate tuple $u$ and $d \geq 0$, the number of $(r, o)$-quasi-even condensate tuples $w$ that are at exactly distance $2d$ from $u$ is upper bounded by
         \begin{equation}\label{eq:condensates_at_distance_d-restated}
             (2^{n+1})^{d/2 + o} (r + d/2 + o)^{d/2 + o}\,.
     \end{equation}
     \end{claim*}
        \begin{proof}
        The proof will proceed as follows: We will first define a more abstract counting problem having to do with placing balls in bins with constraints.  We will then show that the answer to this counting problem upper bounds the number of $(r, o)$-quasi-even condensates at distance $d$ from $u$, and further show that the answer is upper bounded by \eqref{eq:condensates_at_distance_d}.\\

        \noindent
        \textit{Defining the counting problem.} 
        For every tuple $u$, define $\mathrm{pos}(u)$ to be the indices of $u$ that have non-zero entries. The counting problem that we consider is the number of ways to assign  balls into $2^n$ bins such that:
        \begin{enumerate}
            \item Each ball is labeled with an integer {\color{red} -2}, {\color{red} -1}, {\color{blue} +1}, or {\color{blue} +2} and the balls of a particular label are identical.
            \item There are exactly $\lfloor d/2 \rfloor$ many ${\color{blue} +2}$ and $\lfloor d/2 \rfloor$ many ${\color{red} -2}$ balls, and exactly $o$ many ${\color{blue} +1}$ and $o$ many ${\color{red} -1}$ balls.
            \item For any bin outside of $\mathrm{pos}(u)$, the sum of the labels of the balls placed in this bin must be $\geq 0$.
        \end{enumerate}

        We first prove that the number of assignments of balls to bins is an upper bound for the number of $(r, o)$-quasi-even condensates at distance $2d$ from a given quasi-even condensate $u$ by constructing an injective map from quasi-even condensates to assignments. \\

        \noindent \textit{Constructing the injective map.} Fix a quasi-even condensate $w$ at distance $2d$ from $u$ and let $e = w - u$ be the entry-wise difference of $u$ and $w$.  Note that because $u$ and $w$ have the same number of bosons, the sum of entries of $e$ must be $0$. Divide $e$ into positive and negative components---i.e., such that
        \begin{equation}
            e_+ \ge 0, e_{-} \ge 0 \text{~such that}~e = e_+ - e_-.
        \end{equation}
        Define $e_{\mathrm{even}} \defeq 2 \left(\lfloor e_{+} / 2 \rfloor - \lfloor e_{-} / 2 \rfloor\right)$, where operations like the floor and dividing by $2$ act entry-wise on the tuple, and define $e_{\mathrm{odd}} \defeq e - e_{\mathrm{even}}$.  
        We first make some remarks about $e_{\mathrm{even}}$ and $e_{\mathrm{odd}}$. For any entry $e_x$,
        \begin{equation}
            (e_{\mathrm{odd}})_x = \begin{cases}
                0 & \text{if}~e_x~\text{is even}, \\
                \mathrm{sgn}(e_x) & \text{if}~e_x~\text{is odd}.
            \end{cases}
        \end{equation}
        We further note that because we started with two quasi-even condensates, the number of non-zero entries in $e_{\mathrm{odd}}$ is upper bounded by $2o$. %
        Finally, we note that the tuple $e_{\mathrm{even}}$ always satisfies $\abs{e_{\mathrm{even}}} \leq 2\lfloor d/2 \rfloor$, meaning that if $d$ is not even, $e_{\mathrm{even}}$ accounts for strictly less than $d/2$ of the distance between $w$ and $u$. \\

        \noindent Now consider the following assignment of balls to bins. Without loss of generality, assume that the sum of entries of $e_{\mathrm{even}}$ is $\le 0$; the positive case follows by a similar logic.
        \begin{enumerate}
            \item For each bin $x$, if $(e_{\mathrm{even}})_x \ge 0$, add $(e_{\mathrm{even}})_{x} / 2$ many ${\color{blue} +2}$ balls to the bin, and if it $e'_x < 0$, add $(e_{\mathrm{even}})_{x} / 2$ many ${\color{red} -2}$ balls to the bin. 
            \item Let $b_{+}$ and $b_{-}$ be the number of ${\color{blue} +2}$ and ${\color{red} -2}$ balls assigned in the previous step, respectively. Note, $b_+ \leq b_-$. Add $\lfloor d/2 \rfloor - b_{-}$ many ${\color{blue} +2}$ and the same number of ${\color{red} -2}$ balls to the $0$-th bin. This assigns all the ${\color{red} -2}$ balls. 
            \item $b_{\mathrm{rem}} \defeq b_- - b_+ \ge 0$ is the number of ${\color{blue} +2}$ that have not been assigned ($\mathrm{rem}$ stands for ``remainder''). 
            Assign an additional $+ 2$ ball to the bins corresponding to the first (in lexicographical ordering) $b_{\mathrm{rem}}$ entries of $e_{\mathrm{odd}}$ that equal $+1$.
            Then, define $e_{\mathrm{balanced}}$ to be equal to $e_{\mathrm{odd}}$ where we subtract $2$ from the entries affected by this step. 
            \item At this step, $e_{\mathrm{balanced}}$ is a tuple where every entry is $\in \{-1, 0,+1\}$ because the previous step only subtracts $2$ from entries in $e_{\mathrm{odd}}$ that were ${\color{blue} +1}$.
            Furthermore, we have that the sum of entries in $e_{\mathrm{balanced}}$ is $0$ and there are at most $2o$ many non-zero entries of $e_{\mathrm{balanced}}$.
            Thus, there are at most $o'\leq o$ many ${+1}$ and $o'$ many $-1$ entries of $e_{\mathrm{balanced}}$. Thus, we can assign $o'$ many ${\color{blue} + 1}$ and ${\color{red}-1}$ balls according to non-zero entries of $e_{\mathrm{balanced}}$, and assign the remaining $o - o'$ many ${\color{blue} +1}$ and ${\color{red} -1}$ balls to the $0$-th bin. 
        \end{enumerate}
        One can verify that this is a valid assignment.
        Finally, if we take the assignment of balls that this procedure outputs and add up the labels of the balls, we recover the vector $e = w - u$.  Since for any two quasi-even condensates $w$ and $w'$, this difference vector can not be the same, we have that our assignment is injective, and thus the number of assignments upper bounds the number of quasi-even condensates at distance exactly $2d$. \\

        \noindent \textit{Bounding the combinatorial identity.} Now, we upper bound the number of assignments of the balls to bins using the stars and bars counting trick\footnote{\textit{Stars and bars} is a combinatorial technique for counting the number of ways to partition $a$ identical items (``stars'') into $b$ distinct bins~\cite{wiki:Stars_and_bars_(combinatorics)}. The number of solutions is ${a + b - 1 \choose b-1}$. This can be derived by observing that any placement is equivalent to placing $a$ (``stars'') and $b - 1$  ``bars'' in a row, with the number of stars between two bars denoting the number of stars to place in the corresponding bin. The identified bijection proves that the number of ways is equal to selecting $b-1$ locations for the bars among $a + b - 1$.}.  We can over-count this by $(a)$ assigning the ${\color{blue} +2}$ and ${\color{blue} +1}$ balls to $2^{n}-1$ bins, and $(b)$ assigning the ${\color{red} -2}$ and $-1$ balls to the $r + 1$ many bins corresponding to $\mathrm{pos}(u)$, as well as the at most $d/2 + o$ bins that the ${\color{blue} +2}$ and ${\color{blue} +1}$ balls were assigned to. 
        We have the following bounds on $(a)$ and $(b)$
        \begin{xalign}
            (a) &\leq {\frac{d}{2} + k + 2^n \choose {\frac{d}{2},k,2^n}} = \frac{\qty(\frac{d}{2} + k + 2^n)!}{\qty(\frac{d}{2})!k!(2^n)!} \leq \qty(2^{n+1})^{\frac{d}{2} + k} \leq \qty(2^{n+1})^{\frac{d}{2} + o}\\
            (b) &\leq {{r + 1 + \frac{d}{2} + o} \choose {r+1, \frac{d}{2},o}} \leq \qty(r + 1 + d/2 + o)^{\frac{d}{2} + o}\,.
    \end{xalign}
    Here, we are ignoring the effect of taking $\lfloor d/2 \rfloor$, since it only decreases the quantities.  Taking the product of $(a)$ and $(b)$ upper bounds the number of ways to assign the balls to bins, which completes the proof.    
    \end{proof}
\end{mathinlay}
 \newpage
\part{Polynomial-query algorithms generate quasi-even condensates}
\label{part:poly-query-implies-qec}
In this part, we complete the proof of the sampling probability upper bound by showing that polynomial-query algorithms are almost entirely supported on quasi-even condensates.  
\section{A bosonic compressed oracle technique}
\label{sec:bosons}

In this section, we describe a natural compressed oracle technique for a family of random sparse functions, like $S$, and identify a natural basis in which queries to $U$ (when $(S, U)$ are sampled from the distribution $\mathsf{Strong}$) has a simple purified description in terms of bosonic hopping operators.

\subsection{Initial bosonic state}
We first need to compute a purification of the choice of multiset $S$. As mentioned in the introduction, one technique is to sample a uniformly random multiset of size $\ell$ is preparing $\ell$ bosons initialized in the 0-momentum mode and then measuring in the position basis. Therefore, a purification of the oracle $S$ is simply $\ell$ bosons in the 0-momentum mode:
\begin{equation}
    \frac{1}{\sqrt{\ell!}} \qty(\tilde{a}_0^\dagger)^\ell \ket \vac.
\end{equation}

\begin{mathinlay}
\begin{claim}
    The state above is equal to the uniform superposition over all multisets of $\ell$ elements, where each element is sampled uniformly at random, independent of the other elements.
\end{claim}
\begin{proof}
    We can write the purification of a uniformly random multiset of $\ell$ elements in the Fock basis as follows:
    \begin{equation}
        \sum_{\substack{\ell_1, \ldots, \ell_{2^n} \in \ZZ_{\geq 0}\\ \sum_{x \in \bits^n} \ell_x = \ell}} \sqrt{p_{\ell_0, \ldots, \ell_{2^n - 1}}}\ket{\ell_0, \ldots, \ell_{2^n-1}}\,, 
    \end{equation}
    where $p_{\ell_0, \ldots, \ell_{2^n - 1}} = \frac{1}{2^{n\ell}}\frac{\ell!}{\prod_{x \in \bits^n} \ell_x!}$ is the probability of measuring getting $\ell_{x}$ copies of $x$ when sampling $\ell$ uniformly random strings.  We can also write out the result of applying $\ell$ creation operators in the $0$-momentum mode as follows:
    \begin{xalign}
        \frac{1}{\sqrt{\ell!}}\qty(\tilde{a}_0^\dagger)^\ell \ket{\vac} &= \frac{1}{\sqrt{\ell! 2^{n\ell}}} \sum_{x_1, \ldots, x_\ell \in \bits^{n}} \left(\prod_{i = 1}^\ell \widehat{a}_{x_i}^\dagger \right) \ket \vac\\
        &= \frac{1}{\sqrt{\ell! 2^{n\ell}}} \sum_{\substack{\ell_0, \ldots, \ell_{2^n-1} \in \ZZ_{\geq 0} \\ \sum_{x \in \bits^n} \ell_x = \ell}} \frac{\ell!}{\prod_{x \in \bits^n} \ell_x!} \left(\prod_{i = 1}^\ell \widehat{a}_{x_i}^\dagger \right) \ket \vac\\
        &= \frac{1}{\sqrt{2^{n\ell}}}\sum_{\substack{\ell_0, \ldots, \ell_{2^n-1} \in \ZZ_{\geq 0} \\ \sum_{x \in \bits^n} \ell_x = \ell}} \sqrt{\frac{\ell!}{\prod_{x \in \bits^n} \ell_x!}} \ket{\ell_0, \ldots, \ell_{2^n-1}}\,.
    \end{xalign}
    We first use the fact that, when sampling a random collection $x_1, \ldots, x_{\ell}$, the Fock basis state $\ket{\ell_0, \ldots, \ell_{2^n-1}}$ appears exactly $\ell! / \prod_{x \in \bits^n} \ell_x!$ times.  In the final line, we use the fact that by the definition of the creation operators, $\ket{\ell_0, \ldots, \ell_{2^n-1}} = \frac{1}{\prod_{x \in \bits^n} \ell_x!} \prod_{x \in \bits^n} \widehat{a}_{x}^{\ell_x} \ket{\vac}$.  Since these two states are equal, we have completed the proof.
\end{proof}
\end{mathinlay}

\subsection{Purified state of algorithm and oracle registers}
Next, we introduce the purification of the state of a query algorithm querying $U$, where $U$ is sampled from the distribution $\mathsf{Strong}$, \Cref{def:def-of-strong-dist}. For this, we will take inspiration from the compressed oracle technique introduced by Zhandry~\cite{zhandry2019record}. %
To construct the compressed oracle, we write out the following purification of the initial state of the system for both $S$ and $U$; we assume the state is expressed on purifying registers $\S$ and $\reg{U}$.
\begin{equation}
    \ket{\mathrm{init}}_{\reg{SU}} \defeq \frac{1}{\sqrt{\ell!}} \qty(\tilde{a}_{0}^{\dagger})^\ell \ket{\vac}_{\reg{S}} \otimes \ket{\bot}^{\otimes 2^n}_{\reg{U}}\,.
\end{equation}
Here, $\reg{U}$ is divided into $\bigotimes_{y \in \bits^n} \reg{U}_{y}$, where each single-qubit register $\reg{U}_y$ is initially in the state $\ket{\bot}$. We also define the initial state restricted to the $\reg{S}$ or $\reg{U}$ registers,
\begin{equation}
    \ket{\mathrm{init}_{S}}_{\reg{S}} \defeq \frac{1}{\sqrt{\ell!}} \qty(\tilde{a}_{0}^{\dagger})^\ell \ket{\vac}_{\reg{S}} \quad\text{and}\quad\ket{\mathrm{init}_{U}}_{\reg{U}} \defeq \ket{\bot}^{\otimes 2^n}_{\reg{U}}\,,
\end{equation}
We then define the following isometry acting on registers $\reg{US}$:
\begin{equation}\begin{split}
    \mathcal{V}_1 &\defeq \sum_{S} { \ketbra{\mathrm{tt}_S}_{\reg{S}}} \otimes \bigotimes_{y \in \bits^n}  \Bigg( \Bigg(\sqrt{1 - \frac{1}{2}\e^{-\kappa\gamma_y^{(S)}}} \ket{0} + \sqrt{\frac{1}{2}\e^{-\kappa\gamma_y^{(S)}}}\ket{1}\Bigg)\bra{\bot}_{\reg{U}_y} \\ & \hspace{2em} + \left(\sqrt{\frac{1}{2}\e^{-\kappa\gamma_y^{(S)}}} \ket{0} - \sqrt{1 - \frac{1}{2}\e^{-\kappa\gamma_y^{(S)}}}\ket{1}\right)\bra{\top}_{\reg{U}_y}\Bigg) \,.
\end{split}\end{equation}
Note that this is a unitary and only acts on $\reg{US}$.  Finally, we define a second isometry acting on registers $\reg{AU}$, with $\reg{A}$ acting as the algorithm's query register.
\begin{equation}
    \mathcal{V}_2 \defeq \sum_{\substack{y \in \bits^n\\ b \in \bits}} \sum_U (-1)^{b \cdot U(y)} \ketbra{b, y}_{\reg{A}} \otimes { \ketbra{\mathrm{tt}_U}_{\reg{U}}}\,.
\end{equation}
Then, we have the following:
\begin{mathinlay}
\begin{lemma}
\label[lemma]{lem:bot_oracle_identical}
    For all query algorithms $\mathcal{A}$, 
    \begin{equation}
        \Tr_{\reg{US}} \left[\mathcal{A}^{\mathcal{V}_1^{\dagger} \cdot \mathcal{V}_2 \cdot \mathcal{V}_1} (\ketbra{0, \mathrm{init}})\right] = \Exp_{S, U} \left[\mathcal{A}^{U}(\ketbra{0})\right]\,.
    \end{equation}
\end{lemma}
\begin{proof}
    Since $\mathcal{V}_1$ only acts on $\reg{US}$, it commutes with the unitaries that $\Aa$ applies. Moreover, the $\mathcal{V}_1$ from one query cancels out the $\mathcal{V}_1^\dagger$ from the next query. The result is that only the inner-most and outer-most $\mathcal{V}_1$ and $\mathcal{V}_1^\dagger$ are left. Thus, we can rewrite the left side of the equation as
    \begin{equation}
        \Tr_{\reg{US}}[\mathcal{V}_1^{\dagger} \Aa^{\mathcal{V}_2}(\ketbra{0} \otimes \mathcal{V}_1\ketbra{\mathrm{init}} \mathcal{V}_1^{\dagger}) \mathcal{V}_1] = \Tr_{\reg{US}}[\Aa^{\mathcal{V}_2}(\ketbra{0} \otimes \mathcal{V}_1\ketbra{\mathrm{init}} \mathcal{V}_1^{\dagger})]\,.
    \end{equation}
    Then we note that applying $\mathcal{V}_1$ to $\ket{\mathrm{init}}$ yields exactly the following state:
    \begin{equation}
        \mathcal{V}_1 \ket{\mathrm{init}} = \frac{1}{\sqrt{\ell! \cdot 2^{n\ell}}}\sum_{s_1, \ldots, s_\ell \in \bits^n} \bigotimes_{y \in \bits^n} \left(\sqrt{1 - 
        \frac{1}{2}\e^{-\kappa\gamma_y^{(S)}}} \ket{0} + \sqrt{\frac{1}{2}\e^{-\kappa\gamma_y^{(S)}}}\ket{1}\right)_{\reg{U}_y}\otimes \hat a_{s_1}^\dagger \ldots \hat a_{s_\ell}^\dagger \ket \vac \,.
    \end{equation}
    From here, it is clear to see that tracing out the $\reg{US}$ register of $\Aa$ yields a random $U$ and $S$ according to the distribution $\mathsf{Strong}$, and that $\Aa$ querying $\mathcal{V}_2$ yields an identical mixed state to $\Aa$ querying a $U$ (where $(S, U)$ are sampled according to $\mathsf{Strong}$).
\end{proof}
\end{mathinlay}
 Next, we make a simplification to this oracle that gives the oracle a nicer form.
\begin{mathinlay}
\begin{corollary}\label[corollary]{cor:krausoperators}
    Define the following Kraus operators acting on $\reg{S}$, parametrized by $y$.
    \begin{equation} \label{eq:kraus-operators}
        E_0^{(y)} \defeq \sum_{S} (1 - \e^{-\kappa\gamma_y^{(S)}}){ \ketbra{\mathrm{tt}_S}}_{\reg{S}}\quad\text{and}\quad E_1^{(y)} \defeq \sum_{S} \sqrt{\e^{-\kappa\gamma_y^{(S)}} (2 - \e^{-\kappa\gamma_y^{(S)}})} { \ketbra{\mathrm{tt}_S}}_{\reg{S}}\,.
    \end{equation}
    Then define $\mathcal{O}$ to be the following unitary acting on registers $\reg{AUS}$.
    \begin{equation} \label{eq:O-is-funny}
        \mathcal{O} \defeq \sum_{y \in \bits^n, b \in \bits} \ketbra{b, y}_{\reg{A}} \otimes \left(\tilde Z_{\reg{U}_y} \otimes \left(E_0^{(y)}\right)_\reg{S} + \tilde X_{\reg{U}_y} \otimes \left(E_1^{(y)}\right)_{\reg{S}}\right)^{b}\,,
    \end{equation}
    where $\tilde X$ and $\tilde Z$ are the usual Pauli operators except in the $\ket{\bot}$ and $\ket{\top}$ basis (as opposed to the $\ket{0}$ and $\ket{1}$ basis).  Note the exponent of $b$ acting on the unitary applied to the $\reg{US}$ registers. Then,
    \begin{equation}
       \Tr_{\reg{US}}[\Aa^{\mathcal{O}}(\ketbra{0, \mathrm{init}})] = \mathbb{E}_{S, U}\left[\Aa^{U}(\ketbra{0})\right]\,. 
    \end{equation}
\end{corollary}
\begin{proof}
Observe that $(E_0^{(y)})^2 + (E_1^{(y)})^2 = \id$,~$0 \preceq E_0^{(y)} \preceq \id $, and $0 \preceq E_1^{(y)} \preceq \id $. 
From \Cref{lem:bot_oracle_identical}, for any algorithm $\Aa$, querying $\mathcal{V}_1^{\dagger} \cdot \mathcal{V}_2 \cdot \mathcal{V}_1$ is identical to querying $U$ after tracing out $\reg{US}$. To prove the corollary, we show that $\mathcal{V}_1^{\dagger} \cdot \mathcal{V}_2 \cdot \mathcal{V}_1$ is equal to $\mathcal{O}$.  Writing it out, we will get the following:
\begin{xalign}
    &\mathcal{V}_1^{\dagger} \cdot \mathcal{V}_2 \cdot \mathcal{V}_1 \\
    &\hspace{3mm} \begin{matrix*}[l]\displaystyle=\sum_{S,y} \ketbra{1, y}_{\reg{A}} \otimes \Bigg(\left((1 - \e^{-\kappa\gamma_y^{(S)}})\ket{\bot} + \sqrt{\e^{-\kappa\gamma_y^{(S)}} (2 - \e^{-\kappa\gamma_y^{(S)}})} \ket{\top}\right)\bra{\bot}_{\reg{U}_y} \\ \displaystyle \hspace{5mm}+\left(\sqrt{\e^{-\kappa\gamma_y^{(S)}} (2 - \e^{-\kappa\gamma_y^{(S)}})}\ket{\bot} - (1 - \e^{-\kappa\gamma_y^{(S)}}) \ket{\top}\right)\bra{\top}_{\reg{U}_y}\Bigg)\otimes{ \ketbra{\mathrm{tt}_S}}_{\reg{S}} \\
    \hspace{5mm}+\displaystyle \sum_{y}\ketbra{0, y}_{\reg{A}} \otimes \id_{\reg{US}}\end{matrix*}  \\
    &\hspace{3mm}\begin{matrix*}[l]=\displaystyle \sum_{S,y} \ketbra{1,y}_{\reg{A}} \otimes \left((1 - \e^{-\kappa\gamma_y^{(S)}}) \cdot \tilde{Z} + \sqrt{\e^{-\kappa\gamma_y^{(S)}} (2 - \e^{-\kappa\gamma_y^{(S)}})} \cdot \tilde{X}\right)_{\reg{U}_y} \otimes { \ketbra{\mathrm{tt}_S}}_{\reg{S}}\\
    \hspace{5mm}+\displaystyle \sum_{y} \ketbra{0, y} \otimes \id_{\reg{US}}\,.\end{matrix*}
\end{xalign}
Here, $\tilde{Z} = \ketbra{\bot} - \ketbra{\top}$ and $\tilde{X} = \ket{\top}\!\!\bra{\bot} + \ket{\bot}\!\!\bra{\top}$ are the Pauli $Z$ and $X$ operators in the $\ket{\top}$ and $\ket{\bot}$ basis.  Direct calculation shows that this unitary squares to the identity as expected.  Then we can rewrite this using the definition of $E_0^{(y)}$ and $E_1^{(y)}$ as follows:
\begin{equation}
    \mathcal{V}_1^{\dagger} \cdot \mathcal{V}_2 \cdot \mathcal{V}_1 = \sum_{y,b} (\ketbra{b,y}_{\reg{A}} \otimes \id_{\reg{B}}) \otimes \left(\tilde{Z}_{\reg{U}_y} \otimes \left(E_0^{(y)}\right)_{\reg{S}} + \tilde{X}_{\reg{U}_y} \otimes \left(E_1^{(y)}\right)_{\reg{S}}\right)^{b}\,.
\end{equation}
Applying \Cref{lem:bot_oracle_identical} completes the proof.
\end{proof}
\end{mathinlay}
This implies that we can write the purified state of any algorithm $\mathcal{A}$ that only queries $U$ in a simple form. Here we assume the query algorithm can be expressed as a repeating sequence of oracle queries and unitaries $A_{\reg{A}}$ (without loss of generality, we can assume the same unitary $A$ is applied each time):
\begin{equation}
    \underbrace{A \mathcal{O} A \mathcal{O} \ldots \mathcal{O}}_{T~\mathrm{times}} A \ket{0}_{\reg{A}} \ket{\mathrm{init}}_{\reg{US}}
\end{equation}
where there are $T$ queries to the oracle $U$.
The query access to controlled-$U$ described by~\eqref{eq:O-is-funny} can be further abbreviated as
\begin{equation}
    \Oo = \sum_{\substack{y \in \bits^n\\b \in \bits}} \ketbra{b, y}_{\reg{A}} \otimes \qty(\sum_{x \in \bits} \underbrace{\qty(\tilde{Z}^{1-x}  + \tilde{X}^x)_{\reg{U}_{y}}}_{\defeq (K^x)_{\reg{U}_{y}}} \otimes (E_x^{(y)})_{\reg{S}})^{b},
\end{equation}
where we define $K^x$ for notational simplicity.  Note the exponent $b$ in the large parenthetical. Then, an alternating sequence of algorithm and queries can be calculated as follows. 
Given $\y = (y_1, \ldots, y_T)$, $\bb = (b_1, \ldots, b_T)$, let $\y^\bb = (y_i : b_i = 1)$ be the ordered tuple of length $\abs{\bb} \defeq \sum_{i} b_i$ corresponding to indices of $\y$ where $\bb$ is $1$, with $y^{\bb}_i$ being the $i$'th entry of this tuple.  
Then given $\x = (x_1, \ldots, x_{\abs{\bb}})$, define the abbreviations:

\begin{equation}
    (K^\x)_{\y, \bb} \defeq \prod_{i = 1}^{\abs{\bb}} (K^{x_i})_{\reg{U}_{y^{\bb}_i}}, \qquad 
    E_\x^{(\y, \bb)} \defeq \prod_{i = 1}^{\abs{\bb}} E_{x_i}^{(y^{\bb}_i)}, \qquad 
    A_{\y, \bb} \defeq \prod_{i = T}^1 A \cdot \ketbra{b_i, y_i}\,.
\end{equation}
Then, the post-query state will be
\begin{equation}\label{eq:algorithmstate}
    \ket*{\psi_{\mathrm{PQ}}} 
    = \ket*{\psi_{\mathrm{Post-Query}}} 
    \defeq \sum_{\substack{\y \in {\bits^n}^{T}\\ \bb \in \bits^{T}}} 
    \qty[\qty(A_{\y, \bb})_{\reg{A}} \otimes \qty(\sum_{\x \in \bits^{\abs{\bb}}} (K^\x)_{\y, \bb} \otimes E_\x^{(\y, \bb)} )] A \ket{0}_{\reg{A}} \ket{\mathrm{init}}_{\reg{US}}.
\end{equation}
The abbreviated notation will prove itself useful when we start approximating the oracle queries with polynomials.
Note that in the rest of the section, we drop the register indices when they are clear from context, and whenever an operator acts as identity on a register, we omit the $\otimes \id$ part of the operator.  
Having developed an expression for the state after applying $T$ queries, we now press forward and study how the post-query state (expressed in~\eqref{eq:algorithmstate}) is almost entirely supported on quasi-even condensates. This will be the content of the next section. By doing so, we can appeal to~\Cref{thm:main-sampling-upper-bound} to conclude a probability upper bound.

\subsection{Action of the oracle in the momentum basis}

We see that $\mathcal{O}$ consists of Kraus operators $E_0^{(y)}$ and $E_1^{(y)}$ that can be thought of as functions of a map $\ket{\mathrm{tt}_S} \mapsto \gamma^{(S)}_{y} \ket{\mathrm{tt}_S}$.  We now examine the action of this map directly before examining the action of $\mathcal{O}$. By doing so, we will build a natural framework for understanding how to study maps that are more complex functions of $\gamma^{(S)}_{y}$ such as $\ket{\mathrm{tt}_S}$ to either $E_0^{(y)} \ket{\mathrm{tt}_S}$ or $E_1^{(y)} \ket{\mathrm{tt}_S}$.
This action, perhaps surprisingly, has a natural interpretation in the momentum space. To see this, we define the ``single $y$-momentum hopping operator'' $\tilde{\G}_y$ and the ``double $y$-momentum hopping operator'' $\tilde{\H}_y$ as\footnote{We use normalized versions of the single and double $y$-momentum hopping operators. This is because we are only considering systems of $\ell$ bosons. A situation with a variable number of bosons might use the definitions without the $1/\sqrt{\ell}$ and $1/\ell$ constants, respectively.}:
\begin{xalign}
    \tilde{\G}_y &\defeq \frac{1}{\sqrt{\ell}} \sum_{x\in \bits^n} \tilde{a}_{x\oplus y}^\dagger \tilde{a}_{x}. \label{eq:single_jump_operator} \\
    \tilde{\H}_y &\defeq \frac{1}{\ell} \sum_{x, x' \in \bits^n} \tilde{a}_{x\oplus y}^\dagger \tilde{a}_{x'\oplus y}^\dagger \tilde{a}_{x} \tilde{a}_{x'}. \label{eq:double_jump_operator}
\end{xalign}
We also refer to these as single-hopping or double-hopping operators when $y$ is clear from context. (In this work, we never use the analogous \textit{position} hopping operators.)
Morally, for a bosonic system $\ket{\psi}$, we can interpret $\tilde{\G}_y$ as adding momentum $y$ (in superposition) to each boson and $\tilde{\H}_y \ket{\psi}$ as adding momentum $y$ (in superposition) to each pair of bosons. Observe that action by $\tilde{\H}_y \ket{\psi}$ increases the momentum of the entire system by $2y = 0$.  Therefore, the total momentum of the system is conserved under any function of $\tilde{\H}_y$. Second, note that the single- and double-momentum hopping operators are related:
\begin{xalign}
    \tilde{\G}_y^2 &= \frac{1}{\ell} \sum_{x,x'\in \bits^n} \tilde{a}_{x\oplus y}^\dagger \tilde{a}_{x}  \tilde{a}_{x'\oplus y}^\dagger \tilde{a}_{x'} \\
    &= \frac{1}{\ell} \sum_{x,x'\in \bits^n} \tilde{a}_{x\oplus y}^\dagger  \qty(\tilde{a}_{x'\oplus y}^\dagger \tilde{a}_{x} + \delta_{x, x' \oplus y}) \tilde{a}_{x'} \\
    &= \frac{1}{\ell} \sum_{x,x'\in \bits^n} \tilde{a}_{x\oplus y}^\dagger  \tilde{a}_{x'\oplus y}^\dagger \tilde{a}_{x} \tilde{a}_{x'} + \frac{1}{\ell} \sum_{x \in \bits^n} \tilde{a}_{x}^\dagger \tilde{a}_x \\
    &= \tilde{\H}_y + \frac{\tilde{N}}{\ell}.
\end{xalign}
Within the subspace of $\ell$ bosons,
\begin{equation}
\tilde{\G}_y^2 = \tilde{\H}_y + \id.
\end{equation}
We prove the following lemma relating the hopping operators to the map $\ket{\mathrm{tt}_S} \mapsto \gamma_y^{(S)}\ket{\mathrm{tt}_S}$.

\begin{lemma}[Single hopping twice applies $\gamma_y^{(S)}$]
    For all $y \in \bits^{n}$, the following holds.
    \begin{equation}
        \tilde{\G}_y^2 \hat a_{s_1}^\dagger \ldots \hat a_{s_\ell}^\dagger \ket \vac =  \gamma_y^{(S)} \hat a_{s_1}^\dagger \ldots \hat a_{s_\ell}^\dagger \ket \vac\,.
    \end{equation}
\end{lemma}
\begin{proof}
We can directly compute the action of $\tilde{\G}_y$ on a position basis state to relate it to $\gamma_y^{(S)}$.  
\begin{xalign}[eq:action-of-single-hop-operator]
    \tilde{\G}_y \hat a_{s_1}^\dagger \ldots \hat a_{s_\ell}^\dagger \ket \vac &= \frac{1}{\sqrt{\ell}}\sum_{x \in \bits^n} \tilde{a}_{x\oplus y}^\dagger \tilde{a}_{x} \hat a_{s_1}^\dagger \ldots \hat a_{s_\ell}^\dagger \ket \vac \\
    &= \frac{1}{\sqrt{\ell N^\ell}} \sum_{x \in \bits^n} \sum_{t_1, \ldots, t_\ell} (-1)^{t_1 \cdot s_1 + \ldots + t_\ell \cdot s_\ell} \tilde{a}_{x\oplus y}^\dagger \tilde{a}_{x}  \tilde{a}_{t_1}^\dagger \ldots \tilde{a}_{t_\ell}^\dagger \ket \vac \\
    &= \frac{1}{\sqrt{\ell N^\ell}} \sum_{i = 1}^\ell \sum_{t_1, \ldots, t_\ell} (-1)^{t_1\cdot s_1 + \ldots + t_\ell\cdot s_\ell} \tilde{a}_{t_i \oplus y}^\dagger \qty( \prod_{k : k \neq i} \tilde{a}_{t_k}^\dagger)  \ket \vac \\
    &= \frac{1}{\sqrt{\ell N^\ell}} \sum_{i = 1}^\ell \sum_{t_1, \ldots, t_\ell} (-1)^{y \cdot s_i} (-1)^{t_1\cdot s_1 + \ldots + t_\ell \cdot s_\ell} \tilde{a}_{t_1}^\dagger \ldots \tilde{a}_{t_\ell}^\dagger \ket \vac \\
    &=  \qty( \frac{1}{\sqrt \ell} \sum_{i = 1}^\ell (-1)^{y \cdot s_i} ) \hat a_{s_1}^\dagger \ldots \hat a_{s_\ell}^\dagger \ket \vac\,.
\end{xalign}
Here, in the second line we applied the definition of the Fourier basis creation operators, $\hat{a}_{s}^{\dagger} = \frac{1}{\sqrt{N}} \sum_{t} (-1)^{s \cdot t} \tilde{a}_{t}^{\dagger}$, and in the third line we applied the commutation relations between the creation and annihiliation operators.  Since we have shown that $\tilde{\G}_y$ is diagonal in the position basis, acting by $\tilde{\G}_y^2$ simply squares the coefficient, so recalling the definition of $\gamma_y^{(S)}$ (defined in~\eqref{eq:def-of-gamma-s}) completes the proof.
\end{proof}
As an immediate corollary, we have that within the subspace of $\ell$ bosons, both $\tilde{\G}_y$ and $\tilde{\H}_y$ are diagonal in the position Fock basis, and that the operator that maps $\ket{\mathrm{tt}_S}$ to $\gamma_y^{(S)} \ket{\mathrm{tt}_S}$ is in fact equal to $\tilde{\G}_y^2 = \tilde{\H}_y + \id$. The formulation of this map in terms of both the single- and double-momentum hopping operators will be convenient. In some cases, one formulation will be more useful than the other.

The double-momentum hopping operator is particularly nice to analyze. Recall our definition of $\mathsf{Con}_r$ as the projector onto states that are $r$-condensates. $\mathsf{Con}_r$ is a projector diagonal in the momentum Fock basis and with the projector capturing all momentum Fock basis states that have a total of $\ell$ bosons and at least $\ell - r$ bosons in the 0-momentum mode. Then, it is easy to observe that

\begin{fact}[Hopping operators preserve condensates] \label[fact]{fact:double-hop-doesnt-hop-much}
For any $r \geq 0$ and $k \geq 0$,
\begin{equation}
\qty(\id - \mathsf{Con}_{r+2k}) \cdot \tilde{\H}_{y_k} \ldots \tilde{\H}_{y_1} \cdot \mathsf{Con}_r = 0.
\end{equation}
Second, for any polynomial $p: \RR^{2^n} \rightarrow \RR$, let $M_p$ be the operator $p(\gamma_0^{(S)}, \ldots, \gamma_{2^n-1}^{(S)}) \ketbra{\mathrm{tt}_S}$. Then,
\begin{equation}
    \qty(\id - \mathsf{Con}_{r+2 \deg p}) \cdot M_p \cdot \mathsf{Con}_r = 0.
\end{equation}
\end{fact}
\begin{proof}
The first equation follows as each term of $\tilde{\H}_y$ can at most move 2 bosons from the 0-momentum mode.
For the second equation, observe that the action $M_p$ can be expressed as
\begin{equation}
    p \qty( \tilde{\G}_0^2, \ldots, \tilde{\G}^2_{2^n-1}) = p \qty(  \tilde{\H}_0 + \id, \ldots, \tilde{\H}_{2^n-1} + \id )
\end{equation}
due to~\eqref{eq:action-of-single-hop-operator} which can then be expressed as a linear combination of terms of the form $\tilde{\H}_{y_{\leq \deg p}} \ldots \tilde{\H}_{y_1}$. Combining with the first equation completes the proof.
\end{proof}
 Next, we prove norm bounds on the Hamiltonians $\tilde{\G}_y$ and $\tilde{\H}_y$ on the subspaces of $\ell$ bosons as well as the subspace $\mathsf{Con}_r$. By the action of creation and annihilation operators on Fock states calculated in~\eqref{eq:creation_annihilation_definition}, on the subspace of $\ell$ bosons, $\norm*{\tilde{\G}_y} \leq \sqrt{\ell}$ and, therefore, $\norm*{\tilde{\H}_y} \leq \ell - 1$. Unfortunately, these norms are too large for our polynomial approximations to handle. Luckily, both norms are significantly smaller within the space of condensates.
\begin{fact}[Norm of hopping operators for condensates] \label[fact]{fact:h-tilde-norm-bound}
    For $y \neq 0$ and all $r\geq 0$, the action of the single- and double-momentum hopping operator has small norm on the space of $r$-condensates. More specifically,
    \begin{xalign}
        \norm{\tilde{\G}_y \cdot \mathsf{Con}_r}_\mathrm{op} &\leq \sqrt{r} + \sqrt{2 + 4r}\,,\text{and} \\
        \norm{\tilde{\H}_y \cdot \mathsf{Con}_r}_\mathrm{op} &\leq 9r + 9\,.
    \end{xalign}
\end{fact}

\begin{proof}
    Let us define
    \begin{equation}
        \tilde{M}_y \defeq \frac{1}{\sqrt \ell} \qty( \tilde a_{y}^\dagger \tilde{a}_0 + \tilde a_{0}^\dagger \tilde{a}_y), \quad\text{and}\quad \tilde{M}_y' \defeq \frac{1}{\sqrt \ell} \sum_{x \notin \{0, y\}} \tilde a_{x \oplus y}^\dagger \tilde{a}_x\,,
    \end{equation}
    and therefore bounding the norm of $\tilde{\G}_y \cdot \mathsf{Con}_r$ can be achieved by bounding the norms of both $\tilde{M}_y\cdot \mathsf{Con}_r$ and $\tilde{M}_y' \cdot \mathsf{Con}_r$. By construction, $\tilde{M}_y$ only depends on the bosons in the 0- and $y$-momentum modes and $\tilde{M}_y'$ only depends on the other momentum modes.
    Observe that $\tilde{\G}_y$ and $\tilde{M}_y$ both commute with $\tilde{n}_0 + \tilde{n}_y$. This can be verified by direct calculation using the commutation relations given in~\eqref{eq:bosonic-commutation2}. Therefore, $\tilde{M}_y$ preserves the sum of the number of bosons in the 0- and $y$-momentum modes. Equivalently, if we divide the bosonic subspace of $\ell$ bosons into orthogonal subspaces based on the value of $\tilde{n}_0 + \tilde{n}_y$, we find that $\tilde{M}_y$ is block-diagonal with respect to this direct sum. Therefore, we can restrict our analysis to states where there are a total of $L$ bosons combined in the $0$- and $y$-momentum modes for $\ell - L \leq r$. 
    Furthermore, since $\tilde{M}_y$ only depends on the bosons in the 0- and $y$-momentum modes, for computing the spectral norm of $\tilde{M}_y \cdot \mathsf{Con}_r$ we can restrict our analysis to states of a total of $L$ bosons entirely contained in the $0$- and $y$-momentum modes\footnote{To formalize this, observe that the Hilbert space corresponding to $\mathsf{Con}_r$ can be factorized as
    \begin{equation}
        \mathsf{Con}_r = \bigoplus_{L = \ell - r}^\ell \qty( \Hh_{\{0,y\}}^{(L)} \otimes \Hh_{\mathrm{rest}}^{(\ell - L)})
    \end{equation}
    where $\Hh_{\{0,y\}}^{(L)}$ is the Fock space of the 0- and $y$-momentum modes restricted to containing exactly $L$ bosons and $\Hh_{\mathsf{rest}}^{(\ell - L)}$ is the Fock space for the remainder of the modes restricted to $\ell - L$ bosons. Since $\tilde{M}_y$ commutes with $\tilde{n}_0 + \tilde{n_y}$ and does not depend on the state of the other modes, we can also factorize $\tilde{M}_y$ as
    \begin{equation}
        \tilde{M}_y = \bigoplus_{L = \ell - r}^\ell \qty( \tilde{M}_y^{(L)} \otimes \id_{\mathrm{rest}}).
    \end{equation}
    Therefore, we can restrict to analyzing only $\tilde{M}_y^{(L)}$ for every $L$.}.
    Every such state can be expressed as a superposition of states $\ket*{L-j,j}$ representing the number of bosons in the $0$- and $y$-momentum modes.
    Then, we can then calculate the effect of $\tilde{M}_y$:
\begin{xalign}
    \sqrt{\ell} \cdot \tilde{M}_y|\phi\rangle&=\sum_{j = 0}^{L} \left(\tilde{a}_y^\dagger \tilde{a}_0 + \tilde{a}_0^\dagger \tilde{a}_y \right)\alpha_j|L-j,j\rangle\\
    &=\sum_{j = 0}^{L} \alpha_j\left(\sqrt{(L-j)(j+1)}|L-j-1,j+1\rangle+\sqrt{(L-j+1)j}|L-j+1,j-1\rangle\right)\\
    &\begin{matrix*}[l]=\displaystyle\left(\sum_{j' = 1}^{L} \alpha_{j'-1}\sqrt{(L-j'+1)j'}|L-j',j'\rangle\right) \\
    \hspace{2em}+\displaystyle\left(\sum_{j'' = 0}^{L-1}\alpha_{j''+1}\sqrt{(L-j'')(j''+1)}|L-j'',j''\rangle\right)\end{matrix*}\\
    &=\sum_{j = 0}^{L}\left( \alpha_{j-1}\sqrt{(L-j+1)j}+\alpha_{j+1}\sqrt{(L-j)(j+1)}\right)|L-j,j\rangle
\end{xalign}
Above, in the first line, we are applying the annihilation and creation operators using ~\eqref{eq:creation_annihilation_definition}. Then we broke the expression into two terms, and did a change of variable $j'=j+1$ and $j''=j-1$, with $\alpha_{-1} = \alpha_{L+1} = 0$. Then we relabeled $j''\rightarrow j$ and $j'\rightarrow j$.
Now we can bound the norm of $\tilde{M}_y|\phi\rangle$, which is:
\begin{xalign}
\left\|\tilde{M}_y|\phi\rangle\right\|^2&=\frac{1}{\ell}\sum_{j = 0}^{L} \left|\alpha_{j-1}\sqrt{(L-j+1)j}+\alpha_{j+1}\sqrt{(L-j)(j+1)}\right|^2\\
&\leq \frac{2}{\ell}\sum_{j = 0}^{L} \left|\alpha_{j-1}\sqrt{(L-j+1)j}\right|^2+\left|\alpha_{j+1}\sqrt{(L-j)(j+1)}\right|^2\\
&=\frac{2}{\ell}\sum_{j = 0}^{L} |\alpha_{j-1}|^2(L-j+1)j + |\alpha_{j+1}|^2(L-j)(j+1)\\
&=\frac{2}{\ell}\left(\sum_{j' = 0}^{L}|\alpha_{j'}|^2(L-j')(j'+1)\right)+\frac{2}{\ell}\left(\sum_{j''=0}^{L}|\alpha_{j''}|^2(L-j''+1)j''\right)\\
&=\frac{2}{\ell}\sum_{j=0}^{L}|\alpha_j|^2 (L+2jL-2j^2)\\
&\leq \frac{2}{\ell}\sum_{j=0}^{L} |\alpha_j|^2(L+2rL)\\
&=2L(1+2r)/\ell \\
&\leq 2+4r\,.
\end{xalign}
Here, we first used that $|a+b|^2\leq 2|a|^2+2|b|^2$, then broke up the sum and performed a change of variable $j'=j-1$ and $j''=j+1$. Then we relabeled $j''\rightarrow j$ and $j'\rightarrow j$, and used the bound $j\leq r$, $-j^2\leq 0$, $L\leq \ell$, and finally that $\sum_j|\alpha_j|^2=1$.  We also use the fact that $(L-j)(j+1) + (L-j+1)j = L - 2Lj - 2j^2$. This gives the bound on $\norm*{\tilde{M}_y  \cdot \mathsf{Con}_r}_\mathrm{op}$.

Similarly, to bound $\norm*{\tilde{M}_y' \cdot \mathsf{Con}_r}_\mathrm{op}$ observe that since $\tilde{M}_y'$ only depends on the modes outside of $0$ and $y$ and there are at most $r$ bosons in said modes, then by concavity of the square root function and~\eqref{eq:creation_annihilation_definition}, the total norm is at most $r/\sqrt{\ell} \leq r/\sqrt{r} \leq \sqrt{r}$ for $r\leq \ell$.
This proves 
the bound for $\norm*{\tilde{M}_y' \cdot \mathsf{Con}_r}_{\mathrm{op}}$ for $r\leq \ell$.
When $r\geq \ell$ we can use the fact that $\norm{\tilde{\G}_y}\leq \sqrt{\ell}\leq \sqrt{r}$.

To get the bound on $\norm*{\tilde{\H}_y \cdot \mathsf{Con}_r}_\mathrm{op}$, we note that $\tilde{\G}_y|\phi\rangle$ is an $(r+1)$-condensate, so we can apply the bound on $\norm*{\tilde{\G}_y \cdot \mathsf{Con}_r}_\mathrm{op}$ twice:

\begin{xalign}
    \norm{\tilde\H_y \ket{\phi}} &\leq 1 + \norm{\tilde{\G}_y^2 \ket{\phi}} \\
    &\leq 1 + \left(\sqrt{r} + \sqrt{2 + 4r}\right)\left(\sqrt{r+1} + \sqrt{2 + 4(r+1)}\right) \\
    &\leq 9r+9\,.
\end{xalign}
Here in the first line we use the fact that $\tilde{\H}_y = \tilde{\G}_y^2 - \id$ and the triangle inequality, and then take the crude upper bound of $9r + 9$.
\end{proof}

\section{Proving the condensate property}

In this section, we prove that, for any $T$-query algorithm, the purified post-query state of the algorithm is mostly supported on condensates, i.e., that most of the bosons (in the momentum basis) remain in the $0$-momentum mode.  Recall that the actual (purified) post-query state of an algorithm $A$ is given by \eqref{eq:algorithmstate}, 
\begin{equation}
    \ket*{\psi_\mathrm{PQ}} = \sum_{\y, \bb} \qty[\qty(A_{\y, \bb})_{\reg{A}} \otimes \qty(\sum_{\x} (K^\x)_{\y, \bb} \otimes e_{\x}\left(\tilde{\G}_{y^{\bb}_1}^2,\ldots, \tilde{\G}_{y^{\bb}_{\abs{\bb}}}^2\right))] A \ket{0}_{\reg{A}} \ket{\mathrm{init}}_{\reg{US}}\,.
\end{equation}
Our primary goal in this section is to show that $\norm{\ket*{\psi_{\mathrm{PQ}}} - \mathsf{Con}_{R} \cdot \ket*{\psi_{\mathrm{PQ}}}}$ is very small for some suitably large (but polynomial in $n$) $R$.  However, it turns out that we will be able to prove a much stronger statement: that restricting the operators $\widetilde{\G}$ to $\mathsf{Con}_R$ subspace does not change the state much.  Formally we define the condensate sandwiched post-query states as follows.
\begin{definition}[Condensate sandwiched post-query states]
For integers $R$, $r \geq 0$, and a quantum adversary $A$ making $T$-queries to an oracle, define the $(R, r)$-sandwiched state to be, 
    \begin{equation}
    \ket*{\tilde{\psi}_{R,r}}\defeq \sum_{\y, \bb} \qty[\qty(A_{\y, \bb})_{\reg{A}} \otimes \qty(\sum_{\x} (K^\x)_{\y, \bb} \otimes \prod_{i=1}^{\abs{\bb}} \mathsf{Con}_r\cdot e_{x_i}\qty(\mathsf{Con}_R\cdot \tilde{\G}_{y^{\bb}_i}^2\cdot\mathsf{Con}_R )\cdot \mathsf{Con}_r )] A \ket{0}_{\reg{A}} \ket{\mathrm{init}}_{\reg{US}}. 
\end{equation}
\end{definition}
It is clear from the definition that restricting $\ket*{\widetilde{\psi}_{R, r}}$ to the $\mathsf{Con}_r$ subspace does not change the state at all, so showing that $\ket*{\widetilde{\psi}_{R, r}}$ is close to $\ket*{\psi_{\mathrm{PQ}}}$ will accomplish our goal.  Our main theorem is the following.
\begin{theorem}[Sandwiching theorem]\label[theorem]{cor:sandwiching}
For every $\iota > 0$, there exist integers $ r =O(n^2T^5\ln^3(T)\ln^2(1/\iota))$ and $R=O(n^3T^6\ln^4(T)\ln^3(1/\iota))$, such that 
\begin{equation}
\norm{\ket*{\psi_{\mathrm{PQ}}} - \ket*{\tilde{\psi}_{R,r}}} \leq \iota.
\end{equation}
\end{theorem}
The main technique used in the proof is a clever polynomial approximation to the exponential function.  To see what we mean by polynomial approximation in this context, we define function post-query states as follows.
\begin{definition}[Function post-query state]
    Let $\mathcal{A}$ be a $T$-query quantum algorithm that implements unitary $A$ between queries.  Given a family of functions $f_\x = f_{(x_1, \ldots, x_T)}$, we define the $f$-post-query state
\begin{align} \label{eq:polynomial_approximation_state_v3}
	\ket*{\psi_f} &\defeq \sum_{\y, \bb} \qty[\qty(A_{\y, \bb})_{\reg{A}} \otimes \qty(\sum_{\x} (K^\x)_{\y, \bb} \otimes f_{\x}\left(\tilde{\G}_{y^{\bb}_1}^2,\ldots, \tilde{\G}_{y^{\bb}_{\abs{\bb}}}^2\right) )] A \ket{0}_{\reg{A}} \ket{\mathrm{init}}_{\reg{US}}.
\end{align}
\end{definition}
Evaluating the functions $f(X)$ for a diagonalizable matrix $X$ happens by applying diagonalizing $X$ and applying $f$ to the diagonal entries. 
Intuitively, the function post-query state corresponds to the (not necessarily normalized) purified state of an algorithm querying the $U$ oracle, if $U_{y}$ was sampled according to $f(\gamma_{y}^{(S)})$.  To relate these states to our original post-query state, we can write down the functions corresponding to $\ket*{\psi_{\mathrm{PQ}}}$.
\begin{definition}[Kraus operator eigenvalue functions]
Define the pair of functions $e_0$ and $e_1$ as follows,
\begin{xalign}
    e_0(\gamma) &\defeq 1-\e^{-\kappa \gamma}, \\
    e_1(\gamma) &\defeq \sqrt{\e^{-\kappa \gamma}(2-\e^{-\kappa \gamma})} = \sqrt{2} \e^{-\kappa \gamma/2} \sqrt{1 - \e^{-\kappa \gamma}/2}\,.
\end{xalign}
Note that we can express the Kraus operators $E_x^{(y)}$ as $\sum_S e_x(\gamma_{y}^{(S)}) \ketbra{\mathrm{tt}_S}$ to simplify the expressions in~\eqref{eq:kraus-operators}.
\end{definition}
For all $B \in \ZZ_{\geq 0}$ and $\x = (x_1, \ldots, x_B)$, we also use the notational shorthand $e_\x(\gamma_1, \ldots, \gamma_B) \defeq \prod_{i=1}^B e_{x_i}(\gamma_i)$.
Using these definition, we can express the post-query state of any algorithm ( \eqref{eq:algorithmstate}) as the function state corresponding to $e_{\x}$, 
\begin{equation}
    \ket*{\psi_\mathrm{PQ}} = \ket*{\psi_{e}}\,.
\end{equation}
We can now state the notion of polynomial approximation that we want.
\begin{theorem}[Polynomial approximation of post-query state]\label{theorem:poly_approximation_v2}
	For every $\iota > 0$, there exists a family of polynomials $\mathrm{AKraus}_{\x}$ (that implicitly depends on $\iota$) in the operators $\tilde{\G}_{y_i}^2$ with 
        \begin{equation}
    \deg \left(\mathrm{AKraus}_{\x}\right) =O(n^2T^5\ln(T)\ln^2(1/\iota))\,,
    \end{equation}
    such that we can bound the distance between the true post-query state and the polynomial approximation by
	\begin{equation}
		\norm{\ket*{\psi_\mathrm{PQ}} - \ket*{\psi_{\mathrm{AKraus}}}} \leq \iota\,.
	\end{equation}
\end{theorem}

\subsection{Overview of the polynomial approximation} \label{subsec:polynomial_overview} Our starting point for the polynomial approximation is the Taylor series approximation of the exponential.  Roughly speaking, taking the $d$'th level truncation to the Taylor series is a good approximation for inputs up to $O(d)$, and has error that scales as $O(e^{x})$ for $x$ much larger than $d$.  Since the best bound we have on the norm of $\widetilde{\G}_y$ is $\norm{\widetilde{\G}_y}_{\mathrm{op}} \leq \sqrt{\ell}$, this would seem to imply that, if we want to apply the Taylor series approximation to the exponential functions in the post-query state, we must take the degree of the approximation to be very high. 

However, we are not trying to approximate $\widetilde{\G}^2_y$ on all inputs, but rather only on the distribution of $\gamma_{y}^{(S)}$ induced by $\ket{\mathrm{init}_{S}}$.  Since $\gamma_y^{(S)}$ is a sum of random $\pm 1$ variables, standard concentration inequalities tell us that the probability that $\gamma_y^{(S)}$ is much larger than $x$ scales as $e^{-x}$, which exactly matches the error in the Taylor series.  Indeed, in the $T = 1$ case, this argument (sometimes called chaining) can be used to show that the Taylor series with degree $O(T \log(1/\epsilon))$ yields a good polynomial approximation state to the original post-query state.  

Going beyond $T = 1$, the state we are trying to approximate now has factors that look like \begin{equation}\exp\left(\widetilde{\G}^2_{y_1} + \ldots + \widetilde{\G}^2_{y_T}\right)\,.\end{equation} From the same logic as before, if the factors $\gamma_{y_i}^{(S)}$ were independent of each other, we would expect that the Taylor series provide a good approximation to the function $\exp\left(\left(\widetilde{\G}_{y_1} + \ldots + \widetilde{\G}_{y_T}\right) / T\right)$, i.e., the average of the $\widetilde{\G}_{y_i}$ operators.  This motivates the idea, originally proposed by Narayan~\cite{narayanan2024improved}, of first writing $\exp(-z)$ as $\left(\exp(-z / T)\right)^{T}$, and then taking the Taylor series approximation for $\exp(-z / T)$.  

Let us first examine what happens when we approximate $\exp(-z/s)$.  Instead of applying the tail bound to $\gamma_{y}^{(S)}$, our chaining argument can be applied to the operators $\widetilde{\G}^2_y$ themselves. We use the observation (proven in \Cref{fact:h-tilde-norm-bound}) that when restricted to the $r$-condensate subspace, $\widetilde{\G}^2_y$ has bounded norm \emph{and} does not increase the number of non-zero bosons by much. When we apply this to the definition of the Taylor series, we find that when we use the Taylor series to approximate $e^{-z / s} = \sum_{m = 0}^{\infty} \frac{(-z)^{m}}{s^{m}m!}$ for an appropriately chosen $s$ depending on $r$, the factor of $\frac{1}{s^{m} m!}$ beats down the norm of $\widetilde{\G}^2_y \cdot \mathsf{Con}_r$, allowing us to prove a strong bound on the approximation error.  

To go from our approximation of $\exp(-z / s)$ to an approximation of $\exp(-z)$, we would need to take $s$ copies of $\exp(-z/s)$, starting from the $\mathsf{Con}_0$ subspace.  However, every time we apply $\exp(-z/s)$, the number of non-zero bosons increases, and we must increase the degree of our polynomial approximation of the next copy of $\exp(-z/s)$.  It turns out that if we take $s$ copies, the final polynomial approximation degree actually explodes exponentially in $s$.  To fix this, we notice that instead of simply taking $s$ copies of $\exp(-z/s)$, we can actually approximate the function $z^{s}$ with a Chebyshev polynomial of degree $\sqrt{s}$.  Applying this approximation allows us to take fewer copies of $\exp(-z/s)$, which lets us arrive at a low-degree approximation to our post-query state.

\begin{remark}
    Using $\e^z=(\e^{-z/s})^s$, approximating $\e^{-z/s}$ via a truncated Taylor series and $z^s$ by a truncated Chebyshev series is identical to a construction of a flat approximation due to Narayanan~\cite{narayanan2024improved}. 
   Using this idea, Narayanan achieves an exponential improvement over the original flat approximation of the exponential function in Ref.~\cite{bakshi2024learning}. 
 An earlier version of our manuscript used Narayanan's result as a black box and combined the flat approximation with tail bounds on $\gamma^{(S)}_y$ in the position basis. 
\end{remark}

Once we have established the polynomial approximations, the proof of~\Cref{cor:sandwiching} works by going back and forth between the polynomial approximation and the original Kraus operators. 
As the action of the $\tilde{\G}^2_y$ Hamiltonians cannot move more than $2$ bosons from the $0$-momentum mode, we can move/insert the projectors $\mathsf{Con}_r$ virtually anywhere in the polynomial approximation for sufficiently large $r$. 
We first use this to move the projectors $\mathsf{Con}_r$ between the Kraus operators and then exploit that the strategy to prove~\Cref{theorem:poly_approximation_v2} also works if we start in $\mathsf{Con}_r$ instead of $\mathsf{Con}_0$ by choosing a degree that depends on $r$ polynomially.
Moreover, we find that the polynomial approximation remains valid after replacing $\tilde{\G}_y^2$ by the sandwiched $\mathsf{Con}_R\cdot \tilde{\G}^2_y \cdot\mathsf{Con}_R$ in a single Kraus operator.
This observation roughly allows us to approximate $\mathsf{Con}_r\cdot e_x(\tilde{\G}^2_y)\cdot \mathsf{Con}_r$ by a polynomial depending on $r$ and $\iota$, move $\mathsf{Con}_R$ into the argument for large enough $R$, and finally reinsert the original Kraus operator.

\subsection{Polynomial approximations to the Kraus operators}
\label{subsubsec:queries-are-polynomial}

We will introduce two polynomial approximations in this section, the first of which is the truncated Taylor expansion of the exponential, given below.
 \begin{equation}\mathrm{Taylor}_d(z)\defeq \sum_{j=0}^{d} \frac{z^j}{j!}\,.
 \end{equation}
Then, we have the following lemma. 
\begin{mathinlay}
   \begin{lemma}[Truncated Taylor condensate approximation] \label[lemma]{lemma:Taylor-bound-exponential}
   Let $W$ be an operator such that there exists a constant $M > 0$ such that for all integers $m > 0$, the action of $W$ on the subspace of $m$-condensates is bounded in norm by $Mm$. Additionally, assume that $W$ maps the subspace of $m$-condensates into the subspace of $m+2$-condensates for all $m$. I.e., $\forall m > 0$,~$W$ satisfies both of the following:
   \begin{xalign}\label{eq:Anormbound}
       &\norm{\mathsf{Con}_m \cdot W \cdot \mathsf{Con}_m}_{\mathrm{op}} \leq Mm \quad\text{and,}\\
       \label{eq:A-kicks-few-bosons}
        &W\cdot \mathsf{Con}_m =\mathsf{Con}_{m+2}\cdot W\cdot \mathsf{Con}_m =\mathsf{Con}_{m+2}\cdot W\cdot \mathsf{Con}_{m+2}\cdot \mathsf{Con}_m\,.
    \end{xalign}
Then, for all $r \geq 0$, $s\geq 3 M$, $1/\e \ge \eps > 0$, and $d\geq 4\ln(1/\varepsilon)+r$, we have
 \begin{equation}\label{eq:Taylor-lemma-bound}
   \norm{\left(\mathrm{Taylor}_d\left(-W/s \right)- \exp(-W/s)\right)\cdot\mathsf{Con}_r}_{\mathrm{op}}\leq \varepsilon\,.
 \end{equation}
 \end{lemma}
 \begin{proof}
 Expanding out the Taylor series for $\e^{-W/s}$, we have
    \begin{xalign}[eq:error-term-taylor]
        \norm{\left(\mathrm{Taylor}_d\left(-W/s \right)- \exp(-W/s)\right)\mathsf{Con}_r}_{\mathrm{op}}
        & = \norm{\sum_{j={d}+1}^{\infty} \frac{1}{j!}s^{-j} (-W)^j \cdot \mathsf{Con}_r}_{\mathrm{op}} \\
        & \leq \sum_{j={d}+1}^{\infty} \frac{1}{j!}s^{-j} \norm{W^j \cdot \mathsf{Con}_r}_{\mathrm{op}}\\
        &= \sum_{j={d}+1}^{\infty} \frac{1}{j!}s^{-j}\norm{\left(\prod_{k=j}^1\mathsf{Con}_{2k+r} W\mathsf{Con}_{2k+r}\right)\mathsf{Con}_r}_{\mathrm{op}} \label{eq:Taylor-approx-b}\\
        &\leq \sum_{j={d}+1}^{\infty} \frac{1}{j!}s^{-j}\prod_{k=j}^1\norm{\mathsf{Con}_{2k+r}W\mathsf{Con}_{2k+r}}_{\mathrm{op}} \label{eq:Taylor-approx-c}\\
        &\leq \sum_{j={d}+1}^{\infty} \frac{1}{j!} \qty(\prod_{k=j}^1 (k+r/2))\left(\frac{M}{s}\right)^{j} \label{eq:Taylor-approx-d}\\
        &= \sum_{j={d}+1}^{\infty}{r/2 + j\choose j}\left(\frac{M}{s}\right)^{j}\,, \label{eq:Taylor-approx-e}
    \end{xalign}
    where we first used the triangle inequality, then used~\eqref{eq:A-kicks-few-bosons} to resolve~\eqref{eq:Taylor-approx-b}, sub-multiplicativity of the operator norm in~\eqref{eq:Taylor-approx-c}, and~\eqref{eq:Anormbound} for proving~\eqref{eq:Taylor-approx-d}.
     Now, we bound \eqref{eq:Taylor-approx-e}. Consider the following
     \begin{equation}
         S(d)\defeq \sum_{j={d+1}}^{\infty} \underbrace{{r/2+j\choose j}\qty(\frac{M}{s})^{j}}_{\defeq w_j}\,. %
     \end{equation}
     We will bound this from above by a geometric series, which requires bounding the ratio of successive terms.  We find
     \begin{equation}
         \frac{w_{j+1}}{w_j}=\qty(\frac{M}{s})\frac{{r/2+j+1\choose j+1}}{{r/2+j\choose j}}= \qty(\frac{M}{s})\frac{r/2+j+1}{j+1}=\qty(\frac{M}{s})\left(1+\frac{r}{2(j+1)}\right)\,.
     \end{equation}
     For $j\geq d$, this is at most 
     \begin{equation}
     \label{eq:Q_definition}
         Q\defeq \qty(\frac{M}{s})\left(1+\frac{r}{2(d+1)}\right)\,.
     \end{equation}
    Since $M/s \leq 1/3$ and $r/2(d+1) \leq 1/2$, $Q \leq 1/2$. Therefore, the series $S(d)$ is bounded by
    \begin{equation}
        S(d)=\sum_{j=d+1}^{\infty} a_j\leq w_{d+1} \sum_{k=0}^{\infty} Q^k=2w_{d+1} = 2\cdot {r/2+{d}+1\choose d+1}\left(\frac{M}{s}\right)^{d+1} \leq 2^{r/2+d + 2}\cdot 3^{-(d+1)}\,.
    \end{equation}
    For $d \geq 4 \ln(1/\epsilon) + r$, this can be upper bounded by $\epsilon$. Since $S(d)$ equals the error in~\eqref{eq:error-term-taylor}, this completes the proof.
\end{proof}
\end{mathinlay}

The second family of polynomials we consider is the family of Chebyshev polynomials of the first kind.  These are typically denoted using the symbol $T_k$, but since this choice would overload the character $T$, we use the following notation,
\begin{equation}
	\mathrm{Cheby}_k(z)= \cos(k\arccos(z)),\;\; z\in[-1,1]\,.
\end{equation}
The Chebyshev polynomials form a basis for the space of polynomials, and, therefore, there exist coefficients $a_k^{(s)}$ such that the function $z \mapsto z^s$ in the basis of Chebyshev polynomials is given by
	\begin{equation}
		z^s=\sum_{k=0}^{\infty} a_k^{(s)} \mathrm{Cheby}_k(z)\,.
	\end{equation}
	Denote the truncated Chebyshev expansion of $x^{s}$ by
	\begin{equation}
    \label{eq:truncated_chebyshev}
		\mathrm{TCheby}_{s, d}(z)\defeq \sum_{k=0}^{d} a_k^{(s)} \mathrm{Cheby}_k(z)\,.
	\end{equation}
    We will make use of the following facts about Chebyshev polynomials.
    \begin{mathinlay}
	\begin{fact}[Theorem 3.3 in~\cite{sachdeva2014faster}]\label{fact:monomial-approximation}
		Then, $|a_k^{(s)}|\leq 1$ for all $k\geq 0$ and 
		\begin{equation}
			\abs{z^s-\mathrm{TCheby}_{s, d}(z)}\leq 2\cdot \exp\left(-\frac{d^2}{2s}\right)\;\text{for all}\; x\in[-1,1]\,.
		\end{equation}
	\end{fact}
	\begin{fact}[Bounds on the coefficients of the Chebyshev polynomials~\cite{narayanan2024improved})]
    \label{fact:cheby-coefficients-bound}
		For all $k\geq 0$, the coefficients of $\mathrm{Cheby}_k$ in the monomial basis are bounded by $(1+\sqrt{2})^k$.
	\end{fact}
    \end{mathinlay}
Then, we have the following polynomial approximation:
\begin{lemma}[Flat approximations of exponential functions] \label{lemma:approxetoA}
Let $W$ be a PSD operator $W \succeq 0$, satisfying the predicates of~\Cref{lemma:Taylor-bound-exponential}. %
	Then, for any $\varepsilon>0$, $r \ge 0$, there is a polynomial, which we call $\mathrm{Flat}_\eps = \mathrm{Flat}_{ \varepsilon, M, r}$ (but we drop the subscripts $M,r$), of degree at most $100 M\ln(1/\varepsilon)(r+1)$ such that 
	\begin{equation}
		\norm{\left(\exp(-W)-\mathrm{Flat}_{\eps}(W)\right)\cdot \mathsf{Con}_r}_{\mathrm{op}}\leq \varepsilon\,.
	\end{equation}
\end{lemma}
\begin{proof}
    We employ the fact that $W \succeq 0$ to ensure that the spectrum of $\exp(-W/w)$ is contained in $[0,1]$. To apply Lemma~\ref{lemma:Taylor-bound-exponential}, fix $w \defeq 36M^2 \ln(2/\epsilon)$, and $d' \defeq \left\lceil \sqrt{72 M^2 \ln^2 (2/\epsilon)}\right\rceil$. Then,
	\begin{equation}
		\exp(-W)\cdot \mathsf{Con}_r=\left(\exp(-W/w)\right)^w \cdot \mathsf{Con}_r\approx_{\varepsilon/2} \mathrm{TCheby}_{w, d'}\left(\exp(-W/w)\right)\cdot \mathsf{Con}_r \label{eq:exp-is-close-to-tcheby}
	\end{equation}
	in operator norm for all $d'\geq \sqrt{2w\ln(2/\varepsilon)}$ by Fact~\ref{fact:monomial-approximation}. Here, we use the following notation: $X \approx_\eps X'$ is used to denote that $\norm{X - X'}_{\mathrm{op}} \leq \eps$. 
	Now let $b_{k}$ be the coefficient of the monomial $z^{k}$ in the decomposition of the polynomial $\mathrm{TCheby}_{w, d'}(z)$---i.e.,
    \begin{equation}
        \mathrm{TCheby}_{w, d'}(z) = \sum_k b_k z^k\,. \label{eq:tcheby-in-terms-of-coefficients}
    \end{equation}
    For our choice of $w$ and $d'$, we have
    \begin{equation}
    \label{eq:d_prime_equation}
    d'=\left\lceil \sqrt{2w\ln(2/\varepsilon)}\right\rceil,
    \end{equation}
    and thus we can bound for all $k \leq d'$,
    \begin{equation} 
    \label{eq:bound_on_s_in_truncated_exponential}
    \frac{k}{w} \leq \frac{d'}{w} \leq \frac{\sqrt{4 w\ln(2/\varepsilon)}}{w}\,.
    \end{equation}
	Now we fix $d \defeq \left(4\left((9M+1)\ln\left(1/{\varepsilon}\right)+\ln(2)\right)+r\right)$. By \Cref{lemma:Taylor-bound-exponential} that whenever 
    \begin{equation}
        \label{eq:choice-of-d-given-d'}
        d \geq 4\ln\left(2 \cdot \left(\max_{k}\abs{b_k}\right) \cdot d' / \epsilon\right) + r\text{~and~}s\defeq \sqrt{\frac{w}{4\ln(2/\varepsilon)}} \geq 3M\,,
    \end{equation}
    we have that
	\begin{xalign}
		\mathrm{TCheby}_{w, d'}\left(\exp(-W/w)\right)\cdot \mathsf{Con}_r&= \sum_{k=1}^{d'}b_k \cdot \exp(-W/w)^k \cdot \mathsf{Con}_r\\
		&=\sum_{k=1}^{d'} b_k \cdot \exp\left(-\frac{k}{w}W\right)\cdot \mathsf{Con}_r\\
		&\approx_{\varepsilon/2} \sum_{k=1}^{d'} b_k\cdot \mathrm{Taylor}_{d}\left(-\frac{k}{w}W\right) \cdot \mathsf{Con}_r\,, \label{eq:tcheby-approximation-origin} 
	\end{xalign}

    Here, we recall that $\mathrm{Taylor}_{d}$ is the truncated Taylor series for the exponential.
    In the final step, we use the fact that $W$ satisfies the conditions required of the operator, and recall that we chose $w=36M^2\ln(2/\varepsilon)$, combined with \eqref{eq:bound_on_s_in_truncated_exponential}, to satisfy the constraint on $s$ required for \Cref{lemma:Taylor-bound-exponential}. The next part of the proof is to show that $d$ satisfies the required bound.  We will take the following upper bound on $d'$ for sufficiently small $\eps$,
    \begin{equation}
     d' = \left\lceil \sqrt{72M^2 \ln^2(2/\epsilon)} \right\rceil \leq 9M\ln(1/\varepsilon)\,. \label{eq:120}
     \end{equation}
Observe that, by definition,
\begin{equation}
    b_k = \sum_{k' = k}^{d'} a_{k'}^{(s)} \cdot [\text{Coefficient of}~z^{k'}~\text{in}~\mathrm{Cheby}_{k'}(z)]\,,
\end{equation}
and by Fact~\ref{fact:cheby-coefficients-bound}, for all $k$, the coefficients $b_k$ can be bounded by
	\begin{equation}
		|b_k|\leq \qty(\abs*{a^{(s)}_1}+\cdots+\abs*{a^{(s)}_{d'}})\qty(1+\sqrt{2})^{d'}\leq  d' (2.5)^{d'}\,.
	\end{equation}
	Taking the logarithm and applying the choice from~\eqref{eq:120} for any $d' \geq 10$, we have that $\ln(|b_k|) \leq \ln(d') + d' \ln(2.5)$, this is equivalent to
	\begin{xalign}
		\ln(|b_k|) + \ln(d') &\leq 2\ln(d') + d'\ln(2.5)\,.
        \\
        & \leq 2\ln(d') + 9M\ln(1/\varepsilon) \cdot 0.89 \\
        &\leq 9M\ln(1/\varepsilon)\,.
	\end{xalign}
	Plugging into \eqref{eq:choice-of-d-given-d'}, we find that $d$ satisfies the necessary lower bound. Now, take $\mathrm{Flat}_\eps(z)$ to be the polynomial defined to be
	\begin{equation}
		\mathrm{Flat}_\eps(z)\defeq \sum_{k=1}^{d'} b_k \cdot \mathrm{Taylor}_{d}\left(-\frac{k}{w}\cdot z\right)\,.
	\end{equation}
    So $\mathrm{Flat}_\eps$ is a degree $\le d$ polynomial.
    Where again, $d = \left(4\left((9M+1)\ln\left(1/{\varepsilon}\right)+\ln(2)\right)+r\right)$, and the coefficients $b_k$ come from $\mathrm{TCheby}_{w, d'}$ for our choice of $w$ and $d'$ from the beginning of the proof.  It has degree at most 
    \begin{align}
        \deg(\mathrm{Flat}_\eps) \leq d \le 
\left(4\left((9M+1)\ln\left(1/{\varepsilon}\right)+\ln(2)\right)+r\right)
        {\leq 100M\ln(1/\varepsilon)(r+1)\,.}
    \end{align} 
    Here, we use loose upper bounds.  This polynomial appears on the right-hand side of~\eqref{eq:tcheby-approximation-origin},
    then
    \begin{equation}
        \mathrm{TCheby}_{w, d'}(\exp(-W/w)) \cdot \mathsf{Con}_r \approx_{\eps/2} \mathrm{Flat}(W) \cdot \mathsf{Con}_r\,.
    \end{equation}
    The previous equation and~\eqref{eq:exp-is-close-to-tcheby} combine via triangle inequality for the lemma statement.
\end{proof}
We can now finally consider the operators $\displaystyle
	\prod_{i=1}^T e_{x_i}(\tilde{\G}_{y_i}^2)$,
where, recall, $e_{x_i}(z)$ are terms
\begin{xalign}
    e_0(z) &\defeq 1-\e^{-\kappa z}, \quad\text{and} \\
    e_1(z) &\defeq \sqrt{\e^{-\kappa z}(2-\e^{-\kappa z})} = \sqrt{2} \e^{-\kappa z/2} \sqrt{1 - \e^{-\kappa z}/2}\,,
\end{xalign}
to match the definitions of the Kraus operators $E_0^{(y_i)}$ and $E_1^{(y_i)}$ (\Cref{cor:krausoperators}). 
Towards the proof of~\Cref{cor:sandwiching}, we prove the following more general lemma: 
\begin{lemma}\label[lemma]{lemma:Krausapproximation}
  Let $W_1,\ldots, W_T \succeq 0$ be PSD operators that each satisfy the predicates of~\Cref{lemma:Taylor-bound-exponential} and pairwise commute.
  For any $r\geq 0$ and every $\x \in \bits^T$, there exists a multivariate polynomial $\mathrm{AKraus}_{\eps,\x}$ such that  
  \begin{equation}
      \norm{\left(\prod_{i=0}^T e_{x_i}(W_i) - \mathrm{AKraus}_{\eps,\x}(W_1,\ldots,W_T)\right)\cdot \mathsf{Con}_r}_{\mathrm{op}}\leq \varepsilon\,.
  \end{equation}
  Moreover, $\mathrm{AKraus}_{\eps,\x}$ is of degree $O(M T^3 \ln(T) \ln^2(1/\eps) (r+1))$.
\end{lemma}
\begin{proof}
We begin with a single PSD operator $W$ that satisfies the predicates of~\Cref{lemma:Taylor-bound-exponential}.
We replace the function $z \mapsto \sqrt{1-z/2}$ in the definition of $e_1$ with the truncated binomial expansion\footnote{The usual binomial expansion $(1+a)^r=\sum_k \binom{r}{k}a^k$ also extends to real-valued $r$. Here, we are using $r=1/2$ and $a=-z/2$, and truncating to $k\leq d''$.}:
\begin{equation}
	\mathrm{TSqrt}\left(z\right)\defeq \sum_{k = 0}^{d''} \binom{1/2}{k} \left(\frac{-1}{2}\right)^{k} z^k\,,
\end{equation}
with the choice of $d'' = 4 + \frac{3}{2} (\ln(T) + \ln(1/\varepsilon))$. The nomenclature $\mathrm{TSqrt}(\cdot)$ comes from it being a truncated Taylor expansion. Recall that the square-root function does not have a good Taylor series expansion about $0$, which is why we build it from $z \mapsto \sqrt{1 -z/2}$. 
We define the following polynomial approximations to make progress towards approximation $e_0(z)$ and $e_1(z)$.
\begin{xalign}[eq:definition-of-outside-p]
    p_0(z) &\defeq 1 - z^2\,, \\
    p_1(z) &\defeq \sqrt{2} z \cdot \mathrm{TSqrt}(z^2)\,.
\end{xalign}
Then $e_x(z)$ should be approximated by $p_x(e^{-\kappa z/2})$ whenever $\mathrm{TSqrt}$ is a good approximation of $\sqrt{1 - z/2}$. 

\begin{mathinlay}
    \begin{claim} \label[claim]{claim:e-to-p-bound}
        It follows that
        \begin{equation}
            \norm{\prod_{i=1}^T e_{x_i}(W_i)-\prod_{i=1}^T p_{x_i}\left(\exp(-\kappa W_i/2)\right)}_{\mathrm{op}} \leq \eps/2\,.
        \end{equation}
    \end{claim}
    \begin{proof}

Next, observe that by construction for PSD $W$,
\begin{equation}
    p_0(\exp(-\kappa W/2)) = e_0(W).
\end{equation}
The challenge is to prove a similar statement for $p_1$ and $e_1$. Here, we will achieve an approximation.
\begin{xalign}[eq:square_root_approximation]
	&\norm{p_1(\exp(-\kappa W/2))-e_1(W)}_{\mathrm{op}}= \\
    &\hspace{2em}\norm{\sqrt{2}\exp(-\kappa W/2) \qty(\sqrt{1 - \frac{1}{2} \exp(-\kappa W)} -\mathrm{TSqrt}\qty(\exp(-\kappa W)))}_{\mathrm{op}}\\
    &\hspace{2em}\leq \sqrt{2}\norm{\exp(-\kappa W/2)}_{\mathrm{op}}\sum_{k=d''}^{\infty} \abs{{1/2\choose k}}\cdot \frac{1}{2^k}\cdot \norm{\exp(-\kappa W)^k}_{\mathrm{op}}\\
	&\hspace{2em}\leq \frac{\sqrt{2}}{2^{d''}}\sum_{k=0}^{\infty} \frac{1}{2^k} \\
    &\hspace{2em}\leq \frac{\epsilon}{4T}\,,
 \end{xalign}
where we only use the facts that $\abs{{1/2 \choose k}}\leq 1$ and $\norm{\left(\e^{-\kappa W}\right)^k}_{\mathrm{op}}\leq 1$ as $W$ is PSD.
So far we have produced approximations of single terms. We now lift these approximations to products of terms using a hybrid argument. We define hybrids $B_j$:
\begin{equation}
    B_j \defeq \prod_{k=T-j}^T e_{x_k}(W_k) \prod_{k=1}^{T-j-1} p_{x_i}(\exp(-\kappa W_k/2))\,.
\end{equation}
The hybrids combine to yield a total bound of 
\begin{xalign}[eq:e-to-p-hybrid]
    &\norm{\prod_{i=1}^T e_{x_i}(W_i)-\prod_{i=1}^T p_{x_i}\left(\exp(-\kappa W_i/2)\right)}_{\mathrm{op}}\\
    &\hspace{5mm}\leq \sum_{j=1}^T \norm{B_j-B_{j-1}}_{\mathrm{op}}\\
    &\hspace{5mm}\leq\sum_{j=1}^T \norm{e_{x_{T-j}}(W_{T-j})-p_{x_{T-j}}(\exp(-\kappa W_{T-j}/2))}_{\mathrm{op}}\prod_{k=1}^{T-j}\norm{p_{x_{k}}(\exp(-\kappa W_{k}/2))}_{\mathrm{op}}\\
    &\hspace{5mm}\leq T \cdot \frac{\eps}{4T} \cdot \left(1+\frac{\epsilon}{4T}\right)^{T-j}\\
    &\hspace{5mm}\leq \frac{\varepsilon}{2}, \label{eq:easy-algebraic-manipulation}
\end{xalign}
Here, we use the facts that $\norm{e_{x_k}(W_{k})}_{\mathrm{op}} \leq 1$, and that $e_{0}(W_i) = p_{0}(W_i)$, and from \eqref{eq:square_root_approximation} we know that $\norm{e_1(W_i) - p_1(W_i)}_{\mathrm{op}} \leq \epsilon / 4T$.  We also used that via the triangle inequality,
\begin{equation}
\norm{p_{x_k}(\exp(-\kappa W_{k}/2))}_{\mathrm{op}}\leq \frac{\epsilon}{4T}+\norm{e_{x_k}(W_{k})}_{\mathrm{op}}\leq 1+\frac{\epsilon}{4T}\,.
\end{equation}
Lastly, we resolved~\eqref{eq:easy-algebraic-manipulation} using
\begin{equation}
\left(1+\frac{\epsilon}{4T}\right)^{T-j}\leq \left(1+\frac{\epsilon}{4T}\right)^{T}\leq \e^{\frac{\epsilon}{4T} T}\leq \e^{1/4}\leq 2\,.
\end{equation}
    \end{proof}
\end{mathinlay}

We then observe that  $\prod_{i=1}^Tp_{x_i}(\exp(-\kappa W_{i} / 2))$ is a polynomial in the variables $\exp(-\kappa W_{k}/2)$ of degree at most $3d''T$ due to $T$ terms of the form $z \mathrm{TSqrt}(z^2)$ which will yield a degree of $(1 + 2 d'') \le 3d''$. 
For any collection of $J$ many operators $W_{i_j}$, we can write 
\begin{equation}\prod_{j=1}^J \exp\qty({-\kappa W_{i_j}/2})=\exp\qty({-\kappa \sum_{j=1}^J W_{i_j}/2})
\end{equation}
as by assumption, the $W_i$ pairwise commute.
Notice also that $\sum_{j=1}^J W_{i_j}$ satisfies~\eqref{eq:A-kicks-few-bosons} (because each term maps the subspace of $m$-condensates to the subspace of $m+2$-condensates) and by the triangle inequality, we have that
\begin{equation}\label{eq:normboundkrauslemma}
    \norm{\mathsf{Con}_m \cdot \left(\frac{\kappa}{2}\sum_{j=1}^J W_{i_j}\right) \cdot \mathsf{Con}_m}_{\mathrm{op}}\leq \kappa JMm \leq JMm\,.
\end{equation}
In the last line, we dropped the $\kappa$ as it is a constant $< 1$ (think $1/10$).
Next, we take the univariate polynomials $p_x(z)$ defined in~\eqref{eq:definition-of-outside-p} and write them in the monomial basis in terms of families of coefficients $\{c_k^{(0)}\}$ and $\{c_k^{(1)}\}$:
\begin{equation}
    p_x(z) = \sum_{k = 0}^{3d''} c_k^{(x)} z^k.
\end{equation}
This lets us construct multivariate polynomials $\mathrm{AKraus}_{\eps,\x}$ in terms of the $\mathrm{Flat}_\eps$ polynomials built in~\Cref{lemma:approxetoA}. We define $\mathrm{AKraus}_{\eps,\x}$ for $\x = (x_1, \ldots, x_T)$ as
\begin{equation}
    \mathrm{AKraus}_{\eps,\x}(z_1,\ldots,z_T) \defeq \sum_{k_1,\ldots, k_T = 0}^{3d''} \left(\prod_{i=1}^T c_{k_i}^{(x_i)} \cdot \mathrm{Flat}_{\frac{\eps}{2\cdot2^T}}\qty(-\kappa \sum_{i=1}^T k_i z_i/2)\right)\,,
\end{equation}
Using the fact that in the sum in the argument of $\mathrm{Flat}_{\frac{\epsilon}{2 \cdot 2^{T}}}$, we are summing $T$ many operators with coefficients that are integers and range from $0$ to $3d''$, the largest value of $J$ which we need to consider in \eqref{eq:normboundkrauslemma} is upper bounded by $J \leq 3d'' T$.  

The polynomials $\mathrm{AKraus}_{\eps,\x}$ are designed to approximate the polynomials $p_{0}$ and $p_{1}$ in the exponential function when evaluated on operators $(W_1, \ldots, W_T)$, which we prove in the following claim.

\begin{mathinlay}
    \begin{claim} \label[claim]{claim:p-to-akraus} It follows that
        \begin{equation}
            \norm{\left(\prod_{i=1}^T p_{x_i}(\exp(-\kappa W_i/2)) - \mathrm{AKraus}_{\varepsilon,\x}(W_1,\ldots,W_T)\right)\cdot \mathsf{Con}_r}_{\mathrm{op}} \leq \eps/2\,.
        \end{equation}
    \end{claim}

    \begin{proof}
    We use that for any monomial $\prod_{j=1}^J x_{i_j}$, we can apply~\Cref{lemma:approxetoA} to achieve
\begin{equation}
\prod_{j=1}^J \exp(-\kappa W_{i_j}/2)\cdot\mathsf{Con}_r\approx_{\frac{\varepsilon}{2\cdot 2^T}} \mathrm{Flat}_{\frac{\eps}{2\cdot2^T}}\qty(\sum_{j=1}^J \kappa W_{i_j}/2)\cdot \mathsf{Con}_r\,.
\end{equation}
Moreover, we used that the absolute coefficients of $\prod_{i=1}^T p_{x_i}$ add up to at most $2^T$.
To see this, observe that the absolute coefficients of $p_0$ add up to $2$
Moreover, all coefficients in $\mathrm{TSqrt}(z^2)$ are negative except the first, which equals 1. Therefore, the absolute coefficients of $p_1$ add up to 
\begin{equation}
\sum_{k_i=0}^{3d''} \abs{c_{k_i}^{(x_i)}}=\sqrt{2}(2-\mathrm{TSqrt}(1))\leq \sqrt{2}(2-\sqrt{1-1/2}))\leq 2.
\end{equation}
Finally, we can apply a triangle inequality to find
 \begin{xalign}
      &\norm{\left(\prod_{i=1}^T p_{x_i}(\exp(-\kappa W_i/2)) - \mathrm{AKraus}_{\varepsilon,\x}(W_1,\ldots,W_T)\right)\cdot \mathsf{Con}_r}_{\mathrm{op}}\\
      &\hspace{5mm}\leq \sum_{k_1,\ldots,k_T=0}^{3d''}\abs{\prod_{i=1}^T c_{k_i}^{(x_i)} }\cdot \norm{\left(\mathrm{Flat}_{\frac{\eps}{2\cdot2^T}}\qty(-\kappa \sum_{i=1}^T k_i W_i/2)-\exp(-\kappa\sum_{i=1}^T k_i W_i/2)\right)\mathsf{Con}_r}\\
      &\hspace{5mm}\leq \sum_{k_1,\ldots,k_T=0}^{3d''}\abs{\prod_{i=1}^T c_{k_i}^{(x_i)} }\cdot  \frac{\varepsilon}{2\cdot 2^T}\\
         &\hspace{5mm}\leq \sum_{k_1,\ldots,k_T=1}^{3d''}\prod_{i=1}^T \abs{c_{k_i}^{(x_i)} }\cdot  \frac{\varepsilon}{2\cdot 2^T}\\
         &\hspace{5mm}= \prod_{i=1}^T \sum_{k=0}^{3d''}\abs{c_{k_i}^{(x_i)} }\cdot  \frac{\varepsilon}{2\cdot 2^T}\\
      &\hspace{5mm}\leq 2^T \frac{\varepsilon}{2\cdot 2^T}\\
      &\hspace{5mm}\leq \frac{\varepsilon}{2}\,.
 \end{xalign}
    \end{proof}
\end{mathinlay}

Now we bound the distance of $\mathrm{AKraus}_{\eps,\x}$ from the product of exponentials via a middleman:
\begin{xalign}
      &\norm{\left(\prod_{i=0}^T e_{x_i}(W_i) - \mathrm{AKraus}_{\varepsilon,\x}(W_1,\ldots,W_T)\right)\cdot \mathsf{Con}_r}_{\mathrm{op}} \\
      &\hspace{5mm}\begin{matrix*}[l]\displaystyle\leq \norm{\left(\prod_{i=0}^T e_{x_i}(W_i) - \prod_{i=0}^T p_{x_i}(\exp(-\kappa W_i/2)) \right)\cdot \mathsf{Con}_r}_{\mathrm{op}}  \\
      \hspace{5mm} \displaystyle+\norm{\left(\prod_{i=0}^T p_{x_i}(\exp(-\kappa W_i/2)) - \mathrm{AKraus}_{\varepsilon,\x}(W_1,\ldots,W_T)\right)\cdot \mathsf{Con}_r}_{\mathrm{op}} \end{matrix*}\\
      &\hspace{5mm}\leq \frac{\varepsilon}{2}+\frac{\varepsilon}{2}=\varepsilon\,,
\end{xalign}
from a combination of~\Cref{claim:e-to-p-bound} and~\Cref{claim:p-to-akraus}.
The overall degree of $\mathrm{AKraus}_{\eps,\x}$ is then given by  $J \leq 3d'' T$ and $d'' = 4 + \frac{3}{2} (\ln(T) + \ln(1/\varepsilon))$ to yield 
\begin{xalign}
\deg(\mathrm{AKraus}_{\eps,\x}) &\le T \deg \mathrm{Flat}_{\frac{\eps}{2\cdot 2^T},JM, r} \\
&\le 100 JMT (\ln(1/\eps) + T + 1) (r+1) \\
&\le 300 d''MT^3 \ln(1/\eps) (r+1) \\
&\le O(M T^3 \ln(T) \ln^2(1/\eps) (r+1))\,,
\end{xalign}
completing the proof.
\end{proof}
\Cref{theorem:poly_approximation_v2} will follow from the following lemma (\Cref{lemma:otherregisters}) combined with~\Cref{lemma:Krausapproximation} for the choice of $W_i=\tilde{\G}^2_{y_i}$, where $y_i$ corresponds to the query of the verification algorithm. The following lemma applies triangle inequalities to bound the difference in the overall states when the oracle action is lightly replaced.
\begin{lemma}\label[lemma]{lemma:otherregisters}
    Let $B_{\x,\y, \bb}$ and $B'_{\x,\y,\bb}$ be two families of (possibly unnormalized) operators acting on the $\reg{S}$ register. 
    Define
\begin{equation}
   \ket{\psi_B} \defeq \sum_{\y, \bb} \qty[\qty(A_{\y, \bb})_{\reg{A}} \otimes \qty(\sum_{\x} (K^\x)_{\y, \bb} \otimes B_{\x,\y, \bb})] A \ket{0}_{\reg{A}} \ket{\mathrm{init}}_{\reg{US}}
\end{equation}
and $\ket{\psi_{B'}}$ similarly.
    Then for all suitably large $n$, 
    \begin{equation}
    \norm{\ket{\psi_B}-\ket{\psi_{B'}}}\leq 2^{2nT} \max_{\x,\y, \bb}\norm{(B_{\x,\y, \bb}-B'_{\x,\y, \bb})\ket{\mathrm{init}_{S}}}.
    \end{equation}
\end{lemma}

\begin{proof}
We apply the triangle inequality twice (once for the sum over $\y$ and then for the sum over $\x$).
We use that $\norm{K^{x_i}}_{\mathrm{op}}\leq 2$ and that $\norm{\ketbra{y_i} \cdot A}_{\mathrm{op}}\leq 1$ as well as the submultiplicativity of the Schatten $\infty$-norm. 
\begin{xalign}[eq:triangle_psi_psi'_v2]
	&\norm{\ket*{\psi_B} - \ket*{\psi_{B'}}} \nonumber \\
	&\hspace{0.2em} = \norm{\sum_{\substack{\y \in \left(\bits^{n}\right)^{T}\\\bb \in \bits^{T}}} \qty[\qty(A_{\y, \bb})_{\reg{A}} \otimes \qty(\sum_{\x\in \bits^{\abs{\bb}}} (K^\x)_{\y, \bb} \otimes \qty(B_{\x,\y, \bb}-B'_{\x,\y, \bb}) )] A \ket{0}_{\reg{A}} \ket{\mathrm{init}}_{\reg{US}}}\label{eq:tri-psi-psi'-im1_1} \\
    &\hspace{0.2em} \leq \sum_{\substack{\y \in \left(\bits^{n}\right)^{T}\\\bb \in \bits^{T}}} \norm{\qty[A_{\y,\bb} \cdot A \ket{0} \otimes \qty(\sum_{\x\in \bits^{\abs{\bb}}} (K^\x)_{\y, \bb} \ket{\bot}^{\otimes 2^n} \otimes \qty(B_{\x,\y, \bb}-B'_{\x,\y, \bb})\ket{\mathrm{init}_S} )] } \label{eq:tri-psi-psi'-im1_2}\\
    &\hspace{0.2em} = \sum_{\substack{\y \in \left(\bits^{n}\right)^{T}\\\bb \in \bits^{T}}} \qty[\norm{A_{\y, \bb} \cdot A \ket{0}} \cdot \norm{\qty(\sum_{\x\in \bits^{\abs{\bb}}} (K^\x)_{\y, \bb} \ket{\bot}^{\otimes 2^n} \otimes \qty(B_{\x,\y, \bb}-B'_{\x,\y, \bb})\ket{\mathrm{init}_S} )}] \label{eq:tri-psi-psi'-im1_3}\\
    &\hspace{0.2em} \leq \sum_{\substack{\y \in \left(\bits^{n}\right)^{T}\\\bb \in \bits^{T}}} \qty[\norm{A_{\y, \bb} \cdot A \ket{0}} \cdot \sum_{\x\in \bits^{\abs{\bb}}} \norm{\qty((K^\x)_{\y, \bb} \ket{\bot}^{\otimes 2^n} \otimes \qty(B_{\x,\y, \bb}-B'_{\x,\y, \bb})\ket{\mathrm{init}_S} )}] \label{eq:tri-psi-psi'-im1_4}\\
	&\hspace{0.2em} = \sum_{\substack{\y \in \left(\bits^{n}\right)^{T}\\\bb \in \bits^{T}}} \qty[\norm{A_{\y, \bb} \cdot A \ket{0}}\cdot \qty(\sum_{\x\in \bits^{\abs{\bb}}} \norm{(K^\x)_{\y, \bb} \ket{\bot}^{\otimes 2^n}} \cdot \norm{\qty(\left(B_{\x,\y, \bb}-B'_{\x,\y, \bb}\right)\ket{\mathrm{init}_S})} )] \label{eq:tri-psi-psi'-im3_5} \\
	&\hspace{0.2em} \leq 2^T \sum_{\substack{\y\in \left(\bits^{n}\right)^{T}\\\bb \in \bits^{T}\\\x\in \bits^{\abs{\bb}}}} \norm{\qty(B_{\x,\y, \bb}-B'_{\x,\y, \bb})\ket{\mathrm{init}_S}}\label{eq:tri-psi-psi'-im4_6} \\ 
	&\hspace{0.2em} \leq 2^{3T + nT} \max_{\y, \bb, \x}\norm{\qty(B_{\x,\y, \bb}-B'_{\x,\y,\bb})\ket{\mathrm{init}_S}} \label{eq:tri-psi-psi'-im1_7} \\
    &\hspace{0.2em} \leq 2^{2nT} \max_{\y, \x, \bb}\norm{\qty(B_{\x,\y, \bb}-B'_{\x,\y, \bb})\ket{\mathrm{init}_S}} \label{eq:tri-psi-psi'-im1_7}\,.
\end{xalign}
Here,~\eqref{eq:tri-psi-psi'-im1_2} and~\eqref{eq:tri-psi-psi'-im1_4}
are applications of the triangle inequality,~\eqref{eq:tri-psi-psi'-im1_3} and~\eqref{eq:tri-psi-psi'-im3_5} are applying that norm of a tensor product is the tensor product of the norms, and~\eqref{eq:tri-psi-psi'-im4_6} comes from bounding $\norm{A_{\y, \bb} A \ket{0}} \leq 1$, and $\norm*{(K^\x)_{y, b} \ket{\bot}^{\otimes 2^n}} \leq 2^T$ using the bound on $\norm{K_{x_i}}_{\mathrm{op}}$ and the submultiplicativity of the operator norm.
\end{proof}
Finally, we proof~\Cref{theorem:poly_approximation_v2}, which we re-state here.

\begin{theorem*}[Polynomial approximation of post-query state, \Cref{theorem:poly_approximation_v2}, restated]
	For every $\iota > 0$, there exists a family of polynomials $\mathrm{AKraus}_{\x}$ (that implicitly depends on $\iota$) in the operators $\tilde{\G}_{y_i}^2$ with 
        \begin{equation}
    \deg \left(\mathrm{AKraus}_{\x}\right) =O(n^2T^5\ln(T)\ln^2(1/\iota))\,,
    \end{equation}
    such that we can bound the distance between the true post-query state and the polynomial approximation by
	\begin{equation}
		\norm{\ket*{\psi_\mathrm{PQ}} - \ket*{\psi_{\mathrm{AKraus}}}} \leq \iota\,.
	\end{equation}
\end{theorem*}
\begin{proof}
Since $\ket{\mathrm{init}_S}=\mathsf{Con}_0 \ket{\mathrm{init}_S}$, we can directly apply~\Cref{lemma:Krausapproximation} for the operators $W_i=\tilde{\G}_{y_i}^2$ with $r=0$. 
Notice that the $\tilde{\G}_{y_i}^2$ are PSD, pairwise commute as they are all diagonal in the position basis, and satisfy~\eqref{eq:A-kicks-few-bosons} as well as~\eqref{eq:Anormbound} with a constant $M$ by~\Cref{fact:h-tilde-norm-bound}.
The result then follows from~\Cref{lemma:otherregisters} by choosing $\eps$ such that $\eps=\iota/(2^{2nT})$.
\end{proof}

\subsection{Extending the approximation to the sandwiched operator}
In this subsection, we prove~\Cref{cor:sandwiching}, which we first re-state. 
\begin{theorem*}[Sandwiching theorem,~\Cref{cor:sandwiching}, restated]
For every $\iota > 0$, there exist integers \linebreak $r =O(n^2T^5\ln^3(T)\ln^2(1/\iota))$ and $R=O(n^3T^6\ln^4(T)\ln^3(1/\iota))$, such that 
\begin{equation}
\norm{\ket*{\psi_{\mathrm{PQ}}} - \ket*{\tilde{\psi}_{R,r}}} \leq \iota.
\end{equation}
\end{theorem*}

\begin{proof}
	The proof consists of two hybrid arguments on the $\reg{S}$ register. 
    We will use the approximation in~\Cref{lemma:Krausapproximation} in both hybrids. 
    The first hybrid uses~\Cref{lemma:Krausapproximation} starting from a 0-condensate to obtain an approximation $\mathrm{AKraus}_{\varepsilon,\x_{T-j}}(\tilde{\G}^2_{y_1},\ldots \tilde{\G}^2_{y_{T-j}})$ of the operators $\prod_{i=1}^{T-j} e_{x_i}(\tilde{\G}^2_{y_i})$ restricted to $\mathsf{Con}_0$.
    Here, we define the substring $\x_{k}=(x_k,\ldots,x_1)$. Note that these approximation are for products of $e_{x_i}(\tilde{\G}_{y_i}^2)$ of varying lengths.
    By~\Cref{lemma:Krausapproximation} with error $\epsilon / 2T$, we have that the degree of $\mathrm{AKraus}_{\frac{\varepsilon}{2T},\x_{T-j}}$ is bounded by the following for all $j$,
    \begin{equation}
        \deg(\mathrm{AKraus}_{\frac{\varepsilon}{2T},\x_{T-j}})=O(T^3\ln(T)\ln^2(2T/\varepsilon)) =O(T^3\ln^3(T)\ln^2(1/\varepsilon))\,.
    \end{equation}
    Now set $r=2\deg(\mathrm{AKraus}_{\frac{\varepsilon}{2T},\x}) \geq 2 \max_{j} \deg(\mathrm{AKraus}_{\frac{\varepsilon}{2T},\x_{T-j}})$, and define the following hybrids:
	\begin{equation}
		D_j\defeq \prod_{i=T}^{T-j+1} \mathsf{Con}_r \cdot e_{x_i}(\tilde{\G}_{y_i}^2)\prod_{i=1}^{T-j} e_{x_i}(\tilde{\G}_{y_i}^2)\ket{\mathrm{init}_S}\,.
	\end{equation}
Note that the first $T-j$ terms commute so the order of expansion of the product does not matter.  Then because the initial state is a $0$-condensate,
\begin{align}
	\norm{\left(\prod_{i=T}^1 \mathsf{Con}_r \cdot e_{x_i}(\tilde{\G}_{y_i}^2) \cdot\mathsf{Con}_r-\prod_{i=1}^{T} e_{x_i}(\tilde{\G}_{y_i}^2)\right)\ket{\mathrm{init}_S}} = \norm{D_T - D_0} \leq \sum_{j=1}^T \norm{D_{j}-D_{j-1}}\,.
\end{align}
Then, we have
\begin{xalign}
	\norm{D_j-D_{j-1}}&\leq \norm{(\id - \mathsf{Con}_r)\cdot\prod_{i=1}^{T-j} e_{x_i}(\tilde{\G}_{y_i}^2)\ket{\mathrm{init}_S}}\\
	&\leq \frac{\varepsilon}{2T} +\norm{(\id - \mathsf{Con}_r)\cdot \mathrm{AKraus}_{\varepsilon,\x_{T-j}}(\tilde{\G}_{y_1}^2,\ldots,\tilde{\G}_{y_{T-j}}^2)\ket{\mathrm{init}_S}}\\
	&\leq \frac{\varepsilon}{2T} +\norm{(\id - \mathsf{Con}_r)\cdot \mathsf{Con}_r\cdot \mathrm{AKraus}_{\varepsilon\x_{T-j}}(\tilde{\G}_{y_1}^2,\ldots,\tilde{
    \G}_{y_{T-j}}^2)\ket{\mathrm{init}_S}}\\
	&\leq \frac{\varepsilon}{2T}\,.
\end{xalign}
Here we use the fact that $\mathrm{AKraus}_{\varepsilon,\x_{T-j}}$ is a polynomial of degree at most $r/2$ in $\tilde{\G}_y^2$ and therefore moves at most $r$ bosons from the $0$-momentum mode, so the state after applying $\mathrm{AKraus}_{\eps,\x_{T-j}}$ to the initial state is an $r$-condensate.  
Next, since the condensate projectors commute, we have that for all $m\geq 0$ that
\begin{equation} 
\norm{\mathsf{Con}_m\cdot \mathsf{Con}_R\cdot \tilde{\G}_y^2 \cdot \mathsf{Con}_R \cdot \mathsf{Con}_m}_{\mathrm{op}}=\norm{\mathsf{Con}_R \cdot \mathsf{Con}_m \cdot \tilde{\G}_y^2 \cdot \mathsf{Con}_m \cdot \mathsf{Con}_R}_{\mathrm{op}}\leq \norm{\mathsf{Con}_m\cdot \tilde{\G}_y^2\cdot \mathsf{Con}_m}_{\mathrm{op}}\,.
\end{equation}
Therefore, by~\Cref{fact:h-tilde-norm-bound}, the operator $W=\mathsf{Con}_R \cdot \tilde{\G}_y^2 \cdot \mathsf{Con}_R$ satisfies~\eqref{eq:Anormbound} for all $m>0$ with a constant $M = O(1)$.
Moreover, it is easy to check that $W$ is also PSD and satisfies~\eqref{eq:A-kicks-few-bosons} (again, because $\mathsf{Con}_R$ and $\mathsf{Con}_m$ commute).
Hence, from~\Cref{lemma:Krausapproximation} for $T=1$ (here, $T$ is the parameter in \Cref{lemma:Krausapproximation}, not the number of queries as in the rest of the paper) and a constant $M=O(1)$, there exists a polynomial $\mathrm{AKraus}_{\frac{\eps}{4T},x}$ such that:
\begin{xalign}
\label{eq:single_sandwiched_approximation}
    \mathrm{AKraus}_{\frac{\eps}{4T},x}(\mathsf{Con}_R \cdot \tilde{\G}_y^2 \cdot \mathsf{Con}_R) \cdot \mathsf{Con}_r &\approx_{\varepsilon/4T} e_x(\mathsf{Con}_R \cdot \tilde{\G}_y^2 \cdot \mathsf{Con}_R) \cdot \mathsf{Con}_r,~\text{and}\\
     \mathrm{AKraus}_{\frac{\eps}{4T},x}(\tilde{\G}_y^2) \cdot \mathsf{Con}_r &\approx_{\varepsilon/4T} e_x( \tilde{\G}_y^2 )\cdot \mathsf{Con}_r\,,
\end{xalign}
where the degree of $\mathrm{AKraus}_{\frac{\eps}{4T},x}$ will be $O((r+1) \ln(4T/\varepsilon))$.  Since this holds for all $R \geq r$, it will hold for $R=r+2\deg(\mathrm{AKraus}_{\frac{\eps}{4T},x})=r+O((r+1)\ln(4T/\varepsilon))$.
We now proceed with the second hybrid.  Let $C_{j}$ be defined as
\begin{equation}
	C_j \defeq \prod_{i=T}^{T-j+1} \mathsf{Con}_r \cdot e_{x_i}(\mathsf{Con}_R\cdot \tilde{\G}_{y_i}^2 \cdot \mathsf{Con}_R) \cdot \mathsf{Con}_r\cdot \prod_{i=T-j}^{0}\mathsf{Con}_r \cdot e_{x_i}(\tilde{\G}_{y_{i}}^2)\cdot \mathsf{Con}_r\,.
\end{equation}
Then $\norm{C_T - C_0}$ equals
\begin{equation}
	\norm{\prod_{i=T}^1\mathsf{Con}_r\cdot e_{x_i}(\mathsf{Con}_R \cdot \tilde{\G}_{y_i}^2 \cdot \mathsf{Con}_R)\cdot \mathsf{Con}_r-\prod_{i=T}^1 \mathsf{Con}_r \cdot e_{x_i}(\tilde{\G}_{y_i}^2)\cdot \mathsf{Con}_r}_{\mathrm{op}}\leq \sum_{j=1}^T \norm{C_j-C_{j-1}}_{\mathrm{op}}\,.
\end{equation}
Then, we find
\begin{xalign}
	&\norm{C_j-C_{j-1}}_{\mathrm{op}}\\
    &\hspace{5mm}\leq \norm{\mathsf{Con}_r\cdot e_{x_{T-j+1}}(\mathsf{Con}_R\cdot  \tilde{\G}_{y_{T-j+1}}^2 \cdot \mathsf{Con}_R) \cdot \mathsf{Con}_r-\mathsf{Con}_r\cdot e_{x_{T-j+1}}(\tilde{\G}^2_{y_{T-j+1}})\cdot \mathsf{Con}_r}_{\mathrm{op}}\\
	&\hspace{5mm}\begin{matrix*}\leq \norm{\left(e_{x_{T-j+1}}(\mathsf{Con}_R \cdot \tilde{\G}^2_{y_{T-j+1}} \cdot \mathsf{Con}_R)-\mathrm{AKraus}_{\frac{\eps}{4T},x_{T-j+1}}(\mathsf{Con}_R\cdot\tilde{\G}^2_{y_{T-j+1}}\cdot\mathsf{Con}_R)\right) \cdot \mathsf{Con}_r}_{\mathrm{op}}\\
    \hspace{5mm}+\norm{\left(e_{x_{T-j+1}}\qty(\tilde{\G}^2_{y_{T-j+1}})-\mathrm{AKraus}_{\frac{\eps}{4T},x_{T-j+1}}\qty(\tilde{\G}^2_{y_{T-j+1}})\right)\cdot\mathsf{Con}_r}_{\mathrm{op}} \\
	\hspace{5mm}+\norm{\begin{matrix*}[l]~\mathsf{Con}_r \cdot \mathrm{AKraus}_{\frac{\eps}{4T},x_{T-j+1}}\qty(\tilde{\G}^2_{y_{T-j+1}})\cdot\mathsf{Con}_r \\ ~~~~-\mathsf{Con}_r \cdot \mathrm{AKraus}_{\frac{\eps}{4T},x_{T-j+1}}\qty(\mathsf{Con}_R\cdot\tilde{\G}^2_{y_{T-j+1}}\cdot\mathsf{Con}_R)\cdot\mathsf{Con}_r\end{matrix*}}_{\mathrm{op}} \label{eq:c_hybrids_c}\end{matrix*}\\
	&\hspace{5mm}\leq \frac{\varepsilon}{4T}+\frac{\varepsilon}{4T}+0\\
	&\hspace{5mm}\leq \frac{\varepsilon}{2T}\,,
\end{xalign}
    for any $R\geq r+2\deg(\mathrm{AKraus}_{\frac{\eps}{4T},x_{T-j}})$.  Here, in~\eqref{eq:c_hybrids_c}, observe that when the  polynomials $\mathrm{AKraus}_{\frac{\eps}{4T},x}$ are being invoked, the state remains a $r + 2\deg(\mathrm{AKraus}_{\frac{\varepsilon}{4T}x_{T-j}}) \leq R$-condensate so the final term in the 3-fold expression is 0.  
    We are also using the definitions of $C_{j}$ and $C_{j-1}$, together with standard properties of the operator norm, and the polynomial approximations from \eqref{eq:single_sandwiched_approximation} and \Cref{lemma:Krausapproximation}. The result then follows from choosing $\varepsilon=\iota /(2^{2nT})$ and applying~\Cref{lemma:otherregisters}.
    We then upper bound $\mathrm{polylog}(T)\leq O(T)$ for notational convenience.
\end{proof}

\color{black}
\section{Proving the quasi-even property}

In this section, we prove that the purified state of any verification algorithm making $T$-queries to the $U$ oracle is also incredibly close to a $v/4$-quasi-even state, i.e., that the state is almost entirely supported on momentum Fock states that have at most $v/4$ modes occupied by an odd number of bosons.  The proof follows the general framework of \cite{hamoudi2023quantum} for unconstrained search.  At a high level, the authors first define a sequence of projectors $\Pi_0 \preceq \Pi_1 \preceq \ldots $, and then show that querying their oracle maps a state supported on $\Pi_{k}$ to a state that it mostly supported to $\Pi_k$ but with an extremely small support on $\Pi_{k+1}$.  

Our proof will utilize their framework, with the projectors $\mathsf{QE}_o$ playing the role of the $\Pi_{k}$ from \cite{hamoudi2023quantum}.  The main goal of the section will be to argue that the double-hopping operator is very unlikely to create unpaired bosons.  At a high level, the idea is to notice that the double-hopping operator $\tilde{\H}_y$ picks two momentum modes $x$ and $x'$, and moves two bosons from modes $x$ and $x'$ to $x \oplus y$ and $x' \oplus y$, with a normalization that corresponds roughly to the square root of product of the number of bosons that occupy modes $x$, $x'$, $x\oplus y$, and $x' \oplus y$.  Since the ``weight'' of a pair of modes corresponds to the number of bosons in the modes, when we apply the double-hopping operator to a condensate, most of the mass of the post-operation state corresponds to moving bosons in or out of the $0$-momentum mode. This mostly preserves the ``oddness'' of the quasi-even condensate.  

A subtle difference between our analysis is the fact that the oracle from \cite{hamoudi2023quantum} can only output states from $\Pi_{k+1}$ when given a state in $\Pi_{k}$ as input, which means that concluding the proof only requires a simple inductive argument.  Oracle queries in our setting will roughly correspond to applying $\e^{- \kappa \tilde{\H}_y}$ for some $y$.  While $\tilde{\H}_y$ itself only ``moves'' $4$ bosons, the Taylor series of $\e^{- \kappa \tilde{\H}_y}$ includes arbitrarily high powers of $\tilde{\H}_y$, although we can bound these tails using a convenient expansion of the exponential from perturbation theory.  

\subsection{The double-hopping operator is almost always paired on condensates}

The central premise of this next subsection is that the action of the double-hopping operator on condensates is dominated by the action of hopping two bosons into or out of the 0-momentum mode. This is because the action of a creation or annihilation operator on a Fock state is proportional to the number of bosons in said mode. Since almost all the bosons are in the 0-momentum mode, it follows that the behavior of the 0-momentum mode dominates the actions. However, it is important to note that ignoring the other modes would only produce an approximation that would be too coarse for the sampling probability upper bounds we need to achieve. Nevertheless, we can show that on a quasi-even condensate, the action will, with almost certainty, produce another quasi-even condensate where the ``oddness'' is unlikely to be changed. We emphasize that the ``invariance'' we prove is only true if the original state is a condensate. To see otherwise, consider a state where the $\ell$ bosons are paired up in $\ell/2$ distinct momentum modes. Then the action of a double-hopping operator will almost certainly increase the ``oddness'' of the state by 2.

Concretely, in this subsection, we show that applying one double-hopping operator to a quasi-even condensate produces another quasi-even condensate, where the oddness is unlikely to be changed.  This section captures the main idea behind the proof of the quasi-even property.  

We first define the $y$-double difference set of a tuple $w$ and prove a bound on the number of elements of the double-difference set. For simplicity, we will use the following notation:
\begin{equation}
    \Delta_{x, x'}^{(y)} \defeq 1_{x \oplus y} + 1_{x' \oplus y} - 1_{x} - 1_{x'} \in \ZZ^{2^n} \,.
\end{equation}
With this, we define the double difference set.
\begin{definition}[Double difference set]
    Given $y \in \{0, 1\}^{n}$ and $y \neq 0^n$, define the $y$-double difference set of a tuple $w$ as follows.
    \begin{equation}
        \mathrm{diff}^2_y(w) \defeq \qty{ u \in \ZZ_{\geq 0}^{2^{n}}~:~ \exists ~(x, x')~\text{with}~x \not\in \{x', x'\oplus y\}~\text{and}~w = u + \Delta^{(y)}_{x, x'}}\,.
    \end{equation}
    Furthermore, when $u \in \mathrm{diff}^2_y(w)$, we write $(x, x') = \mathrm{diff}^2_y(u, w)$ to represent the choice of $x, x'$ that map from $w$ to $u$.  Without loss of generality, we take $x$ to be the lower of the two in lexicographical ordering (since $x \neq x'$ by assumption).
\end{definition} 
We make the following observation.
\begin{claim}
\label{claim:double_difference_set_bound}
    For every $(R, o)$-quasi-even condensate $u$, and every $y \in \{0, 1\}^{n}$ and $y \neq 0^n$,
    \begin{equation}
        \abs{\mathrm{diff}^2_y(u)} \leq (R+1)^2\,.
    \end{equation}
\end{claim}
\begin{proof}
    Since $u$ is a condensate with at most $R$ bosons not in the $0$-momentum mode, there are at most $R+1$ ways to choose $x$ and $R+1$ was to choose $x'$ so that the resulting tuple has non-negative entries after adding $\Delta_{x, x'}^{(y)}$, because they have to be taken from modes of $u$ that are strictly positive.  We will typically take the crude upper bound of $(R+1)^2 \leq 2R^2$ for suitably large $R$.
\end{proof}
Recall the definition of $\mathsf{QEC}_{r,o}$ as the projector onto momentum Fock states that are $\leq r$ condensates with $\leq o$ many non-zero modes having an odd number of bosons. And $\mathsf{Con}_r$ and $\mathsf{QE}_o$ are the projectors only ensuring the former and latter properties, respectively. Lastly, $\mathsf{QEC}_{r,=o}$ and $\mathsf{QE}_{=o}$ are the projectors requiring exactly $o$ many non-zero modes.
\begin{lemma}
\label[lemma]{lem:double_hopping_preserves_quasi_even}
    For every $R \in [\ell]$, $y \in \{0, 1\}^{n}$, with $y \neq 0^n$, the following holds.
    \begin{equation}
        \norm{\sum_{o \in [R]} \left(\id - \mathsf{QE}_{=o}\right) \cdot \tilde{\H}_y \cdot \mathsf{QEC}_{(R, =o)}}_{\mathrm{op}} \leq \frac{R^5}{\sqrt{\ell}}\,.
    \end{equation}
\end{lemma}
\begin{proof}
    We argue that the operator norm is small by showing that for every input state $\ket{\psi}$, the norm of the state shrinks by an appropriate factor.  Since there are at most $R$ choices of $o$ that we are considering, it will suffice to bound the norm of $(\id - \mathsf{QE}_o) \cdot \tilde{\H}_y \ket{\psi} $ for a state $\ket{\psi}$ supported on $\mathsf{QEC}_{(R, =o)}$ for a fixed $o$ and then apply the Cauchy-Schwarz inequality at the end. We can write $\ket{\psi}$ as a superposition of momentum Fock states $\ket{u}$, where $u$ is a $(R, =o)$-quasi-even condensate, as follows,
    \begin{equation}
        \ket{\psi} = \sum_{u\in \mathsf{QEC}_{(R, =o)}} \alpha_u \ket{u}\,.
    \end{equation}
    Expanding out the definition of the double-hopping operator, after applying $\tilde{\H}_y$, we have the following state.
    \begin{equation}
        \tilde{\H}_y \ket{\psi} = \sum_{u\in \mathsf{QEC}_{(R, =o)}} \alpha_{u} \left(\frac{1}{\ell} \sum_{x, x' \in \{0, 1\}^{n}} \sqrt{u_x u_{x'} (u_{x \oplus y} + 1) (u_{x'\oplus y} + 1)} \ket*{u + \Delta^{(y)}_{x, x'}}\right)\,.
    \end{equation}
    Applying the projector onto states that have $o' \neq o$ many odd entries, will remove all terms in the sum that correspond to $x= x'$ or $x=x' \oplus y$, as these never change the number of odd indices.  Thus, we can re-group the terms and use the definition of the double difference set to get the following.
    \begin{xalign}
        &\left(\id - \mathsf{QE}_{=o}\right) \cdot \tilde{\H}_y \ket{\psi} \\
        &\hspace{5mm}= \sum_{u\in \mathsf{QEC}_{(R, =o)}}  \left(\frac{\alpha_u}{\ell} \sum_{x, x' \in \{0, 1\}^{n}} \delta({u + \Delta_{x, x'}^{(y)} \not\in \mathsf{QE}_{=o}}) \sqrt{u_x u_{x'} (u_{x \oplus y} + 1) (u_{x'\oplus y} + 1)} \ket*{u + \Delta^{(y)}_{x, x'}}\right)\\
        &\hspace{5mm}= \sum_{u \in \mathsf{QEC}_{(R, =o)}} \frac{\alpha_{u}}{\ell} \sum_{\substack{x, x' \in \{0, 1\}^{n}\\ x \not\in \{x', x'\oplus y\}}} \delta({u + \Delta_{x, x'}^{(y)} \not\in \mathsf{QE}_{=o}}) \sqrt{u_x u_{x'} (u_{x \oplus y} + 1)(u_{x' \oplus y} + 1)} \ket*{u + \Delta_{x, x'}^{(y)}}\\
        &\hspace{5mm}=\frac{1}{\ell} \sum_{w \in \mathsf{QEC}_{(R+2, o+2)}} \delta({w \not\in \mathsf{QE}_{=o}})\left(\sum_{\substack{u \in \mathsf{QEC}_{(R,=o)}\\
        u \in \mathrm{diff}_y^2(w)\\
        (x, x') = \mathrm{diff}_y^2(u, w)}} 2 \alpha_{u} \sqrt{u_x u_{x'} (u_{x \oplus y} + 1)(u_{x' \oplus y} + 1)}\right) \ket{w}\,. \label{eq:odd_increase_state_after_projecting}
    \end{xalign}
    Here, we first applied the definition of $\tilde{\H}_y$ and used the fact that $\mathsf{QE}_{=o}$ is diagonal in the momentum Fock basis.  Then we use the fact that whenever $x = x'$ or $x' \oplus y$, the number of odd indices is preserved, which means we can remove those terms from the sum. Then we regrouped terms that have the same value of $u + \Delta_{x, x'}^{(y)}$, noting that every pair $(x, x') = \mathrm{diff}^2_y(u, w)$ appears twice in the sum over $x$ and $x'$ in $\tilde{\H}_y$, which gives us the multiplicative factor of $2$.  
    
    We now bound the squared norm of this state, using the fact that the $\ket{w}$ are orthogonal to each other, and the number of $u$ in $\mathrm{diff}^2_y(w)$ is bounded.  
    
    \begin{xalign}
        &\norm{(\id - \mathsf{QE}_{=o}) \cdot \tilde{\H}_y \ket{\psi}}^2 \\
        &\hspace{5mm}= \frac{1}{\ell^2} \norm{\sum_{w \in \mathsf{QEC}_{(R+2, o+2)}} \delta(w \not\in \mathsf{QE}_{=o})\left(\sum_{\substack{u \in \mathsf{QEC}_{(R, =o)}\\
        u \in \mathrm{diff}_y^2(w) \\ (x, x') = \mathrm{diff}_y^2(u, w)}}2\alpha_{u} \sqrt{u_x u_{x'} (u_{x \oplus y} + 1)(u_{x' \oplus y} + 1)}\right) \ket{w}}^2\\
        &\hspace{5mm}= \frac{4}{\ell^2} \sum_{w \in \mathsf{QEC}_{(R+2, o+2)}} \delta(w \not\in \mathsf{QE}_{=o})\abs{\sum_{\substack{u \in \mathsf{QEC}_{(R, =o)}\\
        u \in \mathrm{diff}_y^2(w)\\ (x, x') = \mathrm{diff}_y^2(u, w)}} \alpha_{u} \sqrt{u_x u_{x'} (u_{x \oplus y} + 1)(u_{x' \oplus y} + 1)}}^2\\
        &\hspace{5mm}\leq \frac{4}{\ell^2} \sum_{w \in \mathsf{QEC}_{(R+2, o+2)}} \sum_{\substack{u \in \mathsf{QEC}_{(R, =o)}\\
        u \in \mathrm{diff}_y^2(w)}} 2R^2 \cdot u_x u_{x'} (u_{x \oplus y} + 1)(u_{x' \oplus y} + 1)\cdot \abs{\alpha_u}^2\label{eq:product_of_u_x_in_sum}\\
        &\hspace{5mm}\leq \frac{8}{\ell^2} \sum_{w \in \mathsf{QEC}_{(R+2, o+2)}} \sum_{u \in \mathrm{diff}_y^2(w)} \ell R^5 \cdot \abs{\alpha_u}^2\\
        &\hspace{5mm}\leq \frac{8}{\ell^2} \sum_{u} \abs{\alpha_u}^2 \ell R^5 \abs{\mathrm{diff}_y^2(u)}\\
        &\hspace{5mm}\leq \frac{16 R^7}{\ell}\,.
    \end{xalign}
    Here, in the first line, we expand the definition of the state from \eqref{eq:odd_increase_state_after_projecting}.  Then we use the definition of the squared norm with the fact that $\ket{w}$ are orthogonal to each other.  Then we use Cauchy-Schwarz to bound the square of the sum by $2R^2$ times the sum of the squares, as there are at most $(R+3) \leq 2R^2$ elements of $\mathrm{diff}_y^2(w)$, applying \Cref{claim:double_difference_set_bound} to a $R+2$ condensate. In the second to last line, we switch the order of the sums and use the fact that $\mathrm{diff}_y^2$ is reflexive --- i.e., $w \in \mathrm{diff}_y^2(u) \Leftrightarrow u \in \mathrm{diff}_y^2(w)$. Then, we again use the fact that there are at most $2 R^2$ elements in $\mathrm{diff}_y^2(u)$ and that the $\abs{\alpha_u}^2$ sum to $1$.  
    Going from \eqref{eq:product_of_u_x_in_sum} to the next line, we note that by the definition of $\mathrm{diff}_y^2$, for every $(x, x') \in \mathrm{diff}_y^2(w)$, at most one of $x$, $x'$, $x\oplus y$ and $x' \oplus y$ can be $0$ since $y \neq 0^{n}$.  Therefore, at least three of the entries $u_{x}$, $u_{x'}$, $u_{x\oplus y} + 1$ and $u_{x' \oplus y} + 1$ can be bounded from above by $R$, the number of non-zero bosons, and the fourth one can be bounded by $\ell$, the total number of bosons.  Therefore, for all $x, x' \in \mathrm{diff}_y^2(u, w)$, $u_x u_{x'} (u_{x \oplus y} + 1)(u_{x' \oplus y} + 1) \leq \ell R^3$.  

    Taking the square root of this expression gives us the following bound for all $o$. 
    \begin{equation}
        \norm{(\id - \mathsf{QE}_{=o}) \cdot \tilde{\H}_y \cdot \mathsf{QEC}_{(R, =o)}}_{\mathrm{op}} \leq \frac{4 R^{7/2}}{\sqrt{\ell}}\,.
    \end{equation}
    Since $o$ can only be as large as $R$, with one more application of the triangle inequality, we can bound the norm of the sum over all $o$ from $0$ to $R$ by $4 R^{9/2} / \sqrt{\ell}$.  To make the numbers cleaner, we take $R^{5} / \sqrt{\ell}$ as an upper bound. 
\end{proof}

We note that the above bound holds for $\tilde{\H}_y$, but for the next section we will want to apply it to $\tilde{\G}_y = \tilde{\H}_y + \id$, which is the positive semi-definite operator that corresponds to applying $\gamma_y^{(S)}$ in the position Fock basis.  

\begin{mathinlay}
\begin{corollary}\label[corollary]{cor:QEboundGsquared}
    For all $R \in [\ell]$ and $y \in \{0, 1\}^{n}$ with $y \neq 0^n$, the following holds
        \begin{equation}
        \norm{\sum_{o \in [R]} \left(\id - \mathsf{QE}_{=o}\right) \cdot \tilde{\G}_y^2 \cdot \mathsf{QEC}_{(R, =o)}}_{\mathrm{op}} \leq \frac{R^5}{\sqrt{\ell}}\,.
    \end{equation}
\end{corollary}
\begin{proof}
    Using, $\tilde{\G}_y^2 = \tilde{\H}_y + \id$, we can apply \Cref{lem:double_hopping_preserves_quasi_even} in a straight-forward way as follows.
    \begin{xalign}
        &\norm{\sum_{o \in [R]} \left(\id - \mathsf{QE}_{=o}\right) \cdot \tilde{\G}_y^2 \cdot \mathsf{QEC}_{(R, =o)}}_{\mathrm{op}} \\
        &\hspace{1em}= \norm{\sum_{o \in [R]} \left(\id - \mathsf{QE}_{=o}\right) \cdot \tilde{\H}_y \cdot \mathsf{QEC}_{(R, =o)} + \sum_{o \in [R]} (\id - \mathsf{QE}_{=o})\cdot \mathsf{QE}_{=o} \cdot \mathrm{Con}_{R}}_{\mathrm{op}}\\
        &\hspace{1em}= \norm{\sum_{o \in [R]} \left(\id - \mathsf{QE}_{=o}\right) \cdot \tilde{\H}_y \cdot \mathsf{QEC}_{(R, =o)}}_{\mathrm{op}}\\
        &\hspace{1em}\leq \frac{R^5}{\sqrt \ell}\,.
    \end{xalign}
    Here, we simply expand out the definition of $\tilde{\G}_y^2$ and use the fact that $\mathsf{QE}_{=o}$ is a projector and therefore $(\id - \mathsf{QE}_{=o}) \mathsf{QE}_{=o} = 0$.  
\end{proof}
\end{mathinlay}

\subsection{Dyson series expansion of the exponential}

The previous section demonstrates that applying the double-hopping operator on a quasi-even condensate is likely to preserve the number of odd entries.  Unfortunately, although applying the double-hopping operator does not yield a state with high overlap on states with $\neq o$ odd entries, it does potentially increase the norm of the input state within the $=o$ subspace.  This is problematic to a naive inductive argument because after applying many double-hopping operators, one might have to scale the bound from \Cref{lem:double_hopping_preserves_quasi_even} by the growing norm of the state.

On the other hand, we expect that this is not a fundamental problem, because the operators that the algorithm applies are unitary, and even more so, operators of the form $\exp(-\tilde{\G}^2_y )$ clearly have bounded operator norm.  In this section, we show how to apply the Dyson series for the exponential function to lift \Cref{lem:double_hopping_preserves_quasi_even} to certain functions of the double-hopping operator.

\begin{mathinlay}
\begin{fact}[Application of Duhamel's principle]
    The following identity holds\footnote{To prove this statement, construct a function $f$ such that $f(0) = \e^{-tA}$ and $f(t) = \e^{-t(A+V)}$. Then the proof follows by observing $f(t) = f(0) + \int_{0}^t f'(s) \dd s$.} for all matrices $A$ and $V$, and $t \geq 0$,
    \begin{equation}
        \exp\left(-t\cdot (A+V)\right) = \exp\left(-t\cdot A\right) - \int_{0 \leq s \leq t} \dd s\ \exp\left(-(t-s) \cdot (A+V)\right) \cdot V \cdot \exp\left(-s \cdot A\right)\,.
    \end{equation}
\end{fact}

Applying this identity repeatedly gives us the following Dyson series for a small perturbation to $A$.  

\begin{fact}[Dyson's formula for the exponential function]
\label[fact]{lem:dyson_formula}
    For all matrices $A$ and $V$ and $\kappa \geq 0$, we can express $\exp\left(-\kappa \cdot (A + V)\right) - \exp\left(-\kappa \cdot A\right)$ as
    \begin{multline}
        \sum_{k = 1}^{\infty} (-1)^{k} \int_{0\leq s_1 \leq \ldots \leq s_k \leq \kappa} \dd \mathbf{s}\ \exp\left(-(\kappa-s_{k}) A\right) \cdot \underbrace{V\cdot \exp\left(-(s_k - s_{k-1}) A\right)\cdot V \ldots V}_{k~\mathrm{times}} \exp\left(-s_1 A\right) \,.
    \end{multline}
\end{fact}

The goal of Dyson's formula is to expand the exponential $\exp\left(-\kappa \cdot (A+V)\right)$ in such a way that the expansion only includes products of the term $V$ and exponentials of $A$. The use case is when $V$ is a small perturbation applied with $A$, and Dyson's formula allows us to expand the exponential and apply bounds we know about $V$.
In mathematical physics, $A$ is sometimes called the ``free'' part of the Hamiltonian, and $V$ the ``interacting'' part.  
\end{mathinlay}

We apply this to show that the exponential of the double-hopping operator does not change the number of odd entries, except with low probability.

\begin{lemma}
\label[lemma]{lem:exponential_does_not_increase_odds}
    The following bound holds for all $\kappa \leq 1$, $R \leq \ell^{1/10}/2$, and $1 \leq d \leq R$.
    \begin{equation}
        \norm{\sum_{o \in [R]} \mathsf{QE}_{\geq o+d} \cdot  \exp\left(-\kappa \left(\mathsf{Con}_{R} \cdot \tilde{\G}_y^2 \cdot \mathsf{Con}_{R} \right) \right) \cdot \mathsf{QE}_{o}}_{\mathrm{op}} \leq \left(\frac{2R^{5}}{\sqrt{\ell}}\right)^{d/4}\,.
    \end{equation}
\end{lemma}
\begin{proof}
    We use \Cref{lem:dyson_formula} applied to the following decomposition of $\mathsf{Con}_{R} \cdot \tilde{\G}_y^2 \cdot \mathsf{Con}_{R}$,
    \begin{align}
        \mathsf{Con}_{R} \cdot \tilde{\G}_y^2 \cdot \mathsf{Con}_{R} &= \underbrace{\sum_{o \in [R]} \mathsf{QEC}_{(R, =o)} \cdot \tilde{\G}_y^2 \cdot \mathsf{QEC}_{(R, =o)}}_{A} + \underbrace{\sum_{\substack{o, o' \in [R]\\o \neq o'}} \mathsf{QEC}_{(R, =o)} \cdot \tilde{\G}_y^2 \cdot \mathsf{QEC}_{(R, =o')}}_{V}\,.
    \end{align}
    We note that $V$ is equal to $\mathsf{Con}_{R} \sum_{o'} (\id - \mathsf{QE}_{=o'}) \cdot \tilde{\G}_y^2 \cdot \mathsf{QE}_{=o'}$, whose operator norm is bounded by \Cref{cor:QEboundGsquared}.  We further note that using the fact that $\tilde{\G}_y^2$ only moves $2$ bosons (and therefore affects at most $4$ momentum modes), and applying the definition of $A$, we have that
    \begin{equation}
        V \cdot \mathsf{QE}_{o} = \mathsf{QE}_{o+4} \cdot V \cdot \mathsf{QE}_o\quad \text{and}\quad A \cdot \mathsf{QE}_o = \mathsf{QE}_o \cdot A \cdot \mathsf{QE}_o\,.
    \end{equation}
    Repeatedly pushing the $\mathsf{QE}_o$ operator through, this implies that $\mathsf{QE}_{\geq o+d} \cdot \e^{-\kappa A} \cdot \mathsf{QE}_o = 0$ and for all $d' < d/4$ and $0\leq s_{1}\leq \ldots\leq s_{d'} \leq \kappa$,
    \begin{equation}
        \mathsf{QE}_{\geq o+d} \cdot \exp\left(-(\kappa-s_{d'}) \cdot A\right) \cdot \underbrace{V \ldots V}_{d' \text{ times}} \cdot \exp\left(-s_1 \cdot A\right) \cdot \mathsf{QE}_{o} = 0\,.
    \end{equation}
    Now applying \Cref{lem:dyson_formula} with $A$ and $V$ as defined in the equation above, we have the following bound on the norm for a fixed $o$,
    \begin{xalign}
        &\norm{\mathsf{QE}_{\geq o+d} \cdot \exp\left(-\kappa \left(\mathsf{Con}_{R} \cdot \tilde{\G}_y^2 \cdot \mathsf{Con}_{R} \right) \right) \cdot \mathsf{QE}_{o}}_{\mathrm{op}} \\
        &\hspace{2mm}= \norm{\sum_{k \geq d/4} (-1)^k \int_{0 \leq s_1 \leq \ldots \leq s_k \leq \kappa} \dd\mathbf{s}\ \mathsf{QE}_{\geq o+d} \cdot \exp\left(-(\kappa-s_k) \cdot A\right) \cdot V \ldots V \cdot \exp\left(-s_1 \cdot A\right) \cdot \mathsf{QE}_{o}}_{\mathrm{op}}\\
        &\hspace{2mm}\leq \sum_{k \geq d/4} \int_{0 \leq s_1 \leq \ldots \leq s_k \leq \kappa} \dd \mathbf{s}\ \norm{\mathsf{QE}_{\geq o+d} \cdot \exp\left(-(\kappa-s_k) \cdot A\right) \cdot V \ldots V \cdot \exp\left(-s_1 \cdot A\right) \cdot \mathsf{QE}_{o}}_{\mathrm{op}}\\
        &\hspace{2mm}\leq \sum_{k \geq d/4} \int_{0 \leq s_1 \leq \ldots \leq s_k \leq \kappa} \dd \mathbf{s}\ \norm{V}^k_{\mathrm{op}} \\
        &\hspace{2mm}\leq \sum_{k \geq d/4} \left(\frac{R^{5}}{\sqrt{\ell}}\right)^{k}=\frac{1}{1-R^5/\sqrt{\ell}}\left(\frac{R^{5}}{\sqrt{\ell}}\right)^{\lceil d/4 \rceil}\\
        &\hspace{2mm}\leq 2\left(\frac{R^{5}}{\sqrt{\ell}}\right)^{d/4} \\
        &\hspace{2mm} \leq \left(\frac{2R^5}{\sqrt{\ell}}\right)^{d/4}\,.
    \end{xalign}
    Here, we first apply Dyson's formula. Recall that $\tilde{\G}_y^2$ can only increase the number of odd indices by $4$. Moreover, by definition, $A$ does not change the number of odd indices and hence neither does any exponential in $A$. Therefore, since the expression is sandwiched between $\mathsf{QE}_{\geq o+d}$ and $\mathsf{QE}_{o}$, the only way to get a non-zero output is to be a product where the number of $V$ terms is at least $d/4$. This means that (1) we can eliminate $\exp(-\kappa \cdot A)$ arising from Dyson's formula, as well as (2) start the sum at $k\geq d/4$. Then we apply the triangle inequality to move the norm into the sum and integral.  Then we use the fact that because $\tilde{\G}_y^2$ is a positive operator (and sandwiching a positive operator by projectors yields a positive operator), $\exp\left(-t \cdot A \right)$ has operator norm at most $1$.  We then use the fact that the operator norm is sub-multiplicative and \Cref{cor:QEboundGsquared} bounds the operator norm of $V$ by $\frac{R^{5}}{\sqrt{\ell}}$.  Finally, we use the fact that $\int_{0 \leq s_1 \leq \ldots \leq s_k \leq \kappa} \dd s = \kappa^{k} / k! \leq 1$ whenever $\kappa \leq 1$ and bound the geometric series.
\end{proof}

\subsection{Recursive bounds on the oddness of products of operators}

The previous lemma bounds the probability that a state having $o$ odd indices ends up in the subspace of momentum Fock states with $o+d$ many odd indices, for every $d$.  We can apply a generalization of the stars-and-bars style counting argument to show that the probability of applying the operator many times does not increase the number of odd indices either.

\begin{mathinlay}
\begin{lemma}
\label[lemma]{lem:stars-and-bars-recursive-drop}
Let $\Pi_0 \preceq \Pi_1 \preceq \ldots$ be a sequence of projectors on a Hilbert space $\Hh$. Let $A_1, \ldots, A_t$ be a family of operators on $\Hh$ such that $\norm{A_i}_{\mathrm{op}} \leq 1$. Furthermore, suppose that for all integers $a, b \geq 0$, $\norm{\qty(\id - \Pi_{a+b}) A_i \Pi_a}_{\mathrm{op}} \leq \eps^{b+1}$. Then for all integers $\lambda \geq 0$ and states $\ket{\psi}$ such that $\Pi_{0} \ket{\psi} = \ket{\psi}$,
\begin{equation}
    \norm{\qty(\id - \Pi_\lambda) A_t \ldots A_1 \ket{\psi}} \leq {{t + \lambda} \choose t - 1} \eps^{\lambda+1}\,.
\end{equation}
\end{lemma}

\begin{proof}
    The statement for $t = 1$ is trivially true. Next define $\Pi_{-1} = 0$ and for$~j = 0, \ldots, \lambda,~\Delta_j = \Pi_j - \Pi_{j-1}$.
    Observe that $\Delta_j = \Pi_j (\id - \Pi_{j-1})$ by containment of the projectors. Then, write
    \begin{equation}\begin{aligned}
        \norm{\qty(\id - \Pi_\lambda) A_{t+1} \ldots A_1 \ket{\psi}} &\leq  \underbrace{\sum_{j = 0}^\lambda \norm{\qty(\id - \Pi_\lambda) A_{t+1} \Delta_j A_t \ldots A_1 \Pi_0\ket{\psi}} }_{(A)} \\
        &\hspace{2em}+ \underbrace{\norm{\qty(\id - \Pi_\lambda) A_{t+1} (\id - \Pi_\lambda) A_t \ldots A_1 \Pi_0 \ket{\psi}}}_{(B)}
    \end{aligned}\end{equation}
    We handle each of the terms separately. To bound $(B)$, we can use induction to bound
    \begin{xalign}
        \norm{\qty(\id - \Pi_\lambda) A_{t+1} (\id - \Pi_\lambda) A_t \ldots A_1 \Pi_0 \ket{\psi}} &\leq \norm{\qty(\id - \Pi_\lambda) A_{t+1} (\id - \Pi_\lambda)}_{\mathrm{op}} \cdot \norm{(\id - \Pi_\lambda) A_t \ldots A_1 \Pi_0}_{\mathrm{op}} \\
        &\leq \underbrace{\norm{A_{t+1}}}_{\le 1} \cdot {{t + \lambda} \choose t - 1} \eps^{\lambda+1}\,.
    \end{xalign} 
    For the terms in $(A)$, by induction,
    \begin{xalign}
        \sum_{j = 0}^\lambda \norm{\qty(\id - \Pi_\lambda) A_{t+1} \Delta_j A_t \ldots A_1 \Pi_0}_{\mathrm{op}} &= \sum_{j = 0}^\lambda \norm{\qty(\id - \Pi_\lambda) A_{t+1} \Pi_j (\id - \Pi_{j-1}) A_t \ldots A_1 \Pi_0}_{\mathrm{op}} \\
        &\leq \sum_{j = 0}^\lambda \norm{\qty(\id - \Pi_\lambda) A_{t+1} \Pi_j}_{\mathrm{op}} \cdot \norm{(\id - \Pi_{j-1}) A_t \ldots A_1 \Pi_0}_{\mathrm{op}} \\
        &\leq \sum_{j = 0}^\lambda \eps^{\lambda - j+1} \cdot {{t + j - 1} \choose {t-1}} \cdot \eps^{j} \\
        &= {{t + \lambda} \choose {t}} \eps^{\lambda +1}.
    \end{xalign}
    Adding up all the terms
    \begin{equation}
        (A)+ (B) \leq \qty({{t + \lambda } \choose t} + {{t + \lambda} \choose {t-1}}) \eps^{\lambda +1} = {{t + \lambda + 1} \choose t} \eps^{\lambda + 1}\,.
    \end{equation}
    completing the inductive proof.
\end{proof}
\end{mathinlay}

\subsection{Wrapping up the proof}

In this section, we apply the previous two lemmas to show that algorithms that make polynomial in $n$ queries to the $U$ oracle are exponentially close to the $\mathsf{QE}_{v/4}$ subspace.

\begin{lemma}
    The following bound holds for all $y \in \{0, 1\}^{n}$, with $y \neq 0^{n}$, $\kappa \leq 1$, $R \leq \ell^{1/10}/2$ and $1 \leq d \leq R$,
    \begin{equation}
        \norm{\sum_{o} \mathsf{QE}_{\geq o + d} \cdot \sqrt{1 - \frac{1}{2} \exp\left(-\kappa \left(\mathsf{Con}_{R} \cdot 
        \tilde{\G}_y^2 \cdot \mathsf{Con}_R\right)\right)} \cdot  \mathsf{QE}_{o}}_{\mathrm{op}} \leq \left(\frac{64R^5}{\sqrt{\ell}}\right)^{d/4}\,.
        \label{eq:square_root_of_exponential}
    \end{equation} 
    \label[lemma]{lem:crefdoesnotwork}
\end{lemma}
\begin{proof}
    We perform the proof by applying the previous two lemmas to the polynomial expansion of the square root.  The Taylor/binomial expansion of the square root is given by
    \begin{equation}
        \sqrt{1 + x} = \sum_{k = 0}^{\infty} \binom{1/2}{k} x^{k}\,.
    \end{equation}
    Applying it to the expression in \eqref{eq:square_root_of_exponential}, we have the following
    \begin{equation}
    \sqrt{1 - \frac{1}{2} \exp\left(-\kappa \left(\mathsf{Con}_R \cdot \tilde{\G}_y^2 \cdot \mathsf{Con}_R\right)\right)} = \sum_{k = 0}^{\infty} \binom{1/2}{k} \left(\frac{-1}{2}\right)^k \exp\left(-\kappa \left(\mathsf{Con}_R \cdot \tilde{\G}_y^2 \cdot \mathsf{Con}_R\right)\right)^{k}\,.
    \end{equation}
    We can bound the norm in \eqref{eq:square_root_of_exponential} using \Cref{lem:stars-and-bars-recursive-drop}, with $\lambda = d-1$, projectors $\Pi_{0} = \mathsf{QE}_{o} \preceq \Pi_{1} = \mathsf{QE}_{o + 1} \preceq \ldots \preceq \mathsf{QE}_{o + d - 1}$, and $A_1 = \ldots = A_{t} = \exp\qty(-\kappa\qty(\mathsf{Con}_R \cdot \tilde{\G}_y^2 \cdot \mathsf{Con}_R))$.
    \begin{xalign}
        &\norm{\mathsf{QE}_{\geq o + d} \cdot \sqrt{1 - \frac{1}{2}\exp\left(-\kappa \left(\mathsf{Con}_R \cdot \tilde{\G}_y^2 \cdot \mathsf{Con}_R\right)\right) } \cdot \mathsf{QE}_o}_{\mathrm{op}}\\
        &\hspace{5mm}= \norm{\sum_{k = 0}^{\infty} \binom{1/2}{k} \left(\frac{-1}{2}\right)^k \cdot \mathsf{QE}_{\geq o + d} \cdot \exp\left(-\kappa \left(\mathsf{Con}_R \cdot 
        \tilde{\G}_y^2 \cdot \mathsf{Con}_R\right)\right)^{k} \cdot \mathsf{QE}_o}_{\mathrm{op}}\\
        &\hspace{5mm}\leq \sum_{k = 1}^{\infty} \abs{\binom{1/2}{k}}\left(\frac{1}{2}\right)^k \cdot \norm{\mathsf{QE}_{\geq o + d} \cdot \exp\left(-\kappa\left(\mathsf{Con}_R \cdot \tilde{\G}_y^2 \cdot \mathsf{Con}_R\right)\right)^{k} \cdot \mathsf{QE}_o}_{\mathrm{op}}\label{eq:square_root_bound_line_c}\\
        &\hspace{5mm}\leq \left(\frac{2R^5} {\sqrt{\ell}}\right)^{d/4} \sum_{k = 1}^{\infty} \abs{\binom{1/2}{k}} \left(\frac{1}{2}\right)^{k} \binom{k + d}{k-1}\\
        &\hspace{5mm}\leq \left(\frac{2R^5} {\sqrt{\ell}}\right)^{d/4} \cdot 2^{d+1}\\
        &\hspace{5mm}\leq \left(\frac{64R^{5}}{\sqrt{\ell}}\right)^{d/4}\,.
    \end{xalign}
    In the first line, we apply the Taylor expansion of the square root from before, noting that the norm in the $k = 0$ term in \eqref{eq:square_root_bound_line_c} is $\norm{\mathsf{QE}_{\geq o + d} \cdot \mathsf{QE}_o}_{\mathrm{op}} = 0$. Then, we apply the triangle inequality.  In the third line, we apply \Cref{lem:stars-and-bars-recursive-drop}. together with the bound of $(2R^5/\sqrt{\ell})^{d/4}$ from \Cref{lem:exponential_does_not_increase_odds}.  For the final line, we bound the infinite sum by
    \begin{xalign}
        \sum_{k = 1}^{\infty} \abs{\binom{1/2}{k}} \left(\frac{1}{2}\right)^k\binom{k + d}{k-1} &\leq \sum_{k = 1}^{\infty} \left(\frac{1}{2}\right)^{k} \binom{k+d}{k-1}\\
        &=\frac{1}{2} \sum_{k = 0}^{\infty} \left(\frac{1}{2}\right)^{k} \binom{(d+2)+k-1}{k}\\
        &= \frac{1}{2} \sum_{k = 0}^{\infty} \left(-\frac{1}{2}\right)^{k} \binom{-(d+2)}{k}\\
        &= \frac{1}{2} \left(1 - \frac{1}{2}\right)^{-d-2}\\
        &= 2^{d+1}\,.
    \end{xalign}
    Here, in the first line, we use the fact that $\abs{\binom{1/2}{k}}\leq 1$ for all $k \geq 1$. Then we re-index the sum to start from $k = 0$, and write $k + d + 1 = (d+2) + k - 1$.  Then we use the equality $\binom{-a}{b} = (-1)^{b} \binom{a + b - 1}{b}$, with $a = d+2$ and $b = k$.  Finally, we use the binomial expansion, $(1-x)^{n} = \sum_{k = 0}^{\infty} \binom{n}{k} x^{k}$, with $n = -(d+2)$ and $x = -1/2$.  
\end{proof}
Recall, $e_0(z)$ and $e_1(z)$ were defined as
\begin{xalign}
    e_0(z) &\defeq 1-\e^{-\kappa z}, \quad\text{and} \\
    e_1(z) &\defeq \sqrt{\e^{-\kappa z}(2-\e^{-\kappa z})} = \sqrt{2} \e^{-\kappa z/2} \sqrt{1 - \e^{-\kappa z}/2}\,.
\end{xalign}
Since the Kraus operators $E_0$ and $E_1$ are defined in terms of the expressions given in~\Cref{lem:exponential_does_not_increase_odds} and~\Cref{lem:crefdoesnotwork}, we can combine these two lemmas to show that the two Kraus operators $E_0$ and $E_1$ also do not increase the number of odd entries by too much, and we then have the following lemma.
\begin{lemma}[Single query preserves quasi-evenness]
\label[lemma]{lem:single_query_low_collision}
    For all $o$, $R \le \ell^{1/10}/2$, $1 \le d \le R$, and operators $A$ with $\norm{A}_{\mathrm{op}} \leq 1$ acting on $\reg{A}$, the following inequality holds.
    \begin{equation}
        \norm{\mathsf{QE}_{\geq o+d} \left(\sum_{y} A \cdot \ketbra{y} \otimes \begin{pmatrix}\tilde{X}_{\reg{U}_y} \otimes e_1(\mathsf{Con}_R \cdot \tilde{\G}^{2}_{y} \cdot \mathsf{Con}_R) \\
        \hspace{1em}+ \tilde{Z}_{\reg{U}_y} \otimes e_{0}(\mathsf{Con}_R \cdot \tilde{\G}^2_{y} \cdot \mathsf{Con}_R)\end{pmatrix}\right) \mathsf{QE}_{o}}_{\mathrm{op}}
        \leq \left(\frac{2^{14}R^{5}d}{\sqrt{\ell}}\right)^{d/4}\,.
    \end{equation}
    Here, we have omitted $\id_{\reg{A}}$ on the projectors $\mathsf{QE}_{o}$ and $\mathsf{QE}_{\geq o+d}$.
\end{lemma}
\begin{proof}
    We expand out the operator norm as follows.
    \begin{xalign}
        &\norm{\mathsf{QE}_{\geq o+d} \left(\sum_{y} A \cdot \ketbra{y} \otimes \begin{pmatrix}\tilde{X}_{\reg{U}_y} \otimes e_1(\mathsf{Con}_R \cdot \tilde{\G}^{2}_{y} \cdot \mathsf{Con}_R) \\
        \hspace{1em}+ \tilde{Z}_{\reg{U}_y} \otimes e_{0}(\mathsf{Con}_R \cdot \tilde{\G}^2_{y} \cdot \mathsf{Con}_R)\end{pmatrix}\right) \mathsf{QE}_{o}}_{\mathrm{op}}\\
        &\hspace{5mm} = \norm{\sum_{y}A \cdot \ketbra{y} \otimes \mathsf{QE}_{\geq o+d}  \begin{pmatrix}\tilde{X}_{\reg{U}_y} \otimes e_1(\mathsf{Con}_R \cdot \tilde{\G}^{2}_{y} \cdot \mathsf{Con}_R) \\
        \hspace{1em}+ \tilde{Z}_{\reg{U}_y} \otimes e_{0}(\mathsf{Con}_R \cdot \tilde{\G}^2_{y} \cdot \mathsf{Con}_R)\end{pmatrix} \mathsf{QE}_{o}}_{\mathrm{op}}\\
        &\hspace{5mm}\leq \norm{\sum_{y} \ketbra{y} \otimes \mathsf{QE}_{\geq o+d}  \begin{pmatrix}\tilde{X}_{\reg{U}_y} \otimes e_1(\mathsf{Con}_R \cdot \tilde{\G}^{2}_{y} \cdot \mathsf{Con}_R) \\
        \hspace{1em}+ \tilde{Z}_{\reg{U}_y} \otimes e_{0}(\mathsf{Con}_R \cdot \tilde{\G}^2_{y} \cdot \mathsf{Con}_R)\end{pmatrix} \mathsf{QE}_{o}}_{\mathrm{op}}\\
        &\hspace{5mm} \leq \max_{y} \norm{\mathsf{QE}_{\geq o+d}  \begin{pmatrix}\tilde{X}_{\reg{U}_y} \otimes e_1(\mathsf{Con}_R \cdot \tilde{\G}^{2}_{y} \cdot \mathsf{Con}_R) \\
        \hspace{1em}+ \tilde{Z}_{\reg{U}_y} \otimes e_{0}(\mathsf{Con}_R \cdot \tilde{\G}^2_{y} \cdot \mathsf{Con}_R)\end{pmatrix} \mathsf{QE}_{o}}_{\mathrm{op}}\\
        &\hspace{5mm}\begin{matrix*}[l]\leq \max_{y} \norm{\mathsf{QE}_{\geq o+d} \left(\tilde{X}_{\reg{U}_y} \otimes e_1(\mathsf{Con}_R \cdot \tilde{\G}^{2}_{y} \cdot \mathsf{Con}_R)\right) \mathsf{QE}_{o}}_{\mathrm{op}}  \\
        \hspace{5mm}+ \max_{y}\norm{\mathsf{QE}_{\geq o+d} \left(\tilde{Z}_{\reg{U}_y} \otimes e_{0}(\mathsf{Con}_R \cdot \tilde{\G}^2_{y} \cdot \mathsf{Con}_R)\right) \mathsf{QE}_{o}}_{\mathrm{op}}\end{matrix*}\\
        &\hspace{5mm}\begin{matrix*}[l]\leq \underbrace{\max_{y} \norm{\mathsf{QE}_{\geq o+d} \cdot \left(e_1(\mathsf{Con}_R \cdot \tilde{\G}^{2}_{y} \cdot \mathsf{Con}_R)\right) \cdot  \mathsf{QE}_{o}}_{\mathrm{op}}}_{(A)} \\
        \hspace{5mm}+ \underbrace{\max_{y}\norm{\mathsf{QE}_{\geq o+d} \cdot\left(e_{0}(\mathsf{Con}_R \cdot \tilde{\G}^2_{y} \cdot \mathsf{Con}_R)\right)\cdot \mathsf{QE}_{o}}_{\mathrm{op}}}_{(B)}\,.\end{matrix*}
    \end{xalign}
    Here, we first use the fact that $\mathsf{QE}_{\geq o+d}$ acts only on the purifying register, then use the fact that $\norm{A}_{\mathrm{op}} \leq 1$.  Then we use the fact that for any operator that is block-diagonal, the operator norm is the max of the operator norms restricted to each block.  Finally, we apply the triangle inequality.  Now we bound parts $(A)$ and $(B)$ separately using~\Cref{lem:exponential_does_not_increase_odds} and~\Cref{lem:crefdoesnotwork}.  We have the following
    \begin{xalign}
        (A) &= \max_{y} \sqrt{2}\norm{\mathsf{QE}_{\geq o + d} \cdot \left(\e^{-(\kappa/2)\cdot \qty(\mathsf{Con}_R \cdot \tilde{\G}^2_{y} \cdot \mathsf{Con}_R)})\sqrt{1 - \frac{1}{2} \e^{-\kappa \qty(\mathsf{Con}_R \cdot \tilde{\G}^2_{y} \cdot \mathsf{Con}_R)}}\right) \cdot \mathsf{QE}_{o}}_{\mathrm{op}}\\
        &\leq \sqrt{2} (d + 1)\left(\frac{64 R^5}{\sqrt{\ell}}\right)^{d/4}\,.
    \end{xalign}
    Here, for every $y$, we use \Cref{lem:stars-and-bars-recursive-drop} with $\lambda = d-1$, the same projectors $\Pi_{1} = \mathsf{QE}_o \preceq \ldots \Pi_{d-1} = \mathsf{QE}_{o+d}$, and the $2$ operators $A_{2} = \e^{-(\kappa/2) \qty(\mathsf{Con}_R \cdot \tilde{\G}^2_{y} \cdot \mathsf{Con}_R)}$ and $A_1 = \sqrt{1 - \frac{1}{2} \e^{-\kappa \qty(\mathsf{Con}_R \cdot \tilde{\G}^2_{y} \cdot \mathsf{Con}_R)}}$.  \Cref{lem:exponential_does_not_increase_odds} and \Cref{lem:crefdoesnotwork} bound each $\norm{(\id - \Pi_{d-1}) A_i \Pi_o}_{\mathrm{op}} \leq (64R^5 / \sqrt{\ell})^{d/4}$, completing the bound.  
    
    Similarly, we bound $(B)$ as follows.
    \begin{xalign}
        (B) &\leq \max_{y} \norm{\mathsf{QE}_{\geq o + d} \cdot \left(\id - \e^{-\kappa \qty(\mathsf{Con}_R \cdot \tilde{\G}^2_{y} \cdot \mathsf{Con}_R)}\right) \cdot \mathsf{QE}_o}\\
        &\leq \left(\frac{2R^5}{\sqrt{\ell}}\right)^{d/4}\,.
    \end{xalign}
    Here, we directly apply \Cref{lem:exponential_does_not_increase_odds}, together with the fact that $\mathsf{QE}_{\geq o + d} \cdot \mathsf{QE}_{o} = 0$ so we can remove the identity term.  
    Putting things together, we have an upper bound on the operator norm of
    \begin{equation}
        (A) + (B) \leq 2d\left(\frac{64R^5}{\sqrt{\ell}}\right)^{d/4} \le \left(\frac{2^{14} R^5}{\sqrt{\ell}}\right)^{d/4} \,.
    \end{equation}
    Here, we take the upper bound $2d \leq 2^{d+1}$ in the final line.
\end{proof}

The previous lemma applies to a standard query $y$. However, if the verification algorithm applies a conditional query and in the case that the conditional query is not applied, it is clear to see that the algorithm cannot change the quasi-evenness of the $\reg{S}$ register.  Formally, we have the following.

\begin{corollary}[Single controlled query preserves quasi-evenness]
\label[corollary]{cor:controlled_query_quasi_even}
    For all $o$, $R \le \ell^{1/10}/2$, $1 \le d \le R$, and unitaries $A$ acting on $\reg{A}$, the following inequality holds.
    \begin{equation}
        \norm{\mathsf{QE}_{\geq o+d} \left(\sum_{y, b} A \cdot \ketbra{b, y} \otimes \begin{pmatrix}\tilde{X}_{\reg{U}_y} \otimes e_1(\mathsf{Con}_R \cdot \tilde{\G}^{2}_{y} \cdot \mathsf{Con}_R) \\
        \hspace{1em}+ \tilde{Z}_{\reg{U}_y} \otimes e_{0}(\mathsf{Con}_R \cdot \tilde{\G}^2_{y} \cdot \mathsf{Con}_R)\end{pmatrix}^{b}\right) \mathsf{QE}_{o}}_{\mathrm{op}}
        \leq \left(\frac{2^{14} R^{5}d}{\sqrt{\ell}}\right)^{d/4}\,.
    \end{equation}
    Here, we have omitted $\id_{\reg{A}}$ on the projectors $\mathsf{QE}_{o}$ and $\mathsf{QE}_{\geq o+d}$.
\end{corollary}
\begin{proof}
    Applying the triangle inequality, we have that
    \begin{xalign}
        &\norm{\mathsf{QE}_{\geq o+d} \left(\sum_{y, b} A \cdot \ketbra{b, y} \otimes \begin{pmatrix}\tilde{X}_{\reg{U}_y} \otimes e_1(\mathsf{Con}_R \cdot \tilde{\G}^{2}_{y} \cdot \mathsf{Con}_R) \\
        \hspace{1em}+ \tilde{Z}_{\reg{U}_y} \otimes e_{0}(\mathsf{Con}_R \cdot \tilde{\G}^2_{y} \cdot \mathsf{Con}_R)\end{pmatrix}^{b}\right) \mathsf{QE}_{o}}_{\mathrm{op}}\\
        &\hspace{5mm} \begin{matrix*}[l] \le \norm{\mathsf{QE}_{\geq o+d} \left(\sum_{y} A \cdot \ketbra{1, y} \otimes \begin{pmatrix}\tilde{X}_{\reg{U}_y} \otimes e_1(\mathsf{Con}_R \cdot \tilde{\G}^{2}_{y} \cdot \mathsf{Con}_R) \\
        \hspace{1em}+ \tilde{Z}_{\reg{U}_y} \otimes e_{0}(\mathsf{Con}_R \cdot \tilde{\G}^2_{y} \cdot \mathsf{Con}_R)\end{pmatrix}\right) \mathsf{QE}_{o}}_{\mathrm{op}} \\
        \hspace{5em} +  \norm{\mathsf{QE}_{\geq o+d} \left(\sum_{y} A \cdot \ketbra{0, y} \otimes \id\right) \mathsf{QE}_{o}}_{\mathrm{op}}\end{matrix*}\\
        &\hspace{5mm}\leq \left(\frac{2^{14}R^{5}d}{\sqrt{\ell}}\right)^{d/4} + 0\,.
    \end{xalign} 
    Here, we use the fact that when $b = 0$, the inner unitary is the identity, and then we use the fact that $\mathsf{QE}_{\geq o + d}$ and $\mathsf{QE}_{o}$ are projectors onto orthogonal subspaces since $d \geq 1$ and they commute past $A$ and $\ketbra{0, y}$.  When $b = 1$, we use the fact that $\norm{A \ketbra{1}}_{\mathrm{op}} \leq 1$ and apply \Cref{lem:single_query_low_collision}.
\end{proof}

\begin{theorem}
\label{thm:low-collision-overlap}
    Let $\mathcal{A}$ be any $T$-query algorithm making queries to $U$ and $v \in \mathbb{Z}_{\geq 0}$ be a multiple of 4.  Then there exists a constant $c$ such that for suitably large $n$, the purified state of the algorithm has squared overlap with the complement of the $(c \cdot T^{10}, v/4)$-low collision subspace that is lower bounded by 
    \begin{equation}
        \norm{\left(\id - \mathsf{QEC}_{(c \cdot T^{10}, v/4)}\right) \ket{\psi_{\mathrm{PQ}}}}^2 \leq \left(\left(\frac{T^{4}}{\ell^{1/32}}\right)^{v} + \e^{-5T}\right)^2\,.
    \end{equation}
\end{theorem}
\begin{proof}
    Assume without loss of generality that $T \geq n$.  Then we apply \Cref{cor:sandwiching} with $\iota = \e^{-5T}$ to show that there exist integers $r = O(T^{10})$ and $R = O(T^{13})$ to get that 
    \begin{equation}
        \norm*{\ket*{\psi_{\mathrm{PQ}}} - \ket*{\psi_{R, r}}} \leq \e^{-5T}\,.
    \end{equation}
    Then, we apply~\Cref{lem:stars-and-bars-recursive-drop} with $\lambda = v/4$, $\Pi_{0} = \mathsf{QEC}_{r, 0} \preceq \ldots \preceq \Pi_{v/4} = \mathsf{QEC}_{r, v/4}$ and operators $A_i$ being
    \begin{equation}
        A_i = \left(\sum_{y, b} A \cdot \ketbra{b, y} \otimes \begin{pmatrix}\tilde{X}_{\reg{U}_y} \otimes e_1(\mathsf{Con}_R \cdot \tilde{\G}^{2}_{y} \cdot \mathsf{Con}_R) \\
        \hspace{1em}+ \tilde{Z}_{\reg{U}_y} \otimes e_{0}(\mathsf{Con}_R \cdot \tilde{\G}^2_{y} \cdot \mathsf{Con}_R)\end{pmatrix}^{b}\right)\,,
    \end{equation}
    where $A$ here is the unitary that the verification algorithm applies, and we note that there are $T \leq \ell^{1/10}/2$ of them.  We also use the fact that the initial state is contained in $\mathsf{QE}_{0}$.  Then~\Cref{cor:controlled_query_quasi_even} gives us the bound on the individual norms of $\left(\frac{64R^5 (v/4)}{\sqrt{\ell}}\right)^{(v/4+1)/4}$.  Resolving \Cref{lem:stars-and-bars-recursive-drop}, we have
    \begin{xalign}
        \norm{\left(\id - \mathsf{QEC}_{(r, v/4)}\right) \cdot \ket{\psi_{R, r}}} &\leq \binom{T + v/4}{T - 1}\left(\frac{2^{14} R^5}{\sqrt{\ell}}\right)^{(v/4+1)/4}\\
        &\leq (2T)^{v/4} \left(\frac{2^{14}  R^5}{\sqrt{\ell}}\right)^{v/16}\\
        &\leq \left(\frac{2^{18} R^5T^4}{\sqrt{\ell}}\right)^{v/16}\,.
    \end{xalign}
    Combining these two equations using the triangle inequality, we have that 
    \begin{xalign}
        \norm{\left(\id - \mathsf{QEC}_{(r, v/4)}\right) \cdot \ket*{\psi_{\mathrm{PQ}}}} &\leq \left(\frac{2^{18} R^5 T^4}{\sqrt{\ell}}\right)^{v/16} + \e^{-5T}\\
        &\leq \left(\frac{4 R^{5/16} T^{1/4}}{\ell^{1/32}}\right)^{v} + e^{-5T}\,.
    \end{xalign}
    Substituting $R = c \cdot T^{13}$ and $r = c \cdot T^{10}$ for some constant $c$ (for suitably large $n$), for sufficiently large $T$, we have that $4R^{5/16}T^{1/4} \le T^4$ which completes the theorem.
\end{proof}

 \paragraph{Proving the main sampling upper bound}
Combining the gentle measurement lemma, \Cref{thm:low-collision-overlap}, and \Cref{thm:main-sampling-upper-bound}, which bounds the probability of sampling $v$ many points given that the state is in $\mathsf{QEC}_{(c\cdot T^{10}, v/4)}$ (with $T = vt$ being the number of queries), we achieve the main result of~\Cref{part:sampler-upper-bound,part:poly-query-implies-qec}:~\Cref{thm:big-upper-bound-thm}.
\begin{theorem*}[Sampling probability upper bound,\Cref{thm:big-upper-bound-thm}, restated]
    For all $v$, for all quantum algorithms $\Aa^U$ accessing an oracle $U$ and outputting $v$ distinct outputs, while making $t$ queries per output, if a pair of oracles $(S,U)$ are sampled according to distribution $\mathsf{Strong}$ (defined in~\Cref{def:def-of-strong-dist}), then the probability that all $v$ of the outputs of $\Aa^U$ are elements of $S$ is at most
    \begin{equation}
        \leq 2\left(\frac{4v((vt)^{30} + v (vt)^{20})\sqrt{\ell}}{2^{n/4}}\right)^{v} + \left(\left(\frac{(vt)^{4}}{\ell^{1/32}}\right)^{v} + e^{-5vt}\right)^2\,.
    \end{equation}
\end{theorem*}

\begin{proof}[Proof of~\Cref{thm:big-upper-bound-thm}]
    Let $\ket*{\psi_{\mathrm{PQ}}}$ be the state of the algorithm after making $vt$ queries to $U$, then we have, for $\Pi_{\mathrm{succ}}$ being the success operator from~\eqref{eq:pi_success}.
    \begin{xalign}
        \norm{\Pi_{\mathrm{succ}}\ket{\psi_{\mathrm{PQ}}}}^2 &\leq  \norm{\Pi_{\mathrm{succ}}\cdot \mathsf{QEC}_{(c \cdot v^{10}t^{10}, v/4)} \cdot \ket{\psi_{\mathrm{PQ}}}}^2 + \norm{\Pi_{\mathrm{succ}}\cdot \qty(\id - \mathsf{QEC}_{(c\cdot v^{10}t^{10}, v/4)})\ket{\psi_{\mathrm{PQ}}}}^2\\
        &\leq 2\left(\frac{4v((vt)^{30} + v (vt)^{20})\sqrt{\ell}}{2^{n/4}}\right)^{v} + \left(\left(\frac{(vt)^{4}}{\ell^{1/32}}\right)^{v} + e^{-5vt}\right)^2\,.
    \end{xalign}
    Here, we apply the triangle inequality and bound the second term using \Cref{thm:low-collision-overlap}, and the first term using \Cref{thm:main-sampling-upper-bound} with $T = vt$.
\end{proof}

\newpage
\part{Theorem statements and concluding remarks}
\label{part:concluding-remarks}
\section{Property-testing and oracle separations}

We can now combine both our sampler success probability upper bound with the lower bound on the success probability implied by a $\mathsf{QCMA}$ algorithm to get a lower bound on the witness length of a successful $\mathsf{QCMA}$ algorithm.

\begin{theorem} \label{thm:qcma-lower-bound}
   Consider any choice of constants $a > 0$ and functions $t(n), q(n)$ that satisfy $t(n) \leq a n^a, q(n) \leq a n^a$ for all $n \geq n_0$. Then, let $n_0$ be the smallest integer such that for all $n \geq n_0$, both~\Cref{thm:formal-removal-of-S} and~\Cref{thm:big-upper-bound-thm} apply using for $t = t(n), q = q(n)$, and $v = 1000q$.  Then for any $n \geq n_0$, and $\mathcal{A}$ a binary-output quantum query algorithm with classical witness of length $q(n)$ and making $t(n)$ queries to the oracles $(S, U)$ of size $n$, there exists a pair of oracle $(S^*,U^*)$ of size $n$ such that
\begin{enumerate}
    \item either $(S^*,U^*)$ are at least $\frac{59}{100}$-spectrally Forrelated, but for all witnesses $w$ of length $q(n)$,
    \begin{equation}
        \Pr[ \Aa^{(S^*,U^*)}(w) = 1] < \frac{2}{3}\,.
    \end{equation}
    \item or $(S^*,U^*)$ are at most $\frac{57}{100}$-spectrally Forrelated, but there exists a witness $\tilde w$ of length $q(n)$ such that
    \begin{equation}
        \Pr[ \Aa^{(S^*,U^*)}(\tilde w) = 1] > \frac{1}{3}\,.        
    \end{equation}
\end{enumerate}
\end{theorem}

\begin{proof}
    If the stated consequence was false, then the algorithm $\Aa$ properly classifies all spectral Forrelation problems promised that either the instance is at least $59/100$- or at most $57/100$-spectrally Forrelated of size $n$. However, setting $\ell = 2^{n/10}$ and $v = 1000 q$, then we can apply~\Cref{thm:formal-removal-of-S} and~\Cref{lem:strong-instance} with $\kappa = 1/10$ and $\rho= \frac{2\ell^2}{2^{n}} \ln\left(\frac{2^{n}}{2\ell^{4}}\right)$, yielding a sampler that outputs $v$ points with probability $O(t^{-1000q})$ when $(S, U)$ is sampled from the distribution $\mathsf{Strong}$.  Applying~\Cref{thm:big-upper-bound-thm} yields a sampling probability upper bound of $O\left((\poly(n)2^{-n/160})^{1000q}\right)$, whenever $t$ and $q$ are $\leq a n^a$, for the distribution$~\mathsf{Strong}$. For suitably large $n$, $(\poly(n)2^{-n/160})^{1000q} \ll t^{-1000q}$, concluding the proof by contradiction.
\end{proof}

We have shown in~\Cref{thm:qcma-lower-bound} that for every sufficiently large instance size $n$, there exists a property testing problem about size $n$ oracles such that there exists an $n$-qubit quantum witness for the problem verifiable by an efficient quantum algorithm. Still, there does not exist any $t = t(n)$ time quantum algorithm accepting $q = q(n)$-length classical witnesses for the problem, where both $t$ and $q$ are polynomial in $n$. This is technically not yet an oracle separation between $\QCMA$ and $\QMA$.

\paragraph{The oracle separation language}
Let $\fn{S,U}{\bits^*}{\bits}$ be oracles and let $S_n,U_n$ be the restriction  to $n$-bit inputs. We define the \emph{unary} language $\Ll^{S,U}$ as follows:
\begin{xalign}
    &1^n \in \Ll^{S,U}~\iff \fn{S_n,U_n}{\bits^n}{\bits}~\text{are}~\ge \frac{59}{100}\text{-spectrally Forrelated},\\
    &1^n \notin \Ll^{S,U}~\iff \fn{S_n,U_n}{\bits^n}{\bits}~\text{are}~\le \frac{57}{100}\text{-spectrally Forrelated}, \\
    &\text{and,}~\text{all other strings~}\in\bits^*~~\text{are}~\notin \Ll^{S,U}.
\end{xalign}
This language will be what we use to prove~\Cref{thm:main}.  We note that $1^{n}$ being in the language is determined by the oracle of size $n+1$ bits, as the pair of oracles $(S_n, U_n)$ is an $n+1$ bit oracle. Roughly speaking, we will choose the oracles $S,U$ such that the $k$-th polynomial-time uniform quantum query algorithm in an enumeration of quantum query algorithms will incorrectly identify membership of $1^{n_k}$ in $\Ll$ of an appropriately chosen integer $n_k$. \\

 As one might suspect, we will use~\Cref{thm:qcma-lower-bound} to diagonalize against all polynomial-sized classical witness oracle algorithms. The challenge is that the algorithms can query the oracle at any length---intuitively, on input $1^n$, querying at lengths $\neq n$ will not help the algorithm. However, formalizing this intuition is tedious. We do this in the following proof by identifying a family of appropriately chosen integers $n_1 < n_2 < \ldots$ such that the algorithm on input of length $n_k$ will not query at lengths $\geq n_{k+1}$. This will be enough to inductively select the definition of the oracle at size $n_k$  (based on the choices of the oracle at all sizes $< n_k$) to complete the diagonalization.

\begin{proof}[Proof of~\Cref{thm:main}]
    For any choice of oracles $S,U$ such that each restriction to size $n$ inputs encodes either a $\geq 59/100$ or $\leq 57/100$ instance of spectral Forrelation, the containment of $\Ll^{S,U}$ in $\QMA^{S,U}$ follows from \Cref{thm:spectral_forrelation_qma_containment}. Now we prove the lower bound for $\QCMA$ algorithms. \\

     Let $M_1, M_2, \ldots$ be an enumeration of all possible Turing machines. Second, identify any surjective function $\iota: \NN \twoheadrightarrow \NN^2$ and define functions $j, a: \NN \rightarrow \NN$ by $(j(k), a(k)) = \iota(k)$. We will use this surjection to diagonalize against all possible polynomial-time quantum algorithms to prove a $\QCMA$-lower bound. \\

     Third, define a function $F: \NN \rightarrow \NN$ by $F(a)$ is the minimum value such that for all $n \geq F(a)$, any $t(n) = a n^a$ query algorithm with $q(n) = a n^a$ length classical witness must misclassify some pair $(S,U)$ of size $n$. By~\Cref{thm:qcma-lower-bound}, for every integer $a$, $F(a)$ is well-defined. In other words, $F(a)$ is the first integer $n$ for which we can guarantee that an $a n^a$ query and $a n^a$ witness restricted algorithm must misclassify some pair $(S,U)$ of size $n$. \\

     Fourth, we identify integers $n_1, n_2, \ldots$ where the oracles will be defined to be non-zero. Define
    \begin{xalign}
        n_1 &\defeq 1 + F(a(1)), \\
        \forall~k>1,~n_k &\defeq 1 + \max \qty{ F(a(k)), ~a(k-1) \qty(n_{k-1})^{a(k-1)}}.
    \end{xalign}
    We now define the oracles $S,U$ by defining the oracles at each length. For any $n \in \NN \setminus \{n_1, n_2, \ldots\}$, let both oracles $S_n$ and $U_n$ equal $0$ everywhere. Then the spectral-Forrelation problem defined by the $n$-th pair of oracles is trivially a no instance for $n \in \NN \setminus \{n_1, n_2, \ldots\}$. We now go through and define the oracles at sizes remaining sizes: $n_1, n_2, \ldots,$. \\

     For each $k = 1, 2, 3, \ldots$, run Turing machine $M_{j(k)}$ on input $1^{n_k}$ for $a(k) n_k^{a(k)}$ steps and interpret its output as a quantum query circuit $\Aa_{n_k}$ which takes as input a classical witness. By adding the halting conditions to the Turing machine, we have implicitly enforced that the witness length and total number of gates (elementary or oracle) are at most $a(k) n_k^{a(k)}$. Furthermore, by construction, the largest input queryable by this algorithm is size $a(k) n_k^{a(k)}$. This is a standard diagonalization trick to ensure that every polynomial-time and polynomial-query algorithm is considered and that no super-polynomial parameterized algorithms are accidentally considered. \\

     Next, we take this algorithm $\Aa_{n_k}$ and we build from it a query algorithm $\Bb_{n_k}$ that only queries oracles of input length $n_k$.  To construct $\Bb_{n_k}$, take the algorithm $\Aa_{n_k}$ which can make oracle queries of varying sizes and for every query it makes of length $< n_k$, use the previously generated definitions of the oracles $S,U$ and hardcode these answers. This is well defined as we are defining the oracles $S,U$ for progressively larger input sizes. For queries it makes of length $> n_k$, replace the oracle gates with identity circuits. The resulting circuit will be $\Bb_{n_k}$, which only makes queries of length $= n_k$. This new algorithm $\Bb_{n_k}$ can be used to derive a pair $(S_{n_k}, U_{n_k})$ by applying~\Cref{thm:qcma-lower-bound} on $\Bb_{n_k}$ to generate the pair $(S_{n_k}, U_{n_k})$. \\

     This completes the construction of the oracles $S,U$ everywhere. It remains to prove that no $\QCMA^{S,U}$ algorithm exists. Assume, for contradiction, there exists a $\P$-uniform family of quantum oracle algorithms $\{\Aa_n\}$ that solve spectral Forrelation for a witness of length $q(n) =\poly(n)$ with $t(n) = \poly(n)$ queries. Then, the family appears in the Turing machine enumeration as some $M_{j^\star}$ and there exists some $a^\star$ such that $t(n),q(n) \leq a^\star n^{a^\star}$. As $\iota$ is a surjection, there exists a $k^\star$ such that $\iota(k^\star) = (j^\star, a^\star)$. We now prove that this algorithm will misclassify the string $1^{n_{k^\star}}$, thereby proving that it does not solve spectral Forrelation. \\

     To prove this, let $\Aa_{n_{k^\star}}$ be the quantum circuit for inputs of length $n_{k^\star}$. Observe that since the oracles are defined as $= 0$ for inputs $\notin \{n_1, n_2, \ldots\}$, and the fact that $n_{k^\star+1} > a^\star n_{k^\star}^{a^\star}$ (by construction), each query gate for inputs of length $> n_{k^\star}$ is effectively an identity gate as it only makes queries on inputs of length at most $a^\star n_{k^\star}^{a^\star}$. Therefore, by hardcoding the behavior on input sizes $< n_{k^\star}$, the query algorithm $\Bb_{n_{k^\star}}$ (previously defined) has the exact same output as $\Aa_{n_{k^\star}}$ on inputs of size $n_{k^\star}$. However, using~\Cref{thm:qcma-lower-bound}, we specifically constructed a pair $(S_{n_k}, U_{n_k})$ that $\Bb_{n_{k^\star}}$ will misclassify. Therefore, the family $\{\Aa_n\}$ will answer incorrectly on input $1^{n_{k^\star}}$, completing the proof.

\end{proof}

\section{Concluding remarks}
\label{sec:concluding-remarks}
Having concluded the proof, we take a moment to reflect on the path that led us here. In particular, we discuss the role of the bosonic purification of the oracle, its conceptual advantages, and the technical challenges that arose in this formulation.

\paragraph{Zhandry's observation.}
Zhandry's observation~\cite{zhandry2024toward} concerned the \emph{use-once} nature of quantum witnesses versus the \emph{reusability} of classical witnesses. This distinction allowed him to construct a sampler that, given oracle access to $S$, could produce multiple distinct samples from $S$, contingent on a technical conjecture. 

The sampler perspective was appealing, as it delineated the boundary between what a classical witness can explicitly list about $S$ and what additional structure must be inferred by the verifier through oracle queries. However, the argument hinged on a conjecture asserting that queries to the unitary oracle $U$ were \emph{computationally} indistinguishable from random---a conjecture which ultimately fails. While Zhandry’s setting involved the quantum Fourier transform and ours employs a related but distinct Fourier-analytic framework, the key issue was the same: queries to $U$ are \emph{not} random, and an algorithm could detect that they are far from uniformly random. The saving observation is that their behavior is amenable to precise analysis through Fourier and, ultimately, bosonic tools. This is made precise by our analysis in the bosonic framework, which exactly characterizes the action of queries to $U$. Our proof can circumvent Zhandry's issue as it proves a sampling probability upper bound based on the bosonic characterization. In some sense, our proof suggests that the $\QMA$ vs. $\QCMA$ separation problem is not as related to pseudorandomness as prior results~\cite{liu2024qma,zhandry2024toward} have suggested.

Recognizing this failure was serendipitous. It led us to the insight that one could instead adapt the sampler proof to operate solely through access to $U$, while still producing valid samples from $S$. This realization shifted our focus entirely to understanding the detailed structure of oracle queries in the $U$-picture.

\paragraph{Strong sampling upper bounds} Recall that the sampling probability lower bound we derive is incredibly small. Its scaling is roughly $2^{-q} \cdot \Omega(t)^{-2v}$ for $q$, the length of the witness, $t$ the number of queries, and $v$ the number of samples produced. However, as long as $t$ is polynomial, this bound is only quasi-exponentially small. Therefore, a contradicting upper bound should be even smaller in order to contradict the existence of a $\QCMA$ algorithm in all polynomial parameter regimes. This stringent requirement means that it does not suffice to approximate the behavior of the state after each query or guess. In particular, what's important to study is the post-selected state of the first guesses $z_1, \ldots, z_{k}$ being correct. Each post-selection is conditioned on an event of exponentially small probability, and therefore, the post-selection state will be very far from the original state. 

In particular, we spent considerable time chasing the idea that the post-measurement state after the first $k$ guesses have been verified as correct will almost certainly be supported on Fock states with total momentum 0. Total momentum 0, by Noether's theorem, implies that any position guess will only succeed with negligible probability. While one can show using the $\tilde{\H}_y$ formalism that the first guess will be on a total momentum 0 state, for future states this can only be guaranteed up to negligible additive error; the issue is that this error compounds, leading to a trivial probability upper bound. This issue forced us to construct a sampling upper bound technique that constructed a probability upper bound for all $v$ samples at once using the quasi-even condensate structure.

\paragraph{The bosonic perspective.}
Once we accepted that the heart of the matter was to analyze queries to $U$, we turned to a wide array of standard quantum query techniques---compressed oracles, polynomial methods, adversary arguments---only to find each approach mired in technical complexity. The crux of the difficulty was purification: expressing $S$ and $U$ as parts of a single coherent quantum system. 

In compressed oracle techniques, one would ideally purify $S$ via independent amplitudes, e.g., sampling each element’s inclusion in $S$ with probability $\ell/2^n$. Yet this fails spectacularly for our setting, since each evaluation of $U$ depends non-locally on \emph{all} elements of $S$. The natural purification $(\cos\theta\ket{0} + \sin\theta\ket{1})^{\otimes 2^n}$ for $\sin^2\theta \approx \ell/2^n$ simply cannot encode this dependency. We needed a new formalism.

After much frustration with indices and combinatorial cases (often arising from identical versus distinct indices in sums), we found relief in the bosonic perspective. The key insight was to treat the oracle’s purification as a system of indistinguishable bosons whose spatial locations encode membership in $S$. This eliminated the combinatorial clutter of index management and replaced it with clean algebraic manipulations via creation and annihilation operators. 

Within this framework, we were able to reinterpret and re-derive earlier results—such as those of Hamoudi and Magniez~\cite{hamoudi2023quantum}—in the language of bosons, gaining both conceptual clarity and technical leverage. Moreover, the physical picture aligned beautifully with the intended semantics of the quantum witness: the prover’s ideal witness is $\ket{S}$, the uniform superposition over elements of $S$. In the bosonic purification, $S$ directly corresponds to the positions of the bosons, so the verifier’s measurement effectively “grabs, in superposition,” a uniformly random boson and records its position. This unifies the intuition that a quantum witness can represent the entire set $S$ at once, while a classical witness must enumerate its elements.

\paragraph{Challenges within the bosonic formulation.}
Despite its elegance, the bosonic formalism introduced its own difficulties. Our initial hope was that the oracle queries would assume a particularly simple and structured form in the momentum basis. Early calculations suggested that queries to $U$ could be expressed via the \emph{double momentum hopping operator} $\tilde{\H}_y$, which preserves total momentum. By analogy with Noether’s theorem~\cite{Noether1918}, we could then argue that since the total momentum remains conserved (and initially zero), the success probability of the first “guess” at a boson’s position should equal $\ell / 2^n$. This was encouraging—it suggested an inductive argument might succeed.

However, post-selection proved fatal to this line of reasoning. Conditioning on a successful guess perturbs the total momentum, breaking the conservation structure and rendering straightforward induction impossible. We therefore needed a more delicate understanding of how the oracle queries affect the system’s momentum distribution.

A central obstacle was the definition of $U$ itself: each element $y$ must be included with probability proportional to $\gamma_y^{(S)}$, the squared Hadamard amplitude of $\ket{S}$. 
A first question that needed resolving was why $\gamma_y^{(S)}$ was the correct parameter to use for sampling $U$. In particular, as previously noted, using the ``raw'' Hadamard amplitudes of $\ket{S}$ leaks the sign information in the Fourier basis, which is the crucial information for approximately synthesizing $\ket{S}$ in the Fourier basis. We suspect that our choice of $\gamma_y^{(S)}$ isn't unique, but we do not have a second example. $\gamma_y^{(S)}$ gave us enough trouble as it is.
While $\gamma_y^{(S)}$ behaves roughly like the square of a Gaussian random variable, its unbounded support complicates any attempt to interpret it as a probability in $[0,1]$. We explored several approaches—thresholding, randomized truncations, and fully quantum purifications—but most were unsatisfactory. Truncating isn't smooth, which yields difficulties in analysis, as we only knew of techniques for handling polynomials in the hopping operators. Another option is conditioning the distribution $\mathsf{Strong}$ on the $\gamma_y^{(S)}$ being bounded by say $n^{100}$. Chernoff bounds argue that this event occurs with probability only exponentially small. Unfortunately, conditioning the distribution on this event changes the initial state from a state of $\ell$ bosons in 0-momentum to a state exponentially close in additive error. However, as previously discussed, additive approximations are not strong enough to prove exponentially small probability upper bounds as the errors compound.

Ultimately, we adopted an exponential function, guided by the “flat-tail” polynomial approximations introduced by Narayanan~\cite{narayanan2024improved}. This choice was technically crucial, as it allowed us to approximate the exponential behavior using low-degree polynomials in the double momentum hopping operators.
With this analytic machinery in hand, we could begin to generalize beyond the specific oracle construction. Standard intuition from quantum uncertainty principles suggests that any state supported on a few momentum modes should yield a sampling upper bound. Yet, since the algorithm can query all $y$ simultaneously, this intuition breaks down. The key insight was that the double momentum hopping operator typically acts by moving a \emph{pair} of bosons, especially when starting from a condensate. 
A useful way to reinterpret this phenomenon is through momentum conservation. 
Each oracle query preserves total momentum, which—by analogy to translational invariance in many-body systems—acts as a global symmetry constraining how amplitude can flow between momentum modes. 
In the initial condensate, the total momentum is zero, and so any admissible process must preserve this invariance. 
This observation already implies a tight bound on the probability of sampling a single correct element from~$S$, as any non-zero-momentum component must arise from a compensating momentum elsewhere in the system. 
Extending this reasoning to many guesses required recognizing that the double-hopping operator effectively creates and annihilates \emph{pairs} of equal-and-opposite momenta—``Cooper-like'' pairs in condensed-matter language—so that quasi-even condensates retain their global symmetry across queries. In this case, since the Fourier transform is considered, a ``Cooper-like'' pair is a pair of bosons in the same mode.
It is precisely this paired structure, rather than a literal recording of queries in~$U$, that allowed us to control the growth of odd momentum modes and prove the quasi-even property.

Finally, establishing that the quasi-even condensate property persists after a polynomial number of queries required substantial new machinery. Our bounds are almost certainly suboptimal, but they suffice to complete the proof. To our knowledge, this represents one of the first quantum query lower bounds addressing an oracle with as much internal quantum structure as $U$. Consequently, the classical tools of quantum query complexity—those tailored to unstructured oracles such as Grover’s search or collision-finding—were insufficient, necessitating the new bosonic framework developed here.

\paragraph{Connections to set-size estimation and average-case $\mathsf{NP}$}

In \Cref{sec:strong_yes}, we define the $\mathsf{Strong}$ distribution, a distribution over ``yes'' instances of spectral Forrelation, and implicitly define ``no'' instances of spectral Forrelation to have $S$ with size $\leq \ell / 100$.  In addition to solving the spectral Forrelation problem, an algorithm that could distinguish between sets of size $\ell$ and size $\ell / 100$ would be able to distinguish between the particular yes and no instances we use to separate $\mathsf{QMA}$ from $\mathsf{QCMA}$.  This puts our problem in line with the previous attempts of Fefferman and Kimmel~\cite{fefferman2015quantum} and Natarajan and Nirkhe~\cite{natarajan2024distribution}, who more directly used the problem of subset size checking to separate $\mathsf{QMA}$ from $\mathsf{QCMA}$ relative to non-standard oracles.  However, \cite{fefferman2015quantum} notes that the problem of subset size checking is in $\mathsf{AM}$ due to Stockmeyer counting\footnote{Roughly speaking, in a Stockmeyer approximate counting protocol~\cite{10.1145/800061.808740}, Arthur's public key message is interpreted as the key $k$ for a hash function $h: \bits^n \to \{0,\ldots, \ell-1\}$ and Merlin is tasked with answering with $x \in S$ such that $h(x) = 0$. Merlin will be able to answer with probability $\approx 1 - \inv{\e} \ge 0.63$ when $|S| \approx \ell$ and will only be able to answer with probability $1 - \e^{-1/100} \le 0.01$ when $|S| \le \ell/100$.

As the hard instances we construct are based on set-size estimation, the $\AM$ containment holds for our construction. (Meanwhile, we do not know (nor believe) that the generic problem of spectral Forrelation is in $\AM$ as not all spectral Forrelation instances are variants of set-size estimation.) However, since the problem is in $\AM$, the problem is also in average-case $\NP$, by simply amplifying the $\AM$ protocol and then fixing the choice of randomness of Arthur. So, there will exist no instances that are falsely accepted; however, this will only occur for an $\exp(-n)$ small fraction of the no instances.}.  This further implies that the problem of set size checking (and the related problem of distinguishing between our choice of yes and no instances) is in average-case $\mathsf{NP}$, because by fixing the randomness of Stockmeyer counting to some optimal string, a deterministic polynomial-time classical verifier distinguishes between small and large sets for most small and large sets.  

On the surface, this seems to imply a problem with our proof technique, since \Cref{thm:big-upper-bound-thm} is an average-case statement over the $\mathsf{Strong}$ distribution, and there is an average-case $\mathsf{NP}$ verifier for that distribution.  We note that the prior works do not run into this problem as they directly apply the adversary method and diagonalization to the problem of subset size checking (given an in-place permutation oracle), and thus do not talk about a fixed distribution over yes and no instances of the subset size checking problem.  

The resolution for this paper is the following subtle point: while ~\Cref{thm:big-upper-bound-thm} is an average case statement, the construction of the sampler from \Cref{thm:formal-removal-of-S} requires that the $\QCMA$ distinguishing algorithm is able to distinguish between \emph{all} yes and no instances.  Even if the average-case $\mathsf{NP}$ verifier correctly outputs ``no'' on a $1 - 2^{-n}$ fraction of small sets, this will not be enough to get a successful sampler.  To see this, we note that the successful sampler that we get from a $\mathsf{QCMA}$ verifier only samples $v$ points with probability $\left(\frac{1}{36t^2}\right)^{v}$ -- an incredibly small probability. This is consistent with the verifier only outputting $v$ many points from that fixed small set $\Delta$ that is misclassified by the average-case $\mathsf{NP}$ verifier. The fact that our reduction to a sampler \textit{requires} a worst-case distinguishing algorithm for our version of the set size checking problem (where the verifier is additionally given access to $\mathsf{U}$) is the reason we avoid this problem, but also requires us to use a more elaborate diagonalization argument (than that of Aaronson and Kuperberg~\cite{aaronson2007quantum}) in to prove \Cref{thm:main}, as our property-testing contradiction only proves the existence of one misclassified instance per input size $n$.

\ifstoc
\else
\section{Acknowledgments}

We thank James Bartusek, Andrea Coladangelo, Uma Girish, Andrew Huang, Jonas Helsen, William Kretschmer, Anand Natarajan, Barak Nehoran, Fermi Ma, Er-cheng Tang, Umesh Vazirani, and Henry Yuen for helpful discussions. We additionally thank Fermi Ma and Andrew Huang for suggesting comments on an early draft of the result that greatly improved the presentation. We thank Nicholas Kocurek and Joe Slote for assistance in creating figures. 

JB is supported by Henry Yuen's AFORS award FA9550-21-1-036 and NSF CAREER award CCF2144219. JH acknowledges funding from the Harvard Quantum Initiative postdoctoral fellowship. This work was partially completed while all authors were participants in the \textit{Simons Institute for the Theory of Computing}~Summer Clustering on Quantum Computing.
\fi

\bibliographystyle{alpha}
\bibliography{supplementary-files/refs}

\appendix

\end{document}